\documentclass{article}
\usepackage{authblk}
\usepackage[utf8]{inputenc}
\usepackage[T1]{fontenc}
\usepackage[a4paper]{geometry}
\usepackage{graphicx}
\usepackage[textwidth=8em,textsize=small]{todonotes}
\usepackage{amsmath,amsthm,amssymb}
\usepackage{natbib}
\usepackage{hyperref}
\usepackage[nameinlink]{cleveref}
\usepackage{subcaption}
\usepackage{bbm}
\usepackage{amsfonts}
\usepackage{mathtools}
\usepackage[titletoc,toc]{appendix}
\usepackage{url}
\usepackage{algorithm}
\usepackage{algpseudocode}
\usepackage{multirow}
\usepackage{booktabs}
\usepackage[normalem]{ulem}
\usepackage{threeparttable}
\usepackage{enumitem}
\usepackage{array}
\usepackage{setspace}
\usepackage{orcidlink}

\definecolor{darkblue}{HTML}{0c7dbb}
\definecolor{darkgreen}{rgb}{0,0.48,0.65}
\hypersetup{
    colorlinks = true,
    citecolor=darkgreen,
    filecolor=black,
    linkcolor=darkblue,
    urlcolor=blue
}
\providecommand{\keywords}[1]
{
    \small
    \textbf{\textit{Keywords---}} #1
}

\DeclareMathOperator*{\argmin}{arg\,min}

\newtheorem{theorem}{Theorem}[section]
\newtheorem{corollary}{Corollary}[section]

\newtheorem{remark}{Remark}
\newtheorem{lemma}[theorem]{Lemma}   
\newtheorem{proposition}[theorem]{Proposition}
\newtheorem{definition}[theorem]{Definition}
\Crefname{algocf}{Algorithm}{Algorithms}
\Crefname{lemma}{Lemma}{Lemmas}
\Crefname{proposition}{Proposition}{Propositions}
\Crefname{definition}{Definition}{Definitions}
\title{Nonparametric Correlation Estimator via Solving Cubic Equations and its Application to Brain Functional Connectivity Analysis}
\author[1]{Shenyuan Yang}
\author[2]{Gary Green}
\author[3]{Jian Zhang}
\author[2]{André Gouws}
\author[1]{Jie Li\orcidlink{0000-0001-8353-1322}\thanks{Corresponding author: Jie Li, Lijiang Culture and Tourism College.\ \textbf{Email}:~\href{jie.li@ljwhlyxy.edu.cn}{jie.li@ljwhlyxy.edu.cn}}}
\affil[1]{Lijiang Culture and Tourism College, Lijiang, Yunnan, China, 674199}
\affil[2]{York Neuroimaging Centre, University of York, Innovation Way, York, YO10 5NY, UK}
\affil[3]{School of Mathematics, Statistics and Actuarial Science, University of Kent, Canterbury, CT2 7NF, UK}

\begin{document}
\maketitle
\begin{abstract}
    In this paper, we propose a novel nonparametric estimator for the correlation coefficient that is based on solving a cubic equation. This approach allows for the estimation of time-varying correlation coefficients in a nonparametric framework, providing flexibility in capturing complex relationships between variables. Furthermore, we adopt the local linear  smoothing technique to correct the boundary effects, which are common in nonparametric estimation.
    We establish the theoretical properties of the proposed estimator, including consistency and asymptotic normality, and demonstrate its performance through simulation studies. Additionally, we apply our method to analyse dynamic functional connectivity in brain networks under six frequency bands, highlighting its potential for uncovering insights into neural interactions and cognitive processes.
\end{abstract}
\section{Introduction}\label{sec:introduction}
Human brain activity exhibits temporal dynamics, with the coupling between brain regions changing across multiple timescales in response to cognitive tasks, brain states, and disease progression~\citep{woolrichDynamicStateAllocation2013,favarettoSubcorticalcorticalDynamicalStates2022,jinDynamicFunctionalConnectivity2023}. Functional connectivity (FC), which characterises the statistical dependence between regional brain signals, has become an important framework for investigating the organisation of brain function, disease-related network abnormalities, and potential biomarkers~\citep{hippLargescaleCorticalCorrelation2012,oneillMeasuringElectrophysiologicalConnectivity2015,brookesInvestigatingElectrophysiologicalBasis2011}. Unlike fMRI, which relies on haemodynamic signals, MEG directly measures the magnetic fields generated by neuronal currents and provides millisecond-level temporal resolution for tracking transient changes in oscillatory phase and amplitude~\citep{hamalainenRealisticConductivityGeometry1989,gramfortMEGEEGData2013,brunsFourierHilbertWaveletbased2004}. It therefore offers a particularly fine-grained temporal basis for investigating dynamic functional connectivity (dFC), beyond what can be captured by static measures based on time-averaged connectivity~\citep{gramfortMEGEEGData2013,brunsFourierHilbertWaveletbased2004}. Consequently, recent research has increasingly shifted from static functional connectivity towards dFC, with growing interest in characterising time-varying patterns of brain networks, transitions between connectivity states, and their associations with cognitive processes and clinical phenotypes~\citep{coelliTimevaryingBrainComprehensive2025,woolrichDynamicStateAllocation2013,brookesInvestigatingElectrophysiologicalBasis2011}.

In neuropsychiatric and neurodegenerative disorders, alterations in brain network organisation may precede or accompany the emergence of clinical symptoms. In Alzheimer's disease (AD), for example, MEG studies have reported detectable alterations in the complexity of dynamic brain networks, transitions between connectivity states, and frequency-specific coupling, with some of these features showing potential for distinguishing patients from healthy controls \citep{jinDynamicFunctionalConnectivity2023,carrasco-gomezDynamicsBrainConnectivity2026}. These findings suggest that dFC may provide complementary information for early detection, disease stratification, and the assessment of disease progression. Consequently, current research into dynamic functional connectivity is broadly categorised into two complementary strands:  fMRI-based approaches and MEG-based approaches, which are discussed in the sections that follow.


\textbf{Dynamic functional connectivity based on fMRI.} Owing to its relatively high spatial resolution and the availability of extensive publicly accessible datasets, fMRI has been the predominant data source for dFC research~\citep{fischlFreeSurfer2012,schaeferLocalGlobalParcellationHuman2018}. Over the past five years, fMRI-based dFC studies have expanded beyond static, time-averaged connectivity towards the characterisation of time-varying organisation and disease-related network abnormalities~\citep{wenSpatiotemporalDynamicFunctional2025,zhangCLASSIFFICATIONMILDCOGNITIVE2025,kajimuraFrequencyspecificBrainNetwork2023}. Methodologically, the field has evolved from conventional sliding-window and clustering approaches towards spatiotemporal modelling and representation learning of dynamic connectivity states \citep{wenSpatiotemporalDynamicFunctional2025,zhangCLASSIFFICATIONMILDCOGNITIVE2025}. Graph neural networks, temporal Transformers, and contrastive learning have further enhanced the capacity for end-to-end representation learning \citep{chenDFCExpertLearningDynamic2025,huangBrainATCLAdaptiveTemporal2025}. However, systematic evaluations have shown that increasingly complex models do not consistently outperform simpler approaches \citep{hanRethinkingFunctionalBrain2026,dingMachineLearningDynamic2026}. Consequently, attention in this field is increasingly shifting away from classification accuracy alone towards the stability, interpretability, and cross-dataset reproducibility of dynamic abnormality patterns.

\textbf{Dynamic functional connectivity based on MEG.} With millisecond-level temporal resolution, MEG is particularly well suited to investigating rapid neuronal oscillations and interactions between spatially distributed brain regions~\citep{gramfortMEGEEGData2013,linSpectralSpatiotemporalImaging2004,brookesInvestigatingElectrophysiologicalBasis2011}. In recent years, MEG-based functional connectivity research has progressed from analyses of single-frequency bands or static group differences towards the joint characterisation of multiple frequency bands, time-varying connectivity states, and the temporal complexity of network dynamics \citep{coelliTimevaryingBrainComprehensive2025}. Within the Alzheimer's disease spectrum, previous studies have reported an overall reduction or reorganisation of MEG-derived dFC, with differences particularly involving the \(\alpha\) and \(\beta\) bands, frontotemporal regions, and the default mode network, as well as associations with cognitive function and structural alterations \citep{carrasco-gomezDynamicsBrainConnectivity2026}. Measures of dynamic network complexity, state-transition characteristics, and multi-band deep-learning features have also shown potential for patient classification and early prediction of clinical progression \citep{jinDynamicFunctionalConnectivity2023}.

Nevertheless, existing MEG-dFC analyses predominantly rely on sliding-window approaches, epoch-based aggregation, amplitude-envelope correlations, phase synchronisation, or adaptive segmentation methods, and the resulting estimates remain sensitive to choices of window length and time-frequency smoothing \citep{coelliTimevaryingBrainComprehensive2025,jinDynamicFunctionalConnectivity2023}. To improve the estimation of direct connectivity and the characterisation of information flow, recent work has incorporated structural connectivity priors, partial coherence analysis, and time-lagged multidimensional pattern connectivity approaches \citep{rahimiTimeLaggedMultidimensionalPattern2023}. Overall, MEG-dFC research is moving beyond demonstrating the presence of connectivity abnormalities towards the robust extraction, statistical interpretation, and independent validation of dynamic abnormality patterns. However, statistical inference frameworks that fully exploit the high temporal resolution of MEG data remain underdeveloped.

Nonparametric estimation of time-varying correlation coefficients is a key methodological problem in the study of dynamic characteristics of brain networks. Existing approaches include sliding-window Pearson correlation, kernel-based Nadaraya--Watson estimation, local polynomial regression models \citep{fanLocalPolynomialKernel1995,fanDataDrivenBandwidthSelection1995,fanEstimationConditionalDensities1996,fanEfficientEstimationConditional1998, staniswalisNonparametricRegressionAnalysis1998}. Although these methods provide desirable theoretical properties, such as consistency and asymptotic normality of the resulting estimators, several practical challenges remain. First, boundary effects can be substantial, particularly in regions where observations are sparse or unevenly distributed in the high-dimensional case. Second, the use of a single globally optimised bandwidth can be disproportionately influenced by sparse regions of the data, potentially leading to increased estimation error in other regions~\citep{liStatisticalInferenceHighdimensional2021}. Third, although correlation-matrix-based approaches can preserve the required positive semidefinite structure, ensuring this property in high-dimensional or large-sample settings can introduce substantial computational costs, making computational efficiency a major bottleneck.

Motivated by these limitations, we propose a new nonparametric approach for estimating time-varying correlation coefficients. The proposed method combines kernel-weighted local maximum likelihood estimation (LMLE) with a cubic-equation-based procedure for obtaining the local correlation function. Rather than estimating a single global bandwidth for the entire correlation matrix, we independently determine an optimal bandwidth for each pair of variables. The resulting pairwise correlation estimates are then assembled to form the time-varying correlation matrix. Since independently estimated pairwise correlations do not necessarily yield a valid correlation matrix, we subsequently apply a matrix correction procedure that projects the estimated matrix onto the nearest valid correlation matrix~\citep{highamComputingNearestCorrelation2002}. The resulting positive semidefinite correlation matrices are then used for subsequent dynamic functional connectivity analyses.

The proposed estimator retains desirable asymptotic properties, including consistency and asymptotic normality, while allowing the bandwidth to adapt to the local characteristics of each variable pair. This pair-specific bandwidth selection provides greater flexibility than a single globally optimised bandwidth and is particularly advantageous for high-dimensional applications in which the underlying dependence structure may vary substantially across variable pairs.

The rest of this paper is organised as follows. In Section~\ref{sec:Methodology}, we introduce the proposed nonparametric estimation method for time-varying correlation coefficients. Section~\ref{sec:simulation_study} presents simulation studies to evaluate the performance of the proposed method. In Section~\ref{sec:real_data_analysis}, we apply the method to MEG data for dynamic functional connectivity analysis. Finally, Section~\ref{sec:conclusion} concludes the paper and discusses future research directions.

\section{Methodology}\label{sec:Methodology}

Let \((U_{i},\mathbf{x}_{i})\), \( \mathbf{x}_{i}= (x_{i1},x_{i2})^{\top}\), \(i=1,\dots,n\), be i.i.d.\ observations, where \(U_{i}\in[-1,1]\) has density \(f_{U}\), and conditional on \(U=u\), we assume the \( \mathbf{X}=(X_{1}, X_{2})^{\top} \) follows a bivariate normal distribution with mean zero and covariance matrix \(\Sigma(\rho(u))\):
\begin{equation*}
    \mathbf{X}=(X_{1},X_{2})^{\top} \mid (U=u)\sim N\!\left(0,\ \Sigma(\rho(u))\right),
    \qquad
    \Sigma(\rho(u))=
    \begin{pmatrix}1&\rho(u)\\ \rho(u)&1
    \end{pmatrix},\ \rho(u)\in(-1,1).
\end{equation*}
Without loss of generality, we assume that \( \Sigma(\rho(u)) \) is a correlation matrix. If \( \Sigma(\rho(u)) \) is not a correlation matrix, we can always standardise \( X_{1}, X_{2} \) to have unit variance. The goal is to estimate the correlation function \(\rho(u)\) nonparametrically.
Define the negative Gaussian log-likelihood contribution (up to irrelevant constants)
\begin{equation*}
    \ell(\mathbf{x};\rho)=\frac{x_{1}^{2}+x_{2}^{2}-2\rho x_{1}x_{2}}{1-\rho^{2}}+\log(1-\rho^{2}),\qquad \rho\in(-1,1).
\end{equation*}
The kernel-weighted local criterion at \(u_{0}\) is
\begin{equation*}
    Q_{n}(\rho;u_{0})=\frac{1}{n}\sum_{i=1}^{n} K_{h}(U_{i}-u_{0})\,\ell(\mathbf{x}_{i};\rho),
    \qquad K_{h}(t)=\frac{1}{h}K(t/h),
\end{equation*}
where \( K(\cdot) \) is the kernel function. Define the local weighted MLE as
\begin{equation*}
    \hat{\rho}(u_{0})\in \argmin_{\rho\in(-1,1)} Q_{n}(\rho;u_{0}).
\end{equation*}
One can differentiate \(Q_{n}\) with respect to \(\rho\), set it to zero, and simplify to get the cubic score equation as follows:
\begin{equation}\label{eq:cubic_equation_2}
    \rho^{3} - B_{n}(u_{0})\rho^{2}+\big(A_{n}(u_{0})-1\big)\rho-B_{n}(u_{0})=0,
\end{equation}
where \(A_{n},B_{n}\) are the kernel-weighted quantities:
\begin{equation*}
    A_{n}(u_{0})=\frac{\sum_{i=1}^{n} (x_{i1}^{2}+x_{i2}^{2})K_{h}(U_{i}-u_{0})}{\sum_{i=1}^{n} K_{h}(U_{i}-u_{0})},
    \quad
    B_{n}(u_{0})=\frac{\sum_{i=1}^{n} x_{i1}x_{i2}K_{h}(U_{i}-u_{0})}{\sum_{i=1}^{n} K_{h}(U_{i}-u_{0})}.
\end{equation*}
Coincidentally, the Nadaraya--Watson estimator of the correlation coefficient is exactly the ratio of these two quantities:
\begin{equation}\label{eq:correlation_coefficient_nonparametric_2}
    \hat{\rho}^{\mathrm{NW}}(u_{0})=\frac{2\sum_{i=1}^{n}x_{i1}x_{i2}K_{h}(u_{i}-u_{0})}{\sum_{i=1}^{n}(x_{i1}^{2}+x_{i2}^{2})K_{h}(u_{i}-u_{0})}=\frac{2B_{n}(u_{0})}{A_{n}(u_{0})}.
\end{equation}

The corresponding population quantities are \(A(u_{0})=\mathbb{E}\left[ X_{1}^{2}+X_{2}^{2}\mid U=u_{0} \right]\), \(B(u_{0})=\mathbb{E}\left[ X_{1}X_{2}\mid U=u_{0} \right]\). Because \(\Sigma(\rho(u))\) is a correlation matrix for every \(u\), we have
\begin{equation}\label{eq:an_bn}
    A(u)\equiv 2 \quad\text{for all } u,\qquad B(u_{0})=\rho(u_{0}).
\end{equation}
Here \( u_{0}  \) can be regarded as a fixed point in the design interval while \( u \) is a generic point in the design interval. The fact that \(A(\cdot)\) is exactly constant drives almost everything in the theory of consistency and limit distribution. Especially, if \( A_{n}(u_{0})=2 \), one real root of~\eqref{eq:cubic_equation_2} is equal to \( B_{n}(u_{0}) \). For an estimator of the correlation coefficient,~\citet{kendallInferenceRelationship1973} pointed out that there is at least one real root of equation~\eqref{eq:cubic_equation_2} lying in the interval \( [-1,1] \). Let \( \hat{\rho}(u_{0}) \) be the real root of cubic equation~\eqref{eq:cubic_equation_2} and minimises \( Q_{n}(\rho;u_{0}) \). Next, we will develop the theory of uniform consistency of the estimator \( \hat{\rho}(u_{0}) \) over compact subsets of the design interval \( (-1,1) \), limit distribution, optimal bandwidth selection and the convergence rate of the estimator \( \hat{\rho}(u_{0}) \). Besides, we also discuss the boundary behaviour of the estimator \( \hat{\rho}(u_{0}) \) and propose a method to correct the boundary effects.


\subsection{Assumptions}\label{subsec:Assumptions}
\begin{enumerate}[label=(A\arabic*)]
    \item \textbf{Kernel}. \( K \) is a symmetric, bounded, and Lipschitz continuous probability density function with compact support \( [-1,1] \), \( K\geqq 0 \), \( \mu_{1}(K)=\int tK(t)\,\mathrm{d}t=0\), \( \mu_{2}(K)=\int t^{2}K(t)\,\mathrm{d}t <\infty\), \( R(K)=\int K^{2}(t)\,\mathrm{d}t<\infty \).
    \item \textbf{Design Density}. The density \( f_{U} \) of \( U \) is bounded away from zero and infinity on \( [-1,1] \), i.e., \( 0<f_{\min}\leqq f_{U}(u)\leqq f_{\max}<\infty \). In the interior of \( [-1,1] \), \( f_{U} \) is twice continuously differentiable with bounded second derivative; at \( \pm 1 \), \( f_{U} \) has finite one-sided derivatives.
    \item \textbf{Smoothness of Correlation Coefficient}.\ \( \rho\in C^{2}[-1,1] \), with \( \sup_{u}{\left\lvert \rho^{\prime}(u) \right\rvert} \), \( \sup_{u}{\left\lvert \rho^{\prime\prime}(u) \right\rvert} <\infty \) and there is \( \eta > 0 \) with \( \sup_{u}{\left\lvert \rho^{\prime}(u) \right\rvert} \leqq 1-\eta\), which means that the true curve stays a fixed distance from the parameters' boundary \( \pm 1 \).
    \item \textbf{Kernel Bandwidth}. The kernel bandwidth \( h \) satisfies \( h\to 0 \) and \( nh\to \infty \) as \( n\to \infty \); for the uniform statements, we require \( nh/\log(1/h)\to \infty \).
    \item \textbf{Moment Condition}. Since \( \mathbf{X}\mid U=u \) is Gaussian with correlation bounded away from \( \pm 1 \) by (A3), all moments of \( X_{1}^{2}+X_{1}^{2},X_{1}X_{2} \) exist and are continuous, hence bounded, uniformly in \( u\in[-1,1] \); in  particular, \( \mathbb{E}\left[ (X_{1}^{2}+X_{2}^{2})^{4}\mid U=u \right] \) is  bounded  uniformly.
    \item \( (U_{i}, \mathbf{x}_{i})^{n}_{i=1} \) are i.i.d.\ observations.
\end{enumerate}
These assumptions are standard in the literature of nonparametric regression and kernel smoothing, ensuring the well-definedness of the estimation problem and the validity of asymptotic results; see~\citet{fanDataDrivenBandwidthSelection1995,fanLocalPolynomialKernel1995,fanEstimationConditionalDensities1996,fanEfficientEstimationConditional1998,staniswalisNonparametricRegressionAnalysis1998}.
\subsection{Theory}\label{subsec:theory}
\begin{theorem}\label{thm:uniform_consistency_an_bn}
    Under the assumptions (A1)-(A6), for any \(\delta>0\), we have
    \begin{equation*}
        \sup_{u_{0}\in[-1+\delta,1-\delta]}| A_{n}(u_{0})-2|  \xrightarrow{p}0,\quad \sup_{u_{0}\in[-1+\delta,1-\delta]}| B_{n}(u_{0})-\rho(u_{0})|  \xrightarrow{p}0.
    \end{equation*}
\end{theorem}
\begin{proof}
    Immediate from~\Cref{lem:uniform-consistency} together with~\Cref{eq:an_bn} (\(A(u)\equiv 2\), \(B(u_{0})=\rho(u_{0})\)), and the assumption (A4) which guarantees \(\sqrt{\log(1/h)/(nh)}+h^{2}\to 0\).
\end{proof}
\Cref{thm:uniform_consistency_an_bn} upgrades pointwise consistency of the coefficients of the cubic equation to uniform convergence over compact subsets of the design interval. This stronger form of convergence ensures that the root-continuity argument holds uniformly rather than just pointwise, in \(u_{0}\).

For finite-sample, the cubic equation may have 3 real roots, or 1 real root and 2 complex roots. The following proposition shows that the cubic equation has at least one real root in \((-1,1)\) for all \(u_{0}\) in compact subsets of the design interval.

\begin{proposition}\label{pro:finite-sample}
    For every \(n\) and every \(u_{0}\) with \(\sum_{i}w_{i}(u_{0})>0\), where \( w_{i}(u_{0})=K_{h}(U_{i}-u_{0}) \), the empirical cubic \(g_{n}(\rho;u_{0})\coloneqq\rho^{3}-B_{n}(u_{0})\rho^{2}+(A_{n}(u_{0})-1)\rho-B_{n}(u_{0})=0\) has at least one real root in \([-1,1]\).
\end{proposition}

\begin{proof}
    Because \(w_{i}\geqq 0\), the elementary inequalities \(2X_{i1}X_{i2}\leq X_{i1}^{2}+X_{i2}^{2}\) and \(-2X_{i1}X_{i2}\leq X_{i1}^{2}+X_{i2}^{2}\) (i.e. \((X_{i1}\pm X_{i2})^{2}\geq 0\)) are preserved under nonnegative weighted averaging, giving the exact, deterministic bound
    \begin{equation}\label{eq:bounds_An_Bn}
        |B_{n}(u_{0})|\leq \frac{1}{2} A_{n}(u_{0}),\qquad A_{n}(u_{0})\geq 0.
    \end{equation}

    Evaluate \(g_{n}(\rho;u_{0})\) at \(\rho=\pm1\):
    \begin{equation*}
        g_{n}(1;u_{0})=A_{n}(u_{0})-2B_{n}(u_{0}),\qquad g_{n}(-1;u_{0})=-A_{n}(u_{0})-2B_{n}(u_{0}).
    \end{equation*}
    By~\eqref{eq:bounds_An_Bn}, \(B_{n}(u_{0})\leq A_{n}(u_{0})/2\Rightarrow g_{n}(1;u_{0})\geq 0\), and \(B_{n}(u_{0})\geq-A_{n}(u_{0})/2\Rightarrow g_{n}(-1;u_{0})\leq 0\). Since \(g_{n}(\cdot\,;u_{0})\) is continuous, the Intermediate Value Theorem yields a root in \([-1,1]\).
\end{proof}

This establishes the empirical, finite-sample analogue of the classical population-level result by~\citet{kendallInferenceRelationship1973}. Crucially, our proof demonstrates that this relationship holds exactly for any non-negative kernel weights, without relying on asymptotic approximations.

\begin{theorem}\label{thm:uniform_consistency}
    \textbf{Uniform Consistency of \(\hat{\rho}(u_{0})\)}. Under Assumptions (A1)--(A6), for any \(\delta>0\) and compact set \(\mathcal{U}_{\delta}=[-1+\delta,1-\delta]\), we have
    \begin{equation*}
        \sup_{u_{0}\in\mathcal{U}_{\delta}}\big|\hat{\rho}(u_{0})-\rho(u_{0})\big|\xrightarrow{p}0,\qquad n\to\infty.
    \end{equation*}
\end{theorem}
\Cref{thm:uniform_consistency} establishes that the estimator \(\hat{\rho}(u_{0})\) converges uniformly to the true correlation function \(\rho(u_{0})\) over compact subsets of the design interval. The main idea of proof is to first  show  that the population cubic equation has a unique simple root \( \rho_{0} \), then use the uniform convergence of the empirical coefficients to the population coefficients (\Cref{thm:uniform_consistency_an_bn})  and the continuity of roots of polynomials, i.e., Rouché's theorem,~\citep{ahlforsComplexAnalysisIntroduction2007}, to conclude that the empirical root \( \hat{\rho}(u_{0}) \) converges uniformly to \( \rho(u_{0}) \), the proof details of~\Cref{thm:uniform_consistency} is in Appendix~\ref{subsec:proof}. Having shown that \( \hat{\rho}(u_{0}) \) is uniformly consistent, we next examine its limiting distribution. The derivation is carried out separately for interior points and for boundary points of the design interval, reflecting the fact that the asymptotic behaviour differs between these two regions.

\begin{theorem}\label{thm:interior_asymptotic_normality}
    \textbf{Interior Asymptotic Normality}. Under Assumptions (A1)--(A6), for any fixed interior point \(u_{0}\in(-1,1)\) with \(h\to0\) and \(nh\to\infty\), we have
    \begin{equation}\label{eq:interior_limit}
        \sqrt{nh}\Big(\hat{\rho}(u_{0})-\rho(u_{0})-\mathrm{Bias}_\rho(u_{0})\Big)\xrightarrow{d}N\!\left(0,\ \frac{R(K)}{f_{U}(u_{0})}\cdot\frac{\big(1-\rho(u_{0})^{2}\big)^{2}}{1+\rho(u_{0})^{2}}\right),
    \end{equation}
    where the bias term is
    \begin{equation}\label{eq:interior_bias}
        \mathrm{Bias}_\rho(u_{0})=\frac{h^{2}\mu_{2}(K)}{2}\Big[\rho^{\prime\prime}(u_{0})+2\frac{f_{U}^{\prime}(u_{0})}{f_{U}(u_{0})}\rho^{\prime}(u_{0})\Big].
    \end{equation}
\end{theorem}

\begin{remark}
    The asymptotic variance in~\eqref{eq:interior_limit} has the form \(R(K)/\big(f_{U}(u_{0})I(\rho_{0})\big)\), where \(I(\rho)=(1+\rho^{2})/(1-\rho^{2})^{2}\) is the Fisher information for the correlation parameter \(\rho\) in a bivariate normal distribution with known unit marginal variances. This coincides with the general local-likelihood theory for nonparametric varying-coefficient models~\citep{fanLocalPolynomialKernel1995,staniswalisNonparametricRegressionAnalysis1998}, providing an independent verification of the calculations.
\end{remark}

\begin{theorem}\label{thm:boundary_asymptotic_normality}
    \textbf{Boundary Asymptotic Normality}. Under Assumptions (A1)--(A6), let \(u_{0}=-1+ch\) for fixed \(c\in[0,\infty)\) (left boundary), with \(h\to0\) and \(nh\to\infty\). Define the truncated kernel moments
    \begin{equation}\label{eq:truncated_moments_0}
        \mu_{j}(K,c)=\int_{-c}^{1} t^{j}K(t)\,\mathrm{d}t,\qquad R(K,c)=\int_{-c}^{1}K(t)^{2}\,\mathrm{d}t.
    \end{equation}
    Then
    \begin{equation}\label{eq:boundary_limit}
        \sqrt{nh}\left(\hat{\rho}(u_{0})-\rho(u_{0})-h\rho^{\prime}(u_{0})\frac{\mu_{1}(K,c)}{\mu_{0}(K,c)}\right)\xrightarrow{d}
        N\!\left(0,\ \frac{R(K,c)}{\mu_{0}(K,c)^{2}f_{U}(-1^+)}\cdot\frac{\big(1-\rho(u_{0})^{2}\big)^{2}}{1+\rho(u_{0})^{2}}\right).
    \end{equation}
    The corresponding result holds symmetrically for the right boundary \(u_{0}=1-ch\).
\end{theorem}

\begin{remark}
    The qualitative difference between~\Cref{thm:interior_asymptotic_normality} and~\Cref{thm:boundary_asymptotic_normality} is the order of bias: \(O(h^{2})\) in the interior versus \(O(h)\) at the boundary. This is a first-order boundary effect characteristic of local-constant (Nadaraya--Watson) type estimators. Higher-order kernels or local-linear modifications can reduce the boundary bias to \(O(h^{2})\).
\end{remark}

The asymptotic bias-variance trade-off in~\Cref{thm:interior_asymptotic_normality,thm:boundary_asymptotic_normality} determines the optimal bandwidth selection and the resulting convergence rates. The following corollary characterises these optimal rates.

\begin{corollary}\label{cor:optimal_bandwidth}
    \textbf{Optimal Bandwidth and Convergence Rates}. Under Assumptions (A1)--(A6), the optimal bandwidth and convergence rates differ between interior and boundary points as follows:
    \begin{enumerate}[label=(\roman*)]
        \item For interior points \(u_{0}\in(-1,1)\), the asymptotic mean squared error \(\mathrm{AMSE}(h)=\mathbb{E}[(\hat{\rho}(u_{0})-\rho(u_{0}))^{2}]\) is minimised at
            \begin{equation}\label{eq:opt_bandwidth_interior}
                h_{\mathrm{opt}}=\left(\frac{C_{2}(u_{0})}{4C_{1}(u_{0})\,n}\right)^{1/5}\asymp n^{-1/5},
            \end{equation}
            where
            \begin{align*}
                C_{1}(u_{0})& =\left(\frac{\mu_{2}(K)}{2}\right)^{2}\Big[\rho^{\prime\prime}(u_{0})+2\frac{f_{U}^{\prime}(u_{0})}{f_{U}(u_{0})}\rho^{\prime}(u_{0})\Big]^{2}, \\
                C_{2}(u_{0})& =\frac{R(K)}{f_{U}(u_{0})}\cdot\frac{\big(1-\rho(u_{0})^{2}\big)^{2}}{1+\rho(u_{0})^{2}}.
            \end{align*}
            At this bandwidth, \(\mathrm{AMSE}(h_{\mathrm{opt}})\asymp n^{-4/5}\) and \(\hat{\rho}(u_{0})-\rho(u_{0})=O_{p}(n^{-2/5})\).

        \item For boundary points \(u_{0}=-1+ch\) with fixed \(c\in[0,\infty)\), the optimal bandwidth is
            \begin{equation}\label{eq:opt_bandwidth_boundary}
                h_{\mathrm{opt}}=\left(\frac{C_{4}}{2C_{3}\,n}\right)^{1/3}\asymp n^{-1/3},
            \end{equation}
            where
            \begin{align*}
                C_{3}& =\Big[\rho^{\prime}(u_{0})\frac{\mu_{1}(K,c)}{\mu_{0}(K,c)}\Big]^{2},                                      \\
                C_{4}& =\frac{R(K,c)}{\mu_{0}(K,c)^{2}f_{U}(-1^+)}\cdot\frac{\big(1-\rho(u_{0})^{2}\big)^{2}}{1+\rho(u_{0})^{2}}.
            \end{align*}
            At this bandwidth, \(\mathrm{AMSE}(h_{\mathrm{opt}})\asymp n^{-2/3}\) and \(\hat{\rho}(u_{0})-\rho(u_{0})=O_{p}(n^{-1/3})\).
    \end{enumerate}
\end{corollary}

\begin{remark}
    The boundary convergence rate \(n^{-1/3}\) is strictly slower than the interior rate \(n^{-2/5}\). This is the standard penalty for local-constant smoothing at design boundaries, where the bias is \(O(h)\) rather than \(O(h^{2})\). This penalty can be eliminated by using local-linear likelihood methods, which achieve \(O(h^{2})\) bias uniformly across the domain, including at boundaries.
\end{remark}

Next, we address the boundary rate discrepancy. The suboptimal \(n^{-1/3}\) convergence rate at boundaries is not inherent to the cubic-equation framework itself, but rather a consequence of using the Nadaraya--Watson estimator for \(B_{n}(u_{0})\). We develop a boundary-corrected variant using a local linear method that restores the faster \(n^{-2/5}\) rate uniformly across the entire design interval, including boundary points.

The Nadaraya--Watson estimator \(B_{n}(u_{0})\) estimates \(\rho(u_{0})=\mathbb{E}[X_{1}X_{2}\mid U=u_{0}]\) using a truncated kernel whose first moment
\begin{equation}\label{eq:truncated_first_moment}
    \mu_{1}(K,c)=\int_{-c}^{1}t\,K(t)\,\mathrm{d}t
\end{equation}
is nonzero once the window \([u_{0}-h,u_{0}+h]\) is truncated by the boundary \(U\geq -1\). In contrast, \(A_{n}(u_{0})\) suffers no bias at any order, because the population function \(A(u)\equiv 2\) is constant; thus \(A_{n}(u_{0})\) estimates a zero derivative, and kernel asymmetry is irrelevant. Consequently, the boundary correction need only modify \(B_{n}(u_{0})\).

This is precisely the situation addressed by local-linear kernel regression~\citep{fanDesignadaptiveNonparametricRegression1992,ruppertMultivariateLocallyWeighted1994,fanLocalPolynomialModelling1996,fanEstimationConditionalDensities1996,fanEfficientEstimationConditional1998}. While the Nadaraya--Watson smoother exhibits \(O(h)\) design bias at boundaries, the local-linear smoother achieves \(O(h^{2})\) bias uniformly across the support, including boundaries. This automatic boundary correction is a well-known advantage of local polynomial methods. Because \(B(u_{0})=\rho(u_{0})\) is the regression function of \(X_{1}X_{2}\) on \(U\), we can replace \(B_{n}(u_{0})\) by its local-linear approximation in the cubic equation without altering the consistency and asymptotic normality arguments previously established.

We define the boundary-corrected estimator by replacing the Nadaraya--Watson component \(B_{n}(u_{0})\) with a local-linear fit, while retaining the local-constant estimator for \(A_{n}(u_{0})\).

\begin{definition}[Boundary-Corrected Estimator]\label{def:boundary_corrected_estimator}
    For \(u_{0}\in[-1,1]\), define \((\hat{\rho}_{0},\hat{\rho}_{1})\) as the minimiser of the local weighted least squares criterion
    \begin{equation}\label{eq:local_linear_wls}
        (\hat{\rho}_{0},\hat{\rho}_{1})=\arg\min_{a,b}\sum_{i=1}^{n} K_{h}(U_{i}-u_{0})\Big[x_{i1}x_{i2}-a-b(U_{i}-u_{0})\Big]^{2},
    \end{equation}
    and set \(B_{n}^{\mathrm{CL}}(u_{0})\coloneqq\hat{\rho}_{0}\).

    Explicitly, let
    \begin{align*}
        \hat{S}_{j}(u_{0})& =\frac{1}{n}\sum_{i=1}^{n} K_{h}(U_{i}-u_{0})(U_{i}-u_{0})^{j},                           \\
        \hat{T}_{j}(u_{0})& =\frac{1}{n}\sum_{i=1}^{n} K_{h}(U_{i}-u_{0})(U_{i}-u_{0})^{j}x_{i1}x_{i2},\quad j=0,1,2.
    \end{align*}
    Then
    \begin{equation}\label{eq:local_linear_explicit}
        B_{n}^{\mathrm{CL}}(u_{0})=\frac{\hat{S}_{2}(u_{0})\hat{T}_{0}(u_{0})-\hat{S}_{1}(u_{0})\hat{T}_{1}(u_{0})}{\hat{S}_{2}(u_{0})\hat{S}_{0}(u_{0})-\hat{S}_{1}(u_{0})^{2}}
        =\frac{1}{n}\sum_{i=1}^{n} \ell_{i}(u_{0})\,x_{i1}x_{i2},
    \end{equation}
    where
    \begin{equation}\label{eq:equivalent_kernel_weights}
        \ell_{i}(u_{0})=K_{h}(U_{i}-u_{0})\,\frac{\hat{S}_{2}(u_{0})-\hat{S}_{1}(u_{0})(U_{i}-u_{0})}{\hat{S}_{2}(u_{0})\hat{S}_{0}(u_{0})-\hat{S}_{1}(u_{0})^{2}}
    \end{equation}
    are the equivalent kernel weights. Unlike the Nadaraya--Watson weights \(K_{h}(U_{i}-u_{0})\geq0\), these weights can be negative near the boundary, which is the mechanism that eliminates the first-order bias.

    Define the boundary-corrected estimator \(\hat{\rho}^{\mathrm{CL}}(u_{0})\) as the unique real root in \((-1,1)\) of the modified cubic equation
    \begin{equation}\label{eq:modified_cubic}
        \rho^{3}-B_{n}^{\mathrm{CL}}(u_{0})\rho^{2}+\big(A_{n}(u_{0})-1\big)\rho-B_{n}^{\mathrm{CL}}(u_{0})=0,
    \end{equation}
    where \(A_{n}(u_{0})\) remains the local-constant estimator defined previously.
\end{definition}

The local-linear estimator \(B_{n}^{\mathrm{CL}}(u_{0})\) can equivalently be viewed as a local-constant estimator using a boundary-corrected kernel. At the population level, the weights in~\eqref{eq:local_linear_explicit} correspond to the boundary kernel
\begin{equation}\label{eq:boundary_kernel}
    K^{*}(t;c)\coloneqq\frac{S_{2}(c)-S_{1}(c)\,t}{S_{0}(c)S_{2}(c)-S_{1}(c)^{2}}\,K(t),\qquad t\in[-c,1],
\end{equation}
where
\begin{equation}\label{eq:truncated_moments}
    S_{j}(c)\coloneqq\mu_{j}(K,c)=\int_{-c}^{1}t^{j} K(t)\,\mathrm{d}t,\quad j=0,1,2,3,
\end{equation}
and \(c\) parameterizes the distance from the left boundary via \(u_{0}=-1+ch\). This construction is analogous to the Müller and Gasser--Müller boundary kernels~\citep{mullerSmoothOptimumKernel1991,gasserResidualVarianceResidual1986}. The key property of \(K^{*}(t;c)\) is that it satisfies the moment conditions
\begin{equation}\label{eq:boundary_kernel_moments}
    \int_{-c}^{1} K^{*}(t;c)\,\mathrm{d}t=1,\qquad
    \int_{-c}^{1} t\,K^{*}(t;c)\,\mathrm{d}t=0,
\end{equation}
which hold exactly for every \(c\geq 0\). The zero first moment is what eliminates the \(O(h)\) bias term, regardless of the degree of truncation.

\begin{remark}
    As \(u_{0}\) moves into the interior, \(c\to\infty\) and the kernel support becomes effectively symmetric. By Assumption~(A1), \(S_{1}(c)\to\mu_{1}(K)=0\), so \(K^{*}(t;c)\to K(t)\). The boundary correction vanishes precisely where it is not needed, ensuring continuity between boundary and interior behaviour.
\end{remark}

We now establish that the local-linear modification eliminates the \(O(h)\) boundary bias while preserving the variance structure. Combining~\Cref{lem:local_linear_bias_variance} with the implicit-function-theorem linearization from the cubic equation, we obtain the following uniform asymptotic normality result.

\begin{theorem}\label{thm:uniform_boundary_interior_normality}
    Under Assumptions~(A1)--(A6), for \(u_{0}=-1+ch\) with \(c\geq 0\) fixed,
    \begin{equation}\label{eq:corrected_asymptotic_normality}
        \sqrt{nh}\Big(\hat{\rho}^{\mathrm{CL}}(u_{0})-\rho(u_{0})-\mathrm{Bias}_{\rho}^{\mathrm{CL}}(u_{0})\Big)\xrightarrow{d} N\!\left(0,\frac{R(K^{*},c)}{f_{U}(u_{0})}\cdot\frac{\big(1-\rho(u_{0})^{2}\big)^{2}}{1+\rho(u_{0})^{2}}\right),
    \end{equation}
    where
    \begin{equation}\label{eq:corrected_bias}
        \mathrm{Bias}_{\rho}^{\mathrm{CL}}(u_{0})=\frac{h^{2}}{2}\rho^{\prime\prime}(u_{0})\,\mu_{2}(K^{*},c).
    \end{equation}
\end{theorem}

\begin{proof}
    The proof mirrors that of~\Cref{thm:interior_asymptotic_normality}, replacing \(B_{n}(u_{0})\) with \(B_{n}^{\mathrm{CL}}(u_{0})\). The implicit function theorem applied to~\eqref{eq:modified_cubic} yields
    \begin{equation*}
        \hat{\rho}^{\mathrm{CL}}(u_{0})-\rho(u_{0})=\frac{1}{1+\rho(u_{0})^{2}}\Big(B_{n}^{\mathrm{CL}}(u_{0})-\rho(u_{0})\Big)+o_{p}\big((nh)^{-1/2}\big).
    \end{equation*}
    Applying~\Cref{lem:local_linear_bias_variance} and the central limit theorem for \(B_{n}^{\mathrm{CL}}(u_{0})\) produces~\eqref{eq:corrected_asymptotic_normality}. Existence and uniqueness of the real root near \(\rho(u_{0})\) follow from the same Rouché-theorem argument as in~\Cref{thm:uniform_consistency}, which requires only \((A_{n},B_{n}^{\mathrm{CL}})\xrightarrow{p}(2,\rho_{0})\). This is guaranteed by~\eqref{eq:local_linear_bias}--\eqref{eq:local_linear_variance} and Assumption~(A4).
\end{proof}

\begin{remark}
    The critical feature of~\Cref{thm:uniform_boundary_interior_normality} is that the bias is \(O(h^{2})\) uniformly for all \(c\geq 0\), from the extreme boundary (\(c=0\)) through the deep interior (\(c\to\infty\)). This stands in contrast to~\Cref{thm:boundary_asymptotic_normality}, where boundary bias was \(O(h)\), necessitating a slower optimal bandwidth.
\end{remark}

The uniform \(O(h^{2})\) bias in~\Cref{thm:uniform_boundary_interior_normality} restores the faster convergence rate at boundary points. The asymptotic mean squared error is
\begin{equation}\label{eq:corrected_amse}
    \mathrm{AMSE}_{\mathrm{LL}}(h;u_{0})=h^{4}\,C_{1}^{*}(u_{0},c)+\frac{1}{nh}\,C_{2}^{*}(u_{0},c),
\end{equation}
where
\begin{align}
    C_{1}^{*}(u_{0},c)& =\frac{1}{4}\Big[\rho^{\prime\prime}(u_{0})\mu_{2}(K^{*},c)\Big]^{2},\label{eq:c1_corrected}                           \\
    C_{2}^{*}(u_{0},c)& =\frac{R(K^{*},c)}{f_{U}(u_{0})}\cdot\frac{\big(1-\rho(u_{0})^{2}\big)^{2}}{1+\rho(u_{0})^{2}}.\label{eq:c2_corrected}
\end{align}
This has the same \(h^{4}+1/(nh)\) structure as the interior case in~\Cref{cor:optimal_bandwidth}. Minimizing~\eqref{eq:corrected_amse} yields
\begin{equation}\label{eq:corrected_opt_bandwidth}
    h_{\mathrm{opt}}(u_{0},c)=\left(\frac{C_{2}^{*}(u_{0},c)}{4C_{1}^{*}(u_{0},c)\,n}\right)^{1/5}\asymp n^{-1/5},
\end{equation}
with corresponding convergence rate
\begin{equation}\label{eq:corrected_convergence_rate}
    \hat{\rho}^{\mathrm{CL}}(u_{0})-\rho(u_{0})=O_{p}\big(n^{-2/5}\big).
\end{equation}

\begin{corollary}\label{cor:uniform_optimal_rate}
    Under Assumptions~(A1)--(A6), the boundary-corrected estimator \(\hat{\rho}^{\mathrm{CL}}(u_{0})\) achieves the optimal rate \(n^{-2/5}\) uniformly across the entire design interval \([-1,1]\), including boundary points. The optimal bandwidth is \(h_{\mathrm{opt}}\asymp n^{-1/5}\) uniformly, eliminating the boundary penalty present in the local-constant estimator.
\end{corollary}

To summarise the improvement, the local-linear correction raises the boundary convergence rate from \(n^{-1/3}\) to \(n^{-2/5}\), matching the interior rate. Only the finite-sample constants \(C_{1}^{*}(u_{0},c)\) and \(C_{2}^{*}(u_{0},c)\) depend on boundary proximity through \(\mu_{2}(K^{*},c)\) and \(R(K^{*},c)\); the polynomial orders in \(n\) and \(h\) are uniform across the support.

\paragraph{Finite-sample existence guarantee.}
The deterministic existence result in~\Cref{pro:finite-sample} relied on the nonnegativity of Nadaraya--Watson weights, which guaranteed \(|B_{n}(u_{0})|\leq A_{n}(u_{0})/2\) via the Cauchy--Schwarz inequality. The local-linear weights~\eqref{eq:equivalent_kernel_weights} can be negative, so this bound need not hold in finite samples. Consequently,~\eqref{eq:modified_cubic} may occasionally have no root in \([-1,1]\) for particular finite samples.
This does not affect the asymptotic theory: by~\Cref{thm:uniform_boundary_interior_normality}, with probability approaching one, a unique real root exists near \(\rho(u_{0})\in(-1,1)\). In practice, two approaches restore the finite-sample guarantee: Clip \(B_{n}^{\mathrm{CL}}(u_{0})\) to \([-A_{n}(u_{0})/2, A_{n}(u_{0})/2]\) before solving~\eqref{eq:modified_cubic}. This has asymptotically negligible effect but ensures existence.

\section{Simulation Study}\label{sec:simulation_study}
In this section, we conduct a series of simulation studies to compare the performance of the estimation method based on solving the cubic equation with the other methods. The estimators we consider in this section include the Nadaraya-Watson estimator in~\eqref{eq:correlation_coefficient_nonparametric_2}, the cubic-equation-based estimator established in~\Cref{sec:Methodology}, the cubic-equation-based estimator with boundary correction using local linear adjustments, and the direct local linear estimator described in Appendix~\ref{sec:appendix_proofs}. We call these four estimators NW, CE, CL, and LL, respectively.

\paragraph{Correlation Function}\label{para:correlation_function}
In this simulation study, we consider using the underlying correlation functions as follows: parabolic function \( \rho(u)=1-2u^{2} \), sinusoidal function \( \rho(u)=\sin(\pi u/2) \), Cosine function \( \rho(u)=\cos(\pi u/2) \), cubic function \( \rho(u)=u^{3} \), linear function \( \rho(u)=u \), exponential function \( \rho(u)=(\exp(u)-\exp(-1))/ (\exp(1)-\exp(-1))\), upper half semi-circle \( \rho(u)=\sqrt{1-u^{2}} \), normalised inverse sinusoidal function \( 2/\pi*\arcsin(u) \) and the third Chebyshev polynomial \( \rho(u)=4u^{3}-3u \). For simplicity, we denote them in order as \( F1, F2, F3, F4, F5, F6, F7, F8, F9 \). These correlation functions represent a wide range of correlation structures, including monotonic, non-monotonic, and oscillatory behaviours. The choice of these functions allows us to evaluate the performance of the four estimators under various scenarios that may arise in practical applications, see~\Cref{fig:figs/correlation_functions.pdf}.

\paragraph{Kernel Function}\label{para:kernel_function}
We use Gaussian, Epanechnikov, Triangular, Uniform, Quadratic, and Triweight kernel functions in the simulation study, see~\Cref{fig:figs/kernel_functions_comparison.pdf}. For simplicity, we denote them in order as \( K1, K2, K3, K4, K5, K6 \) respectively.
\begin{figure}[ht]
    \centering
    \begin{subfigure}[bt]{0.4\textwidth}
        \centering
        \includegraphics[width=\textwidth]{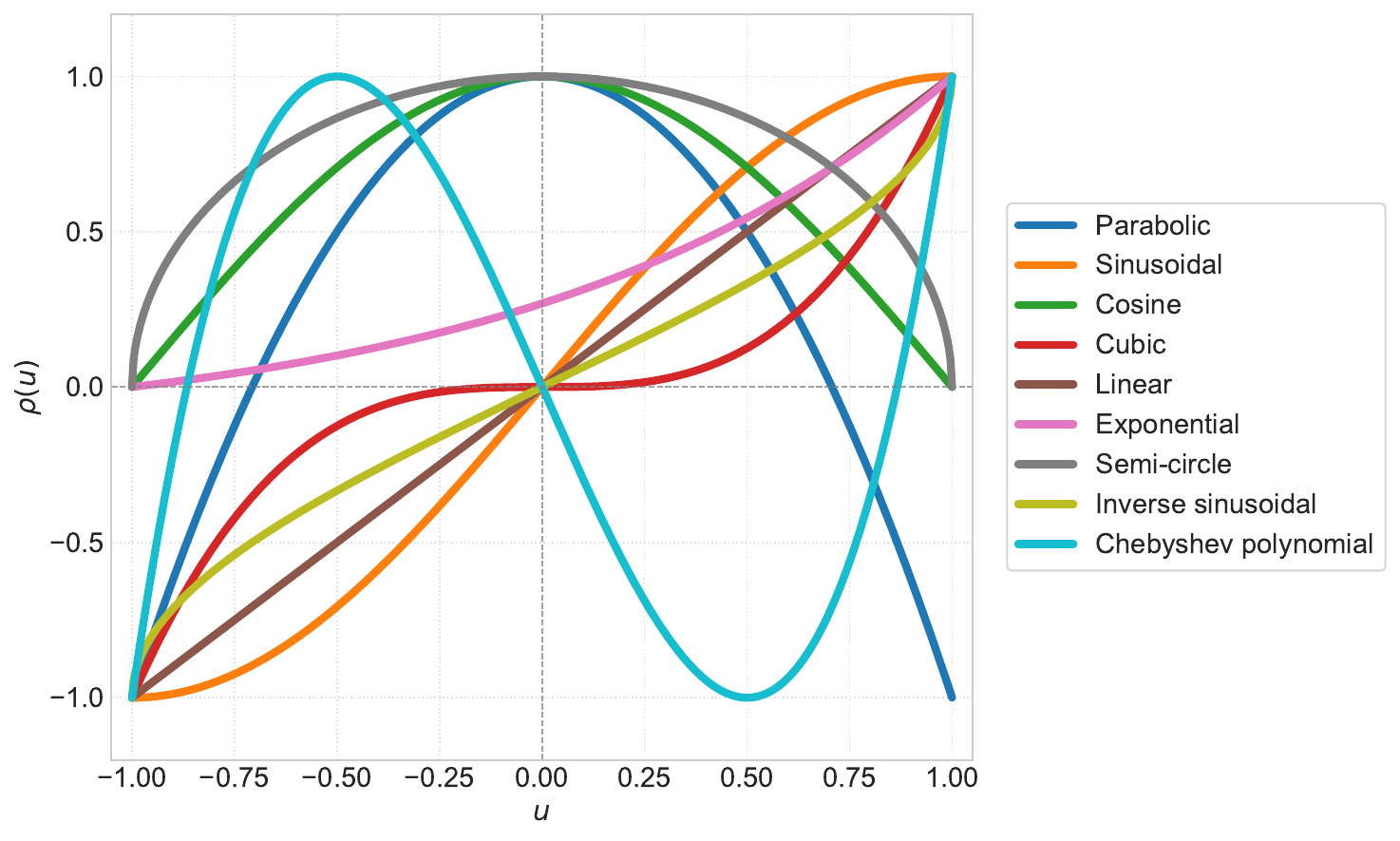}
        \caption{The Underlying Correlation Functions}\label{fig:figs/correlation_functions.pdf}
    \end{subfigure}
    \begin{subfigure}[bt]{0.4\textwidth}
        \centering
        \includegraphics[width=\textwidth]{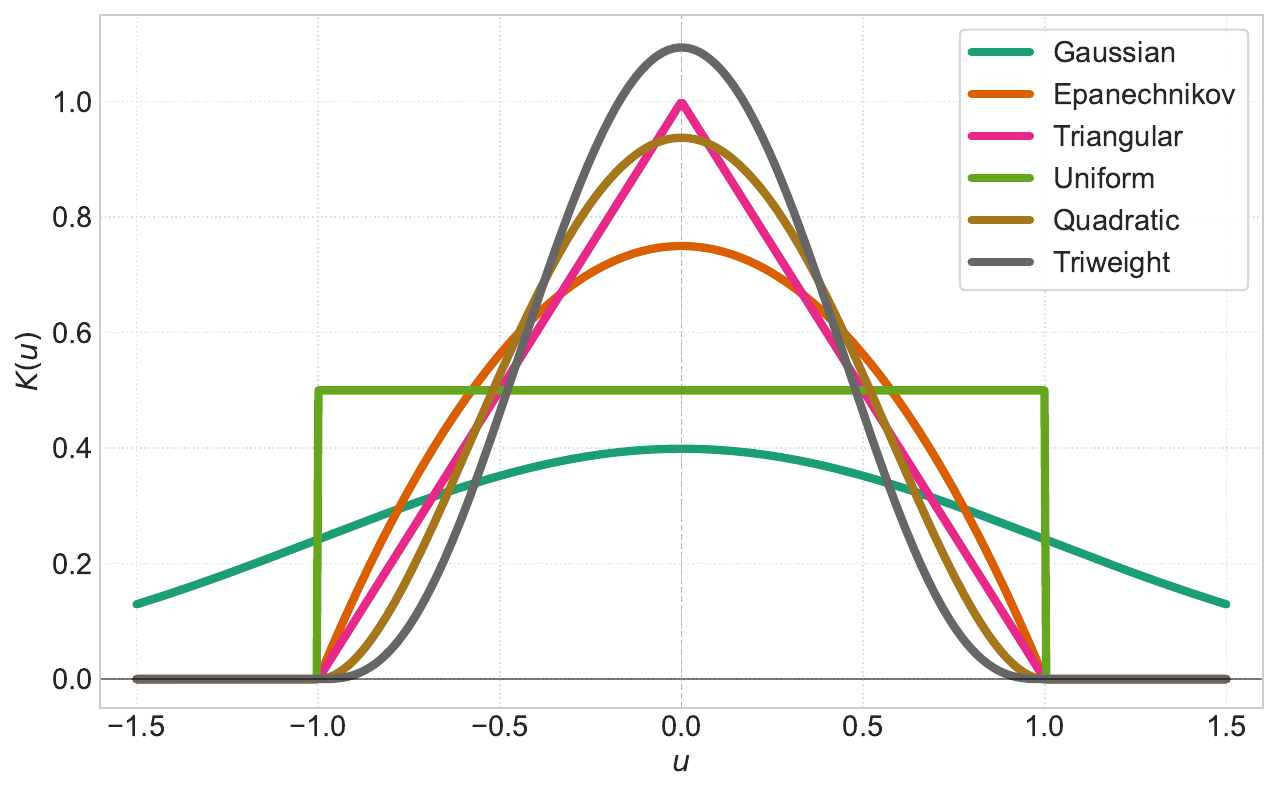}
        \caption{The Kernel Functions}\label{fig:figs/kernel_functions_comparison.pdf}
    \end{subfigure}
    \caption{Correlation and Kernel Functions}\label{fig:correlation_kernel_functions}
\end{figure}
\paragraph{Data generation and parameter settings}\label{para:parameter_settings}
The sample size \( n \) is set to 100, 200, 400, 800, 1600 and 3200. Given the underlying correlation function and kernel function, we first generate the \( U_{i}, i=1,2,\ldots,n \) from the Uniform distribution \( \mathcal{U}(-1,1) \), then for each \( U_{i} \), randomly generate the \( x_{i1}, x_{i2} \) from the bivariate normal distribution with zero mean and covariance matrix \(
    \begin{bmatrix}
        1          & \rho(U_{i}) \\
        \rho(U_{i})& 1
\end{bmatrix}  \). Once we obtain the observations \( (U_{i}, x_{i1}, x_{i2}) \), we apply the four estimators to estimate the correlation function. The optimal bandwidth is selected via grid search over 100 points spaced logarithmically from \( 10^{-3} \) to \( 10^{0} \). To evaluate the performance of the four estimators at boundaries, we adopt two sampling strategies: (i) randomly sample \( u_{j}^{\text{test}}, j=1,2,\ldots,m \) from the Uniform distribution \( \mathcal{U}(-1,1) \), and (ii) randomly sample 25\%, 50\%, 25\% of \( u_{j}^{\text{test}}, j=1,2,\ldots,m \) from the Uniform distributions \( \mathcal{U}(-1,-0.9]\), \( \mathcal{U}(-0.9,0.9)\) and \( \mathcal{U}[0.9,1)\) respectively. The first strategy allows us to evaluate the performance of the four estimators across the entire support, while the second strategy focuses on the boundary regions. We use the mean squared error (MSE) to evaluate the performance of the four estimators, which is defined as \( \text{MSE} = \frac{1}{m}\sum_{j=1}^{m}(\hat{\rho}(u_{j}^{\text{test}})-\rho(u_{j}^{\text{test}}))^{2} \), where \( m \) is the number of test points. We set \( m=100 \) in this simulation. The total number of combinations of the sample size, the correlation functions, kernel functions and sampling strategies is 6*9*6*2=648. For each combination, we repeat the simulation 200 times to obtain the MSE for each estimator. Furthermore, we also record the computation time for each estimator to evaluate their computational efficiency.

\paragraph{The effect of sample size}\label{para:the_effect_of_sample_size}
\Cref{fig:figs/sample_size_vs_mse.pdf} shows the log-transformed MSEs of four methods under the two sampling strategies against the different sample size. The ``uniform'' and ``boundary'' sampling strategies are denoted as ``S1'' and ``S2''  respectively in the rest of this section. Given the sample size \( n \), for every combination of methods and sampling strategies, we estimate the kernel density of the log-transformed MSEs by pooling them from across 9  correlation functions, 6 kernel functions, and 200 replications. The left dot on the kernel density estimation (KDE) is the corresponding median of the log-transformed MSEs. The KDE is estimated as follows: we adopt a univariate Gaussian KDE; the bandwidth is selected according to the normal-reference rule \( h=1.06\hat{\sigma}N^{-1/5} \), where \( N=9\times6\times200=10,800 \) and \( \hat{\sigma} \) is the standard deviation of log-transformed MSEs~\citep{silvermanDensityEstimationStatistics1986}. The resulting KDEs are used solely to provide a smoothed visual representation of the empirical distributions and are not used for formal statistical inference.

\Cref{fig:figs/sample_size_vs_mse.pdf} highlights several important facts.
\begin{itemize}
    \item For every combination of estimation method and sampling strategy, the MSE decreases as the sample size increases, in agreement with the theoretical results established in~\Cref{sec:Methodology}. For example, under the NW+S1 configuration, the MSE decreases from approximately \( 10^{-1.75} \) at \( n=100 \) to nearly \( 10^{-3.0} \) at \( n=3200 \) in~\Cref{fig:figs/sample_size_vs_mse.pdf}.
    \item When comparing the two sampling strategies, we find that, for every combination of estimation method and sample size, the uniform sampling strategy (S1) consistently yields a more favourable distribution of the MSE than the boundary sampling strategy (S2). This pattern is consistent across all combinations of methods and sample sizes. This result is largely attributable to the differences in the design of sampling strategies. S2 draws 50\% of the test points from the narrow boundary regions \((-1,-0.9]\cup[0.9,1)\), which are more challenging for estimation due to the boundary effects. In contrast, S1 samples across the entire support, leading to a more balanced representation of the correlation function and thus lower MSEs. The scheme of S2 is introduced specifically to assess the effectiveness of the boundary correction in LL and CL methods.~\Cref{fig:figs/sample_size_vs_mse.pdf} provide clear evidence that the boundary correction indeed reduces the estimation error. For instance, when \( n=200 \), the vertical distance between the median MSEs of LL under S1 and S2 (represented by two green dots) is substantially smaller than that of NW under S1 and S2 (represented by two blue dots). Similarly, the vertical distance between the median MSEs of CL under S1 and S2 (represented by two orange dots) is smaller than that of CE under S1 and S2 (represented by two red dots). This indicates that the boundary correction effectively mitigates the additional estimation error induced by boundary sampling. A more detailed assessment of the boundary effect is provided later, based on the MSE calculated specifically from test points in the boundary regions.
    \item From a method comparison perspective, NW and CE do not incorporate boundary correction, while LL and CL do. Consequently, the latter two methods consistently outperform the former two in terms of MSE, regardless of the sampling strategy adopted. A more detailed comparison of the four methods is provided later in this section.
\end{itemize}
\begin{figure}[ht]
    \centering
    \includegraphics[width=0.95\textwidth]{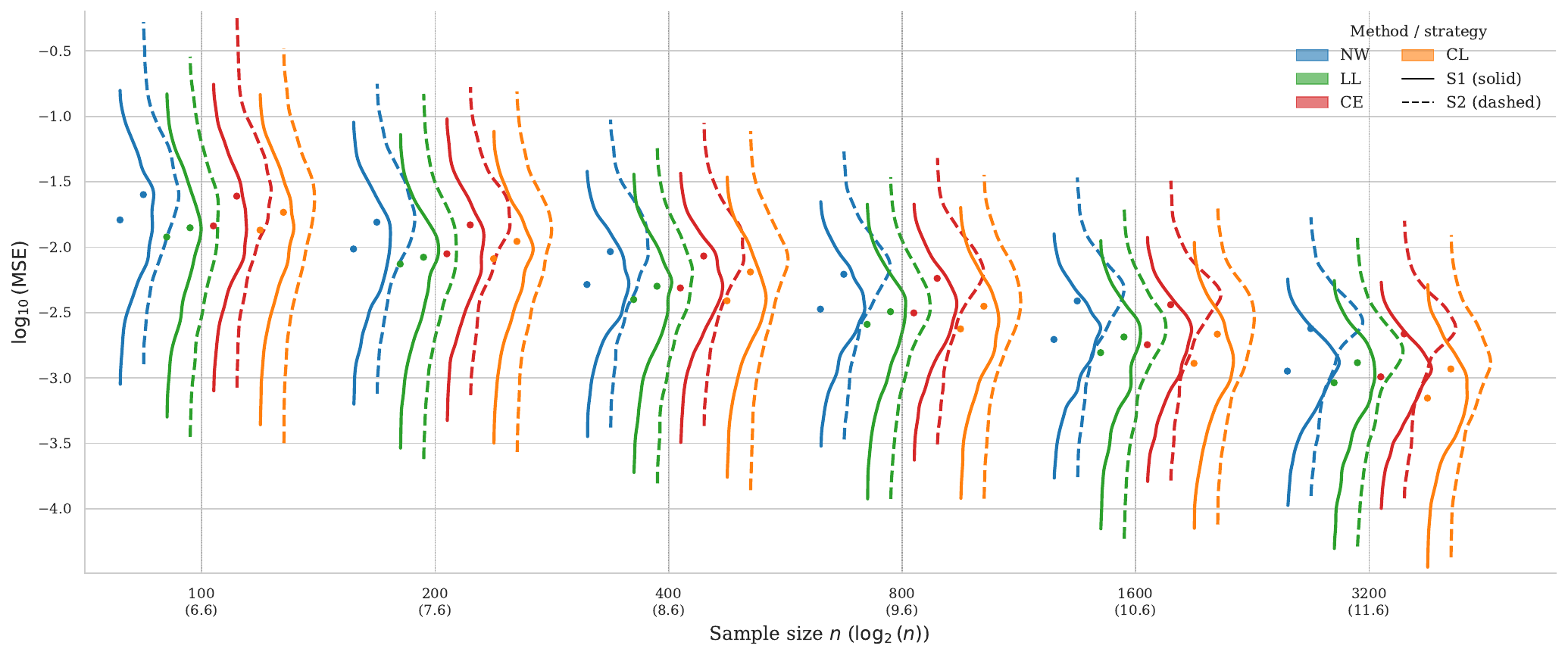}
    \caption{The MSEs of four methods (NW, LL, CE and CL) against the sample size $n$. Both the
        $x$-axis and $y$-axis are plotted on logarithmic scales to accommodate the wide range of the data
        and improve visual clarity. ``S1'' and ``S2'' represent the uniform sampling strategy and bound-
        ary strategy, respectively. For each sample size $n$, the speciﬁc method and sampling strategy,
        the kernel density of the log-transformed MSEs is estimated from pooled MSE values across 9
        correlation functions, 6 kernel functions, and 200 replications. The dot on the left of each kernel density
    function represents the corresponding median value of the log-transformed MSE.}\label{fig:figs/sample_size_vs_mse.pdf}
\end{figure}

Furthermore, we generate the heatmap of the median logarithmic MSEs across different correlation functions and kernel functions for each combination of estimation method, sampling strategy, and sample size in~\Cref{fig:figs/panel_mse_heatmap.pdf} of the Appendix. This heatmap provides a comprehensive overview of the performance of the four methods under various experimental settings, allowing for a more detailed comparison beyond the sample size analysis presented earlier.

\paragraph{The boundary correction}\label{para:the_boundary_correction}
To specifically evaluate the performance of the four estimators in the boundary regions, we calculate the MSEs based on the test points sampled from the boundary regions \((-1,-0.9]\cup[0.9,1)\) under the boundary sampling strategy. The boundary MSEs for each estimator in the boundary regions are defined as follows:
\begin{equation*}
    \mathrm{MSE}_{\mathrm{B}} = \frac{1}{m_{\mathrm{B}}}\sum_{j\in \mathcal{J}}(\hat{\rho}(u_{j}^{\mathrm{test}})-\rho(u_{j}^{\mathrm{test}}))^{2},
\end{equation*}
where \( \mathcal{J}\coloneqq\{ j \vert u_{j}^{\mathrm{test}}\in(-1,-0.9]\cup[0.9,1) \} \) and \( m_{B} \) is the cardinality of set \( \mathcal{J} \). The boundary MSEs of the four methods against the sample size \( n \) are shown in~\Cref{fig:figs/sample_size_vs_mse_boundary.pdf}. Similar to~\Cref{fig:figs/sample_size_vs_mse.pdf}, for each sample size \( n \), the violin boxplot of boundary MSEs is based on the pooled boundary \( \mathrm{MSE}_{\mathrm{B}} \) values across 9 correlation functions, 6 kernel functions, and 200 replications. From~\Cref{fig:figs/sample_size_vs_mse_boundary.pdf}, we can observe that the boundary MSEs of LL and CL are consistently lower than those of NW and CE across all sample sizes. This indicates that the boundary correction implemented in LL and CL effectively reduces the estimation error in the boundary regions, confirming the theoretical results established in~\Cref{sec:Methodology}. Among the four methods, LL exhibits the best performance in the boundary regions, suggesting that direct local linear estimation in the Fisher z-space is the most effective approach. However, we can also observe from the figure that as \( n \) increases, the median \( \mathrm{MSE}_{\mathrm{B}} \) of CL gradually approaches that of LL, and by \( n=3200 \), it surpasses LL's median \( \mathrm{MSE}_{\mathrm{B}} \). This indicates that with larger sample sizes, the CL method can achieve performance comparable to that of LL.\ It is noteworthy that CL outperforms CE in the boundary regions, demonstrating that local linear adjustment for \( B_n \) is effective. On the other hand, CL only corrects \( B_n \) without applying boundary correction to \( A_n \), which may explain why CL performs slightly worse than LL in smaller sample sizes.

\begin{figure}[ht]
    \centering
    \includegraphics[width=0.95\textwidth]{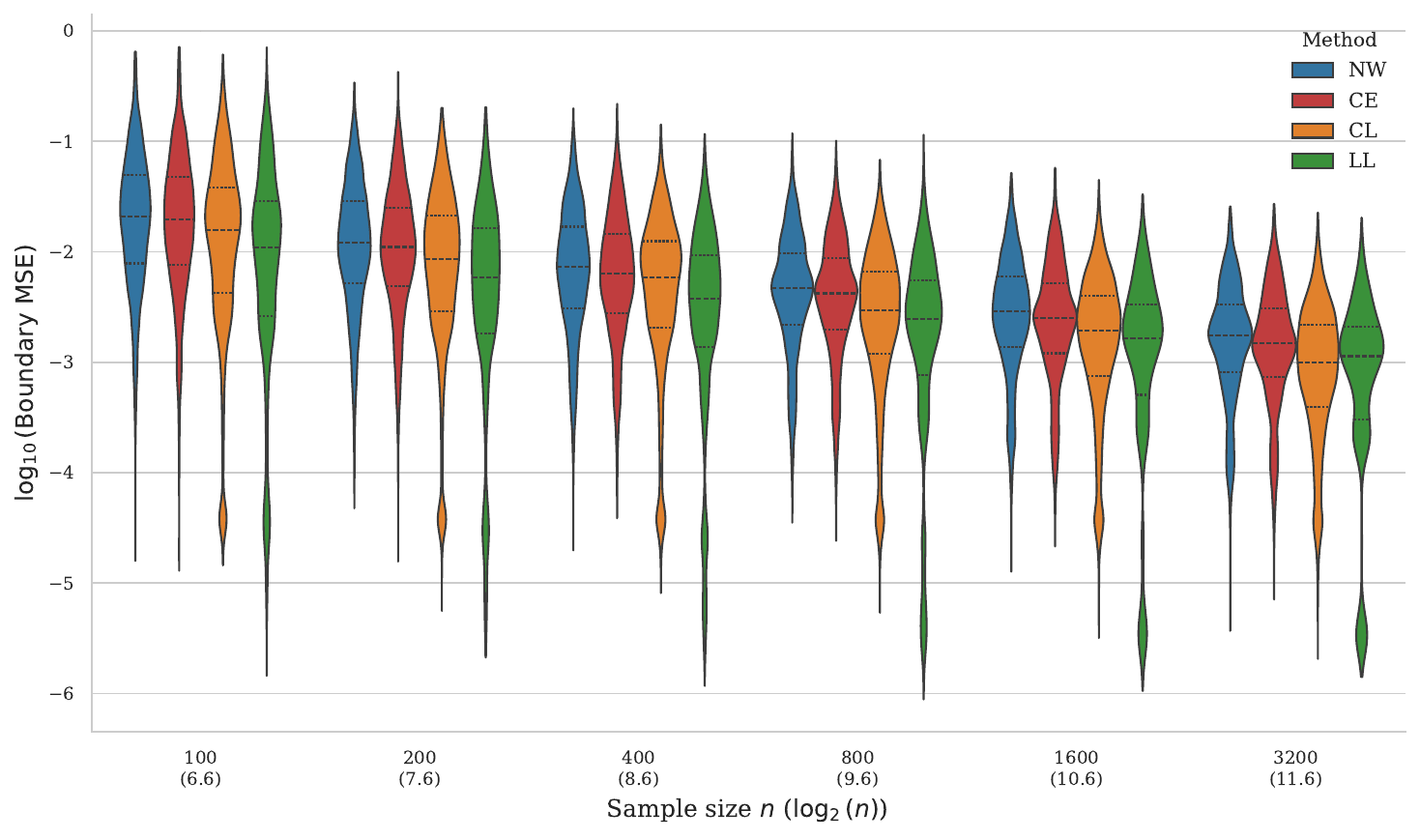}
    \caption{The boundary MSEs of four methods (NW, LL, CE and CL) against the sample size $n$. Both the
        $x$-axis and $y$-axis are plotted on logarithmic scales to accommodate the wide range of the data
        and improve visual clarity. For each sample size $n$, the speciﬁc method and sampling strategy,
        the violin boxplot is based on the pooled $\mathrm{MSE}_{\mathrm{B}}$ values across 9
    correlation functions, 6 kernel functions, 200 replications.}\label{fig:figs/sample_size_vs_mse_boundary.pdf}
\end{figure}

\paragraph{The time consumption}\label{para:the_time_consumption}
~\Cref{fig:figs/panel_time_heatmap.pdf} presents heatmaps of the time consumption across different methods, sample sizes, sampling strategies, correlation functions, and kernel functions. The horizontal axis consists of the 54 combinations of correlation functions and kernel functions, while the vertical axis represents the 48 combinations of estimation methods, sample sizes, and sampling strategies. Each cell reports the median time consumption over the 200 simulation replications for a given combination of correlation function, kernel function, method, sample size, and sampling strategy. Darker shading indicates a larger median time consumption.
Among the four methods, LL consistently incurs the highest computational cost. For all methods, time consumption increases with sample size, reflecting the expected decline in computational efficiency as the sample size increases. In contrast, the choice of correlation function and kernel function has relatively little effect on the time consumption.

\begin{figure}[ht]
    \centering
    \includegraphics[width=0.95\textwidth]{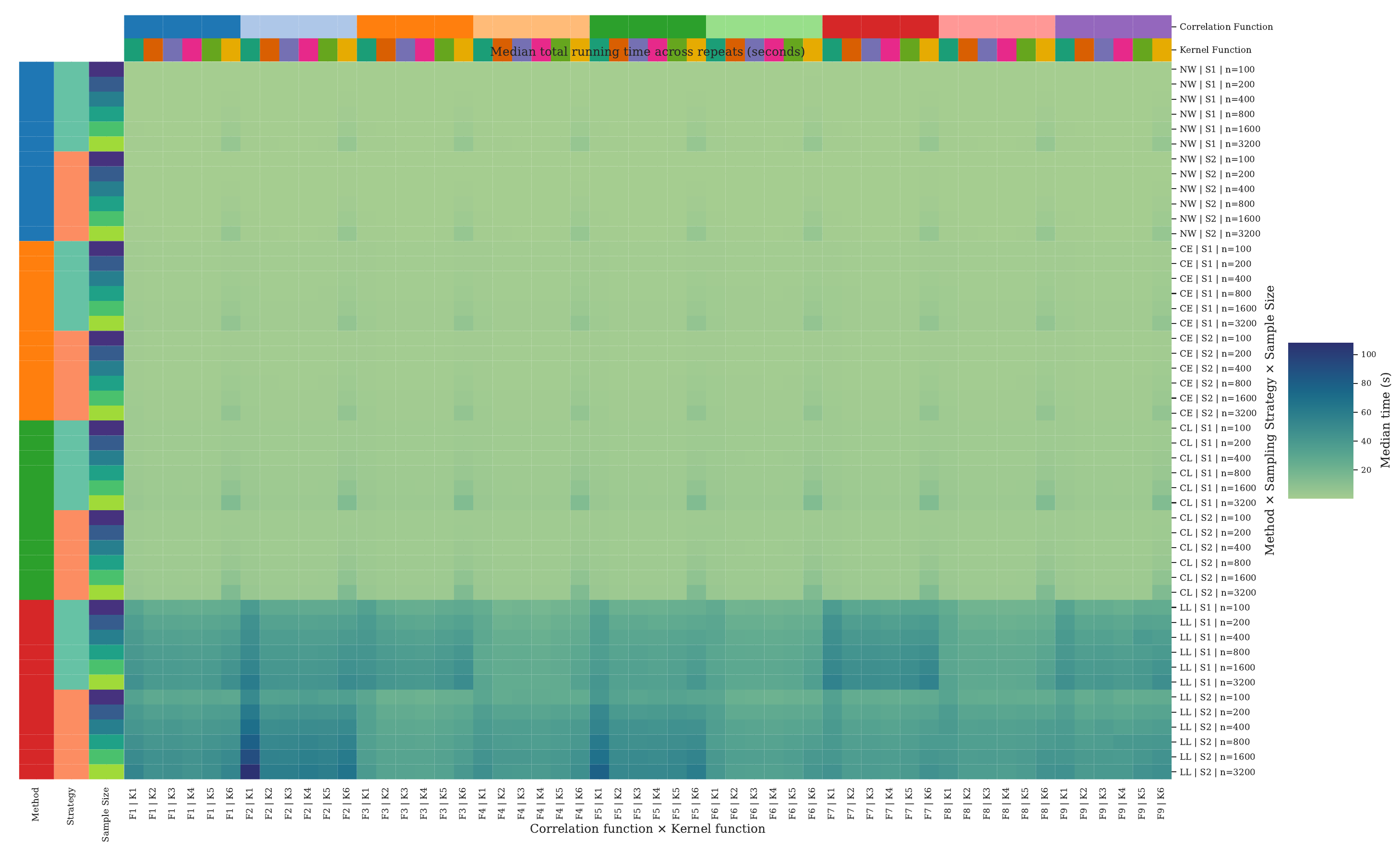}
    \caption{Heatmap of the median time consumption (in seconds) for each combination of estimation method, sample size, sampling strategy, correlation function, and kernel function, with darker shading indicating larger time consumption.}\label{fig:figs/panel_time_heatmap.pdf}
\end{figure}

To further examine how time consumption varies with sample size, we take the median time consumption across the 9 correlation functions, 6 kernel functions, and 200 simulation replications, and plot time consumption against sample size for the eight combinations of method and sampling strategy. The resulting plots are shown in \Cref{fig:figs/panel_time_vs_sample_size.pdf}. The solid line with square markers represents the median time consumption under sampling strategy S1 as a function of sample size, whereas the dashed line with star markers represents the corresponding median under sampling strategy S2. The 95\% empirical confidence intervals are also shown using different colours, with their boundaries indicated by solid and dashed lines for S1 and S2, respectively.
As expected, the time consumption increases with sample size for all methods and under both sampling strategies. Moreover, S2 consistently requires more time consumption than S1 across all sample sizes and methods. This difference is attributable to the greater number of test points located near the boundaries under S2, which increases the computational burden associated with estimation in the boundary regions.

To provide a more direct illustration of the trade-off between estimation accuracy and computational efficiency, we plot the accuracy-speed trade-off for the four methods, as shown in \Cref{fig:figs/panel_accuracy_speed_tradeoff.pdf}. Each point represents the median time consumption plotted against the median MSE.\ From this perspective, a particularly noteworthy result occurs when n=3200: CL achieves a lower MSE than LL while requiring substantially less time consumption. Overall, LL provides the highest estimation accuracy but at the greatest computational cost. In contrast, CL offers a more favourable balance between accuracy and computational efficiency and, in some settings, even outperforms LL in both respects. NW and CE are computationally more efficient, but this advantage comes at the cost of lower estimation accuracy.
\begin{figure}[ht]
    \centering
    \begin{subfigure}[bt]{0.48\textwidth}
        \centering
        \includegraphics[width=\textwidth]{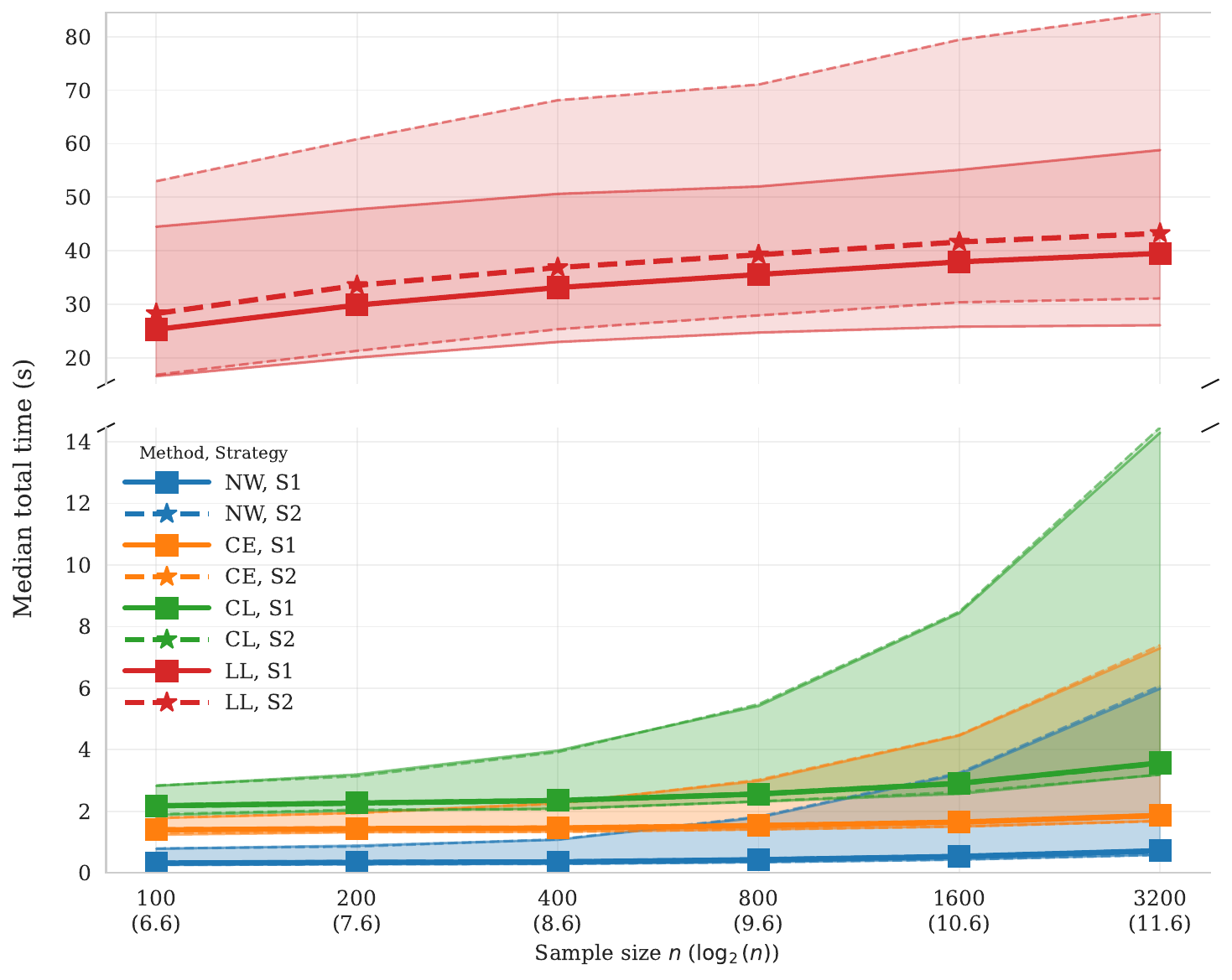}
        \caption{Median time consumption (in seconds) against sample size}\label{fig:figs/panel_time_vs_sample_size.pdf}
    \end{subfigure}
    \begin{subfigure}[bt]{0.48\textwidth}
        \centering
        \includegraphics[width=\textwidth]{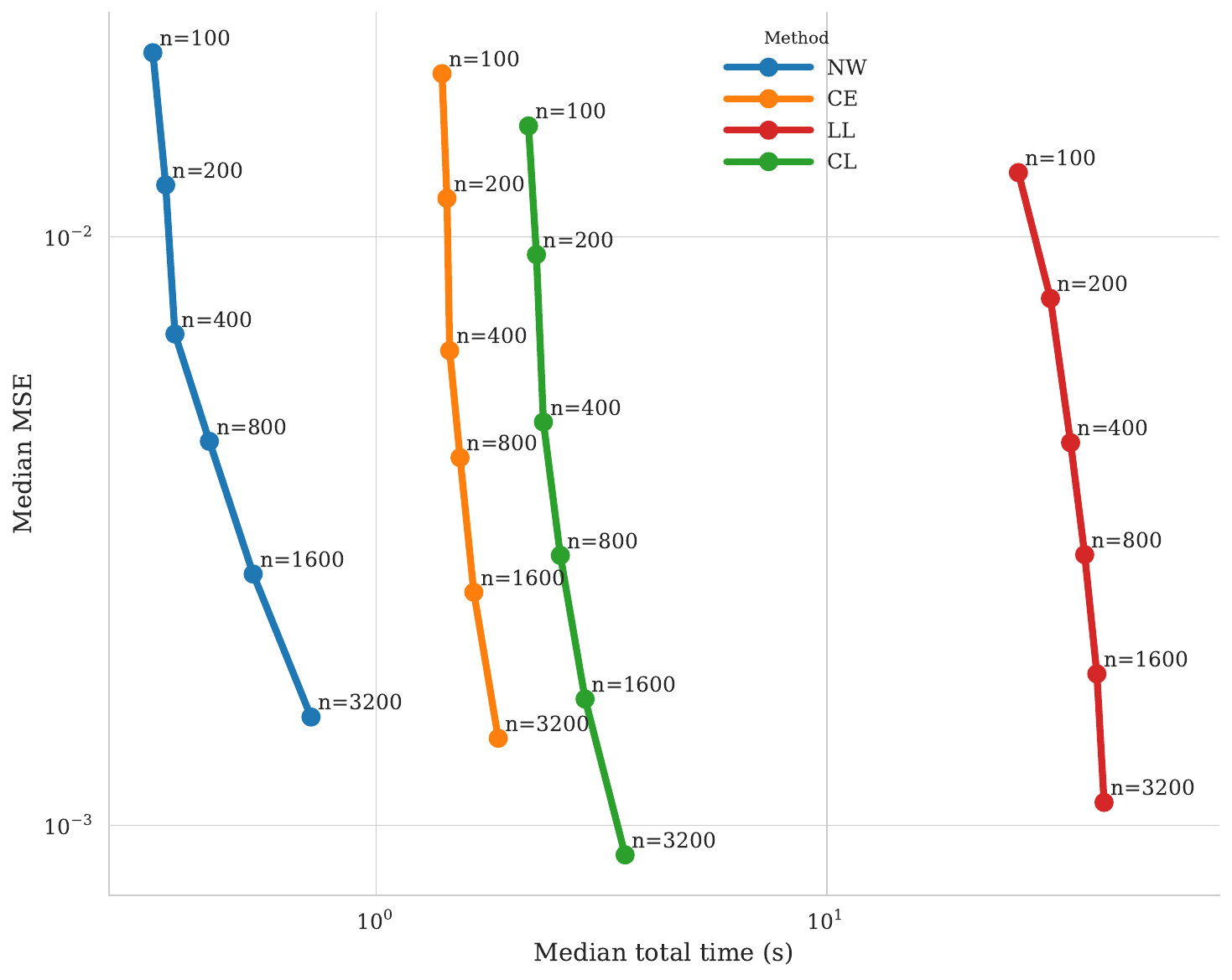}
        \caption{The trade-off between estimation accuracy and computational efficiency for the four methods.}\label{fig:figs/panel_accuracy_speed_tradeoff.pdf}
    \end{subfigure}
    \caption{Computational cost of the four estimation methods and their trade-off with estimation accuracy. (a) Median time consumption as a function of sample size for the four methods under the sampling strategies S1 and S2, showing that time consumption increases with sample size and is larger under S2; (b) the accuracy-speed trade-off, where each point plots the median time consumption against the median MSE, illustrating that LL achieves the highest accuracy at the greatest computational cost while CL offers the most favourable balance.}\label{fig:the_time_consump}
\end{figure}

\paragraph{How to choose the method?}\label{para:how_to_choose_method}
In practice, estimation accuracy alone is insufficient; the associated computational burden imposes a critical constraint, particularly in online or high-throughput settings. For instance, in typical MEG source reconstruction, the sampling rate is often 1000 Hz, yielding 1000 time points per second, amounting to 600,000 time points over a 10-minute recording. Even when downsampled to 100 Hz, this still results in 60,000 observation points, and the combinatorial pairing across cortical regions further inflates the total computational load, making algorithmic efficiency a primary practical concern.

To balance these two competing dimensions: statistical precision and computational cost, we introduce two complementary summary metrics. Let each experimental configuration be indexed by \(c=(n,F,K,S)\in\mathcal{C}\), where \(|\mathcal{C}|=6\times 9\times 6\times 2=648\), and let \(m\in\{\text{NW},\text{LL},\text{CE},\text{CL}\}\) index the candidate methods. For each configuration \(c\), method \(m\), and replication \(r=1,\dots,200\), we compute the mean squared error (MSE) based on the estimates at the test location \(u_{j}^{\mathrm{test}}, j=1,2,\ldots,200\). To obtain a robust summary per scenario, we first take the median across replications:
\begin{equation*}
    \widehat{\mathrm{MSE}}_{c,m}=\operatorname{median}_{r=1,\dots,200}\bigl(\mathrm{MSE}_{c,m,r}\bigr).
\end{equation*}
Similarly, the observed runtime is summarised by the median time required to estimate the 200 test points \(u_j^{\mathrm{test}}, j=1,2,\ldots,200\) under configuration \(c\) and method \(m\), denoted by \(\hat{t}_{c,m}\).

For each configuration \(c\), we rank the four methods separately according to \(\widehat{\mathrm{MSE}}_{c,m}\) and \(\hat{t}_{c,m}\), yielding the ordinal scores \(R^{\mathrm{MSE}}_{c,m}\) and \(R^{\mathrm{time}}_{c,m}\), where a lower rank indicates superior performance. Averaging these ranks across all configurations gives the marginal mean ranks:
\[
    \bar{R}^{\mathrm{MSE}}_m=\frac{1}{|\mathcal{C}|}\sum_{c\in\mathcal{C}} R^{\mathrm{MSE}}_{c,m},\qquad
    \bar{R}^{\mathrm{time}}_m=\frac{1}{|\mathcal{C}|}\sum_{c\in\mathcal{C}} R^{\mathrm{time}}_{c,m}.
\]
Our first composite index is then defined as the average of the two marginal ranks:
\[
    \bar{R}^{\text{dual}}_m=\frac{1}{|\mathcal{C}|}\sum_{c\in\mathcal{C}}
    \left(\frac{R^{\mathrm{MSE}}_{c,m}+R^{\mathrm{time}}_{c,m}}{2}\right),
\]
which assigns equal implicit weight to accuracy and speed; smaller values indicate a better trade-off.

Recognising that the relative importance of MSE versus runtime is inherently application-dependent, we further devise a flexible weighted criterion. For each configuration \(c\), we apply min-max normalisation separately to \(\log_{10}(\widehat{\mathrm{MSE}}_{c,m})\) and \(\log_{10}(\hat{t}_{c,m})\) across the four methods, so that both transformed quantities lie in \([0,1]\) with 0 corresponding to the best performance within that scenario. Denote these normalised scores by \(\log_{10}\mathrm{MSE}^{\text{norm}}_{c,m}\) and \(\log_{10}t^{\text{norm}}_{c,m}\). For a user-specified weight \(w\in[0,1]\) that controls the emphasis on estimation accuracy, we define the effective score
\[
    \text{eff}_m(w)=\frac{1}{|\mathcal{C}|}\sum_{c\in\mathcal{C}}
    \Bigl(w\cdot \log_{10}\mathrm{MSE}^{\text{norm}}_{c,m}
    +(1-w)\cdot \log_{10}t^{\text{norm}}_{c,m}\Bigr).
\]
We then evaluate this criterion over a range of weights \(w\in\{ 0.1,0.2,0.3,0.4,0.5,0.6,0.7,0.8,0.9 \}\), and for each \(w\), we rank the four methods according to \(\text{eff}_m(w)\) to obtain \(R^{\mathrm{eff}}_{m}(w)\). The resulting rankings in~\Cref{tab:method_ranking} provide a comprehensive picture of how the methods compare under varying preferences for accuracy versus computational efficiency. We can observe that CL exhibits the most robust overall performance: it attains the best effective rank \(R^{\mathrm{eff}}_{m}(w)\) for all weights \(w=0.4\) to \(0.8\) and ranks second at \(w=0.3\) and \(w=0.9\), and it also achieves the smallest dual rank \(\bar{R}^{\mathrm{dual}}_{m}=2.398\), indicating a strong balance between estimation accuracy and computational cost. LL delivers the highest estimation accuracy, as reflected by the smallest mean MSE rank \(\bar{R}^{\mathrm{MSE}}_{m}=1.455\), but its dual rank \(\bar{R}^{\mathrm{dual}}_{m}=2.728\) is the least favourable among the four methods owing to its substantially higher computational cost; consequently, LL attains the top effective rank only when accuracy is assigned the largest weight (\(w=0.9\)). NW is computationally the most efficient (\(w=0.1,0.2,0.3\)), yielding the second-best dual rank \(\bar{R}^{\mathrm{dual}}_{m}=2.400\), yet it suffers from the worst estimation accuracy (\(\bar{R}^{\mathrm{MSE}}_{m}=3.799\)), and its effective rank deteriorates as \(w\) increases, so that it is only attractive when computational speed is prioritised over precision. CE performs moderately on both dimensions, occupying an intermediate position in the effective ranking across all weights.

\begin{table}[ht]
    \centering
    \begin{tabular}{lccc*{9}{c}}
        \toprule
        \multirow{2}{*}{Method}& \multirow{2}{*}{\(\bar{R}^{\mathrm{MSE}}_{m}\)}& \multirow{2}{*}{\(\bar{R}^{\mathrm{dual}}_{m}\)}& \multicolumn{9}{c}{\(R^{\mathrm{eff}}_{m}\) with different weight} \\
        \cmidrule(lr){4-12}
                               &                                                &                                                 &0.1 & 0.2& 0.3& 0.4& 0.5& 0.6& 0.7& 0.8& 0.9 \\
        \midrule
        CL                     & 1.796                                          & 2.398                                           &3   &3   &2   &1   & 1  & 1  & 1  & 1  & 2 \\
        LL                     & 1.455                                          & 2.728                                           &4   &4   &4   &4   & 4  & 2  & 2  & 2  & 1 \\
        CE                     & 2.949                                          & 2.475                                           &2   &2   &3   &3   & 3  & 4  & 3  & 3  & 3 \\
        NW                     & 3.799                                          & 2.400                                           &1   &1   &1   &2   & 2  & 3  & 4  & 4  & 4 \\
        \bottomrule
    \end{tabular}
    \caption{The ranking of the four methods in terms of MSE, dual criterion, and effective score rank under different weights \(w\).}\label{tab:method_ranking}
\end{table}

\paragraph{Wilcoxon Signed-rank Analysis: CL vs LL}\label{para:wilcoxon_signedrank_analysis_cl_vs_ll}
In this section, we perform one-sided paired Wilcoxon signed-rank tests to compare the performance of the CL and LL methods across all simulation cells. The tests are conducted by pairing Monte Carlo replicates via the repetition index, with the alternative hypothesis that CL has smaller mean squared error (MSE) than LL. To be concrete, let \( \mathrm{MSE}^{\mathrm{CL}}_{r,c} \) and \( \mathrm{MSE}^{\mathrm{LL}}_{r,c} \) denote the MSEs of the CL and LL methods, respectively, for the \(r\)-th repetition in the \(c\)-th simulation cell. For each cell \( c \in \mathcal{C}\), we define the  paired differences \( d_{r,c}\coloneqq  \mathrm{MSE}^{\mathrm{CL}}_{r,c} - \mathrm{MSE}^{\mathrm{LL}}_{r,c} \). The one-sided paired Wilcoxon signed-rank test is then applied to test whether CL is better than LL in terms of MSE:
\begin{equation*}
    \mathbb{H}_{0}: ~\text{median}(d_{r,c}) \geq 0 \quad \text{vs} \quad \mathbb{H}_{1}: \text{median}(d_{r,c}) < 0.
\end{equation*}
For each simulation cell \( c \), we report the one-sided \( p\)-value (\( p_{c} \)), and mean difference \( \bar{d}_c \coloneqq \frac{1}{200} \sum_{r=1}^{200} d_{r,c} \). A cell is marked as ``CL better'' if \( p_{c} < 0.05 \) and \( \bar{d}_c < 0 \). To summarise the overall performance across all simulation cells, we unify significance and direction by the transformed score \( s_c \coloneqq -\log_{10}(p_c) \cdot \text{sign}(\bar{d}_c) \). The resulting scores can then be visualised using heatmaps in~\Cref{fig:figs/wilcoxon_heatmap_s1_s2.pdf}. Furthermore, we compute the significance win rate, defined as the proportion of cells where CL is significantly better than LL, and demonstrate it in~\Cref{fig:figs/wilcoxon_significance_rate.pdf}.
\begin{figure}[!ht]
    \centering
    \begin{subfigure}[bt]{0.9\textwidth}
        \centering
        \includegraphics[width=\textwidth]{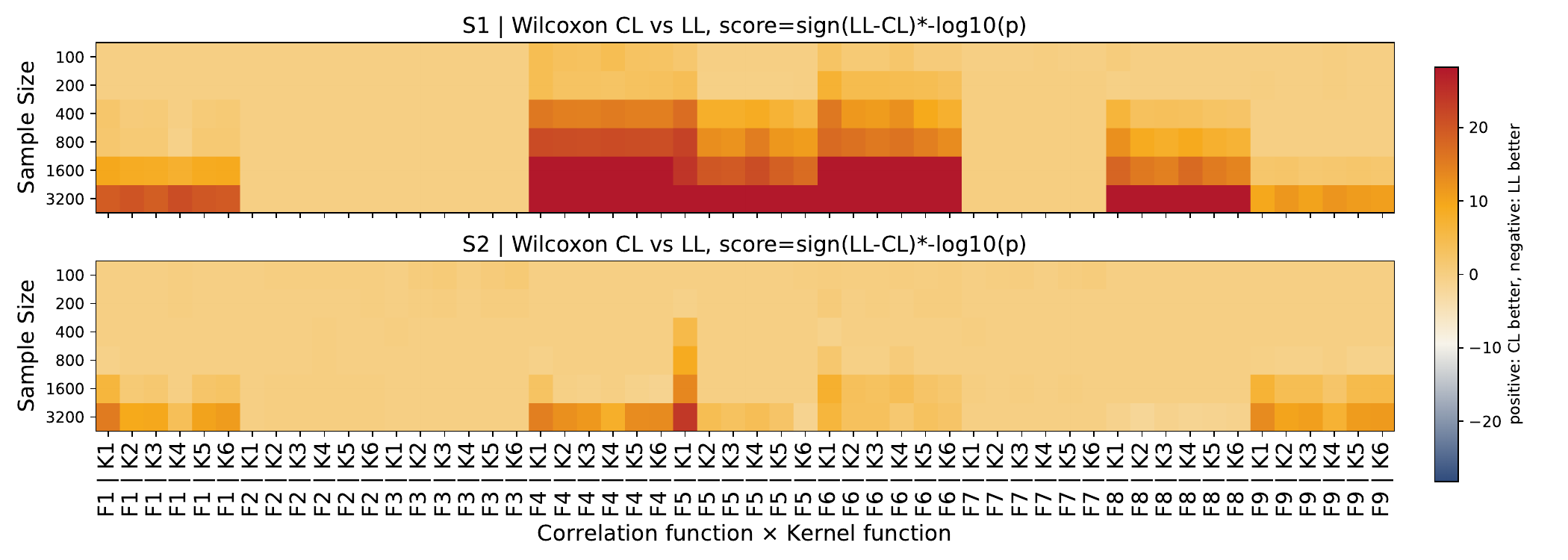}
        \caption{The heatmap of Wilcoxon signed-rank test transformed scores for CL vs LL across simulation cells.}\label{fig:figs/wilcoxon_heatmap_s1_s2.pdf}
    \end{subfigure}
    \begin{subfigure}[bt]{0.8\textwidth}
        \centering
        \includegraphics[width=\textwidth]{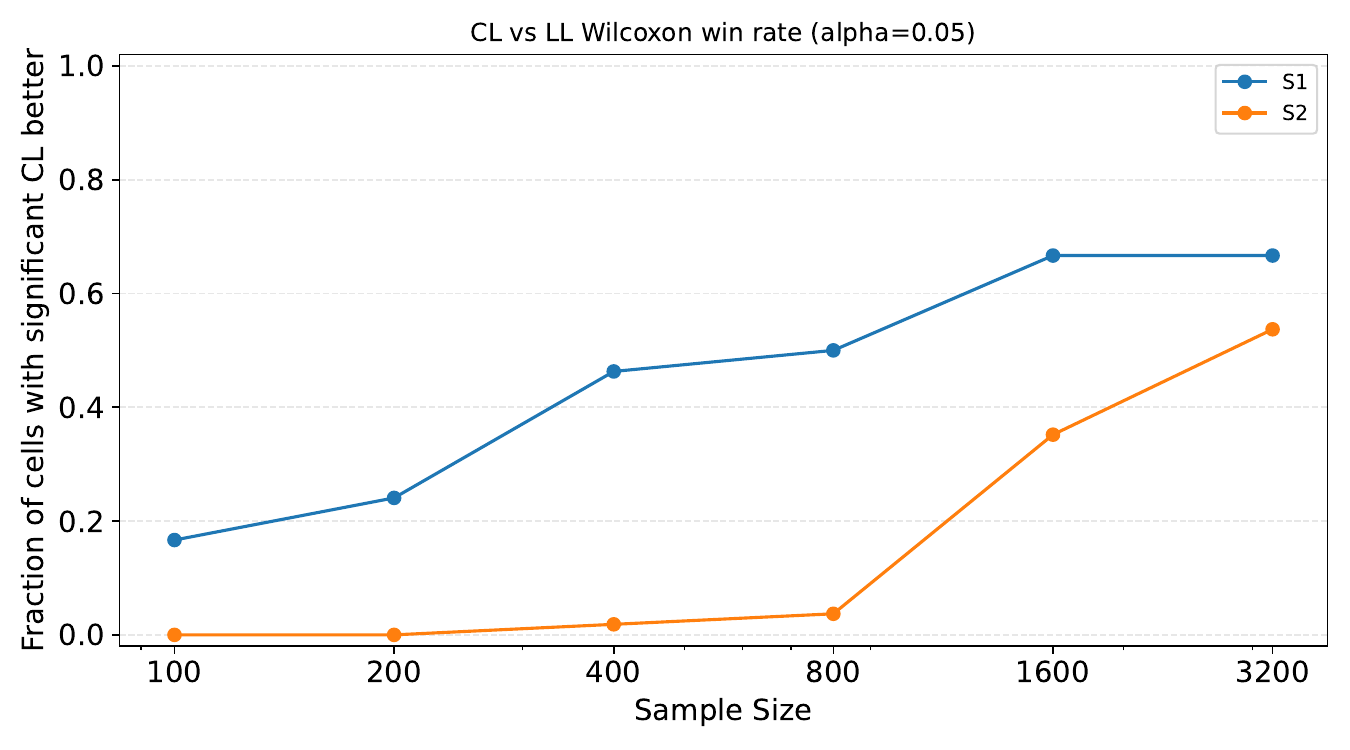}
        \caption{The significance win rate of CL over LL across simulation cells.}\label{fig:figs/wilcoxon_significance_rate.pdf}
    \end{subfigure}
    \caption{Wilcoxon signed-rank test results for CL vs LL across simulation cells.}\label{fig:wilcoxon_results}
\end{figure}

Across 648 simulation cells, 199 show nominal significance (\( p_c <0.05\)), of which 197 favoured CL over LL and 2 favoured LL over CL. The two cases favouring LL over CL are $(F8,K2,S2,n=3200), p=0.0104,\bar{d}=6.515\times 10^{-6}$and $(F8,K4,S2,n=3200), p=0.0296,\bar{d}=7.901\times 10^{-6}$. Furthermore, we apply BH-FDR at \( q=0.05 \) over 648 cells and yield 191 discoveries, of which 190 favoured CL over LL and 1 favoured LL over CL.~Results are robust to BH-FDR correction. The proportion of CL-favouring significant cells increases with sample size (from 8.33\% at $n=100$ to 60.19\% at $n=3200$), and is markedly higher under S1 than S2 (45.06\% vs 15.74\%).

\section{Real Data Analysis}\label{sec:real_data_analysis}
Magnetoencephalography (MEG) provides a non-invasive measure of the magnetic fields generated by neuronal currents in the cerebral cortex, with millisecond-scale temporal resolution, making it particularly well suited for investigating the dynamic organisation of functional brain network~\citep{schoffelenSourceConnectivityAnalysis2009,bastosTutorialReviewFunctional2016}. In resting-state MEG studies, power envelope correlation (PEC) has become a widely used approach for quantifying the synchrony of low-frequency (<1 Hz) fluctuations in the amplitudes of neural oscillations across cortical regions~\citep{liuLargescaleSpontaneousFluctuations2010,brookesInvestigatingElectrophysiologicalBasis2011,hippLargescaleCorticalCorrelation2012,oneillMeasuringElectrophysiologicalConnectivity2015}. A key finding from this line of research is that the power envelopes of band-limited neural oscillations exhibit significant and reproducible correlations between distinct brain regions. These inter-regional correlations are thought to reflect the spatial organisation of resting-state functional networks, including the default mode and visual networks~\citep{liuLargescaleSpontaneousFluctuations2010,brookesInvestigatingElectrophysiologicalBasis2011,hippLargescaleCorticalCorrelation2012}.

Traditional PEC estimates functional connectivity by computing the Pearson correlation coefficient between power-envelope time series over the entire recording period, thereby yielding a static, time-invariant measure of connectivity:

\begin{equation*}
    \rho_{rq}=\mathrm{corr}\big(e_r(t),,e_q(t)\big)
    =\frac{\sum_t\big(e_r(t)-\bar e_r\big)\big(e_q(t)-\bar e_q\big)}
    {\sqrt{\sum_t\big(e_r(t)-\bar e_r\big)^2+\sum_t\big(e_q(t)-\bar e_q\big)^2}},
\end{equation*}
where \(e_r(t)\) denotes the power envelope of the \(r\)-th ROI at time \(t\). However, functional connectivity in the resting brain is inherently dynamic, with the strength of inter-regional coupling varying over time~\citep{hippLargescaleCorticalCorrelation2012,bastosTutorialReviewFunctional2016}. The nonparametric correlation estimator CL proposed in this study extends the conventional static correlation framework by incorporating time \(t\) into the observation model, thereby representing the correlation as a smooth function \(\rho(t)\). This formulation allows the strength of functional connectivity to vary continuously over time, without imposing a specific parametric form on its temporal evolution. In the MEG application, \(\hat{\rho}(t)\) provides a nonparametric estimate of time-varying functional connectivity between pairs of ROIs.

In this case study, we use MEG data from the Cam-CAN dataset~\citep{cam-canCambridgeCentreAgeing2014} to demonstrate the practical utility of the proposed nonparametric correlation estimation method. The Cam-CAN dataset contains resting-state MEG recordings from healthy adults across a broad range of ages. The recordings were acquired at a sampling frequency of 1000 Hz and lasted 8 min 40 s, providing a rich temporal profile of ongoing neural activity.
To provide a clear illustration of the proposed method, we focused on a single participant (sub-CC110033).

To keep the computational burden manageable, we restricted the analysis to a moderate number of brain regions. Instead of the commonly used 68-region Desikan–Killiany (DK) atlas~\citep{desikanAutomatedLabelingSystem2006}, we adopted the 200-region Schaefer Local-Global Parcellation rather than its 1000-region counterpart~\citep{schoffelenSourceConnectivityAnalysis2009}. Because functional network organisation can vary substantially across frequency bands~\citep{kajimuraFrequencyspecificBrainNetwork2023,guglielmiFrequencyDependentFunctionalConnectivity2022,zinkRestingstateEEGDynamics2021,contiFrequencyNetworkSpecificChanges2026}, we followed the conventional classification of neural oscillations and extracted power envelopes separately for six frequency bands: delta (1-4 Hz), theta (4-8 Hz), alpha (8-12 Hz), beta (12-40 Hz), slow gamma (40-80 Hz), and fast gamma (80-120 Hz)~\citep{colginFrequencyGammaOscillations2009,buzsakiRhythmsBrain2011,zhengSpatialSequenceCoding2016}. The extracted power envelopes were subsequently log-transformed and used as input to the proposed nonparametric correlation estimation method.

We first describe the procedure for extracting power envelopes from the raw MEG recordings and then demonstrate how the proposed method can be applied to these real-world data.

\subsection{Power Envelope Extraction}\label{subsec:power_envelope_extraction}

We first prepare the raw data required for the analysis. Specifically, we utilize the resting-state MEG recording, which has been preprocessed using MaxFilter and the Signal Space Separation (SSS) method, alongside the corresponding empty-room recording for environmental noise estimation. The spatial alignment between the MEG sensor space and the structural MRI space is achieved using a pre-calculated MEG–MRI coregistration transformation matrix. Because the experimental recordings have already undergone the SSS procedure~\citep{tauluApplicationsSignalSpace2005,tauluSpatiotemporalSignalSpace2006}, Maxwell filtering is not repeated in the subsequent analysis pipeline.

\paragraph{Data preprocessing}\label{para:data_preprocessing}
The preprocessing pipeline consists of the following steps. (1) Channel selection and bad-channel handling. For both the experimental and empty-room recordings, we retained only MEG channels and excluded channels marked as ``bads''. We then restricted both datasets to their common set of channels to ensure consistency in the subsequent noise covariance estimation. The empty-room data were reordered to match the channel ordering of the experimental recording. (2) Band-pass filtering. To remove very slow drifts while retaining the lowest frequency range of interest (1–4 Hz), both recordings were high-pass filtered at 0.5 Hz using a zero-phase FIR filter designed with `firwin`~\citep{widmannDigitalFilterDesign2015,gramfortMEGEEGData2013}. This cutoff preserves the lowest frequencies of the delta band without introducing substantial attenuation. (3) Power line noise removal. Notch filters were applied at 50, 100, and 150 Hz to suppress power-line interference and its harmonics.

\paragraph{Source Reconstruction and Atlas Mapping  }\label{para:Source-Reconstruction-and-Atlas-Mapping  }
It consists of four main steps:
\begin{enumerate}
    \item \textbf{Anatomical processing and source space.} The T1-weighted structural MRI was processed with FreeSurfer~\citep{fischlFreeSurfer2012} to reconstruct the cortical surfaces. An ``oct6'' source-space discretisation was then used to generate a cortical source grid containing approximately 8,196 vertices, with 4,098 vertices per hemisphere.
    \item \textbf{Boundary element model and forward solution.} The geometry of the volume conductor was encapsulated by a single-layer Boundary Element Method (BEM) model. The inner skull surface boundary was automatically segmented from structural T1-weighted MRI data utilising the watershed algorithm~\citep{segonneHybridApproachSkull2004} via FreeSurfer. Inside this bounding surface, the brain tissue compartment was treated as a homogeneous and isotropic volume conductor with a standard consensus conductivity assigned at 0.3 S/m~\citep{geddesSpecificResistanceBiological1967,vorwerkGuidelineHeadVolume2014}. To establish the forward solution, the MEG sensor locations were aligned with the structural MRI coordinate space via an MEG-MRI coregistration transform based on anatomical fiducials and digitised headshape points. Subsequently, the source-to-sensor gain matrix (lead-field matrix) was computed using the single-layer BEM framework~\citep{hamalainenRealisticConductivityGeometry1989}. To prevent mathematical singularity and numerical instability in closer sensors, a minimum source-to-sensor distance boundary of 5 mm was enforced during the forward field calculation~\citep{huangVoxelwiseRestingstateMEG2014}. Both the forward solution and the source space were cached to avoid unnecessary recomputation.
    \item \textbf{Noise covariance and inverse operator.} The noise covariance matrix, \(\hat{\Sigma}_n\), was estimated from the empty-room recording. A minimum-norm inverse operator was then constructed using a loose orientation constraint of 0.2, depth weighting of 0.8, and an assumed signal-to-noise ratio (SNR) of 3~\citep{linSpectralSpatiotemporalImaging2004,linAssessingImprovingSpatial2006}. The inverse solution was computed using dynamic statistical parametric mapping (dSPM)~\citep{hamalainenInterpretingMagneticFields1994,daleDynamicStatisticalParametric2000,marinkovicSpatiotemporalDynamicsModalityspecific2003}. The regularisation parameter was determined from the assumed SNR as \( \lambda^2=\frac{1}{\mathrm{SNR}^2}=\frac{1}{9}\). After whitening the sensor-space data, dSPM obtains cortical current estimates through the minimum-norm solution and subsequently normalises them by their estimated noise sensitivity. The resulting estimates therefore have an approximately unit-variance \(t\)-statistic interpretation~\citep{daleDynamicStatisticalParametric2000,marinkovicSpatiotemporalDynamicsModalityspecific2003}.
    \item \textbf{Broadband Source Reconstruction and ROI Time Series.} The inverse operator was applied to the broadband MEG data following 0.5 Hz high-pass filtering. To constrain the inverse solution to the physiologically relevant anatomy, a cortical surface normal constraint was employed to project the estimated current vectors exclusively onto the local surface normals, reflecting the anatomical alignment of macroscopically oriented cortical pyramidal cells~\citep{daleImprovedLocalizadonCortical1993,linAssessingImprovingSpatial2006}. This produced continuous cortical source estimates \(S(\mathbf{r},t)\). The source estimates were then parcellated according to the Schaefer-200 cortical atlas~\citep{schaeferLocalGlobalParcellationHuman2018}. To extract a single representative time course for each region of interest (ROI), we applied an orientation-aligned sign-flipping averaging procedure~\citep{woolrichDynamicStateAllocation2013}.
        Specifically, the ROI-level time series was defined as \( x_r(t)=\frac{1}{n_r}\sum_{j\in\mathrm{ROI}_r}\sigma_j\hat{S}(\mathbf{r}_j,t), \) where \(n_r\) denotes the number of vertices in ROI \(r\), and \(\sigma_j\in\{-1,+1\}\) is a sign-flip factor determined by the orientation of the surface normal at vertex \(j\). This procedure addresses the complex folding geometry of the cortex; it mathematically flips the sign of time series from vertices whose local surface normals oppose the dominant orientation of the ROI, thereby preventing phase cancellation and ensuring an accurate, robust representation of regional population activity~\citep{woolrichDynamicStateAllocation2013,ahlforsSensitivityMEGEEG2010}.
\end{enumerate}
\paragraph{Source Leakage Correction}\label{para:source_leakage_correction}
Because the forward-model point-spread functions of neighbouring sources overlap on the cortex, the time course reconstructed for a given parcel contains contributions from the surrounding tissue, an effect known as spatial leakage. Leakage inflates and distorts inter-regional correlation estimates, and a standard remedy in MEG connectomics is to whiten the parcel-level signals by symmetric orthogonalization~\citep{colcloughSymmetricMultivariateLeakage2015}. Let $\mathbf{X}$ denote the $R\times T$ matrix of band-limited ROI time courses, with the $R=200$ rows corresponding to the Schaefer parcels and each row centred over time. The symmetric orthogonalisation applies the inverse matrix square root of the row covariance,
\begin{equation}\label{eq:lowdin_full}
    \mathbf{X}_{\mathrm{orth}}=\bigl(\mathbf{X}\mathbf{X}^{\top}\bigr)^{-\frac{1}{2}}\mathbf{X},
\end{equation}
which, for full-rank $\mathbf{X}$, decorrelates the ROI signals while rotating them as little as possible, in a least-squares sense, from the original time courses.

The full-rank construction does not apply to the present data. The 200 parcels are reconstructed from MEG measurements whose effective rank, determined by the empty-room noise covariance used for whitening, is only about 73, and the row-standardised band-limited ROI matrices show a pronounced singular-value gap: in the delta band, for instance, the relative spectrum falls from $3.7\times10^{-2}$ to $2.0\times10^{-7}$ between the 70th and 71st singular values. The row covariance $\mathbf{X}\mathbf{X}^{\top}$ is therefore numerically singular, and the inverse square root in~\eqref{eq:lowdin_full} is unstable. We instead used a rank-aware, rank-truncated L\"owdin orthogonalisation. Truncating the eigenvalue spectrum to 73 components redistributes the variance associated with the discarded dimensions among the retained components, thereby inflating the retained eigenvalues and potentially altering the magnitude of the leakage correction~\citep{colcloughSymmetricMultivariateLeakage2015}. Writing the singular value decomposition of the row-centred, row-standardised ROI matrix as $\mathbf{X}=\mathbf{U}\mathbf{S}\mathbf{V}^{\top}$ with $s_{1}\geq s_{2}\geq\cdots\geq s_{R}$, and denoting by $\mathbf{U}_{r}$, $\mathbf{S}_{r}$ and $\mathbf{V}_{r}$ the leading $r$ singular vectors and values, the leakage-corrected signals are
\begin{equation}\label{eq:lowdin_trunc}
    \mathbf{X}_{\mathrm{orth}}=\mathbf{U}_{r}\mathbf{V}_{r}^{\top}
    =\bigl(\mathbf{U}_{r}\mathbf{S}_{r}^{2}\mathbf{U}_{r}^{\top}\bigr)^{-\frac{1}{2}}\mathbf{X},
\end{equation}
the Moore--Penrose analogue of~\eqref{eq:lowdin_full} restricted to the retained subspace. All 200 rows are kept; the truncation confines the corrected signals to an $r$-dimensional subspace in which they are mutually uncorrelated, so that the row Gram matrix $\mathbf{X}_{\mathrm{orth}}\mathbf{X}_{\mathrm{orth}}^{\top}$ is a rank-$r$ projector rather than the identity. The correction is therefore best described as a rank-deficient modification of the standard symmetric orthogonalization, not as an ordinary 200-dimensional orthogonalization. Because the truncation does not preserve row norms, each corrected row was rescaled to restore the root-mean-square amplitude of the original ROI signal, so that the envelopes computed below retain the between-region amplitude structure of the source estimates.

The effective rank was selected separately for each frequency band from the spectrum of the row-standardised band-limited ROI matrix. We computed the adjacent singular-value ratios $s_{i}/s_{i+1}$ and selected the rank $r$ at which $s_{r}/s_{r+1}$ is largest, restricting the search to $2\leq r\leq 200$ and accepting the candidate only when this ratio exceeded 10. When no gap satisfied this criterion, we fell back to the smallest rank whose cumulative explained variance $\sum_{i=1}^{r}s_{i}^{2}/\sum_{i=1}^{R}s_{i}^{2}$ reached 0.999. The full singular spectrum, the selected rank, the explained variance and the gap ratio were saved for every band as diagnostics alongside the envelope data.

Power envelopes were derived from the leakage-corrected, band-limited signals by the Hilbert transform. For each ROI, the analytic signal was formed and the envelope taken as its modulus; with the power option, the envelope is the squared modulus, so that for a band-limited signal $x(t)$ the envelope is $|x(t)+i\,\mathcal{H}[x(t)]|^{2}$, where $\mathcal{H}$ denotes the Hilbert transform~\citep{brunsFourierHilbertWaveletbased2004}. The envelope represents the slowly varying amplitude of the oscillation, and the coupling of these slow amplitude fluctuations between regions is exactly the quantity that power envelope correlation is designed to measure \citep{liuLargescaleSpontaneousFluctuations2010,brookesInvestigatingElectrophysiologicalBasis2011,hippLargescaleCorticalCorrelation2012}. Two processing choices were made before the envelopes entered the connectivity analysis. The envelope was low-pass filtered at 1 Hz with a zero-phase FIR filter (firwin) \citep{widmannDigitalFilterDesign2015}. The 1 Hz cutoff retains the sub-hertz amplitude modulations that carry the resting-state network structure \citep{hippLargescaleCorticalCorrelation2012} and removes the faster fluctuations inherited from the carrier band, which would otherwise inflate the correlations; the zero-phase design preserves the timing of the modulations, on which the correlation estimates directly depend. The filtered envelope was then resampled to 10 Hz. Because the filtered envelope is effectively band-limited to 1 Hz, a 10 Hz output rate leaves ample margin above the Nyquist frequency while shortening each ROI time series from 574,000 to 5,740 samples, a reduction that matters given that correlation is subsequently computed for all 19,900 ROI pairs. These settings, the power option, the 1 Hz cutoff and the 10 Hz output rate, follow the conventions of the power envelope correlation literature \citep{brookesInvestigatingElectrophysiologicalBasis2011,hippLargescaleCorticalCorrelation2012}.

Power envelopes were extracted from dSPM source estimates, and therefore represent squared noise-normalized source amplitudes rather than absolute source power. The envelope matrix returned by this step, one row per ROI sampled at 10 Hz, is the variable that enters the correlation calculation. After a log transform, each pair of rows supplies the two time series from which the time-varying correlation $\hat{\rho}(t)$ is estimated for that ROI pair.

\subsection{Pairwise Correlation Estimation}\label{subsec:pairwise_correlation_estimation}
The time-varying correlation $\hat{\rho}(u_0)$ is defined only after a local
bandwidth $h$ has been fixed, and the temporal scale of the amplitude
modulations that carry the resting-state networks differs between ROI pairs and
frequency bands~\citep{hippLargescaleCorticalCorrelation2012,
brookesInvestigatingElectrophysiologicalBasis2011}. We therefore selected $h$
separately for each pair by leave-local-block-out cross-validation rather than
imposing a single global scale. For a pair of envelope time series
$(x_{i1},x_{i2})$ on a common grid $t_1<\cdots<t_T$, the local statistics at
$u_0$ were formed with a Gaussian kernel $K_h(t-u_0)$ from which a central
exclusion gap of 20 ms had been removed, so that every validation point is
predicted from its temporal neighbourhood rather than from its immediate
vicinity. The choice of 20 ms is consistent with the temporal resolution commonly used in MEG connectivity analyses, where 20 ms sampling intervals or temporal lags have been employed~\citep{depasqualeDynamicCoreNetwork2016,bettiTopologyFunctionalConnectivity2018}. The parameter $A_n(u_0)$, the local weighted mean of $x_1^2+x_2^2$,
was estimated by the Nadaraya--Watson smoother, whereas $B_n(u_0)$, the local
cross-moment of $x_1x_2$, was estimated by the local-linear fit
\begin{equation}\label{eq:cl_B}
    B_n^{\mathrm{CL}}(u_0)
    =\frac{\hat S_2(u_0)\hat T_0(u_0)-\hat S_1(u_0)\hat T_1(u_0)}
    {\hat S_2(u_0)\hat S_0(u_0)-\hat S_1(u_0)^2},
\end{equation}
where $\hat S_j(u_0)=\sum_i K_h(t_i-u_0)(t_i-u_0)^j$ and
$\hat T_j(u_0)=\sum_i K_h(t_i-u_0)(t_i-u_0)^j x_{i1}x_{i2}$. For a symmetric
kernel the first moment $\hat S_1$ vanishes in the interior and~\eqref{eq:cl_B}
coincides with the Nadaraya--Watson estimate; near the edges of the recording
the kernel is truncated, $\hat S_1(u_0)\neq0$, and the local-linear form removes
the resulting first-order boundary bias in the manner of local-polynomial
smoothing~\citep{fanLocalPolynomialModelling1996,
fanDesignadaptiveNonparametricRegression1992,fanEfficientEstimationConditional1998}. The correlation estimate
$\hat{\rho}(u_0)$ was taken as the real root in $(-1,1)$ of the cubic
$\rho^3-B_n^{\mathrm{CL}}(u_0)\rho^2+(A_n(u_0)-1)\rho-B_n^{\mathrm{CL}}(u_0)=0$,
solved by a few Newton iterations. Following established grid-based bandwidth selection strategies~\citep{staudenmayerLocalPolynomialRegression2004,tholkageConditionalKaplanMeier2022}, we employed a coarse-to-fine search scheme as described below. The bandwidth was chosen by minimising the
predictive Gaussian negative log-likelihood evaluated at a sparsely sampled set
of validation times, first over a coarse logarithmic grid of candidate
bandwidths and then by a pair-specific fine search centred on the coarse
winner; the Gaussian kernel was truncated at four standard deviations, and the
candidate bandwidths spanned the physiologically relevant scales of the
sub-hertz envelope fluctuations.

For each frequency band, we estimated the nonparametric correlation for all 19,900 unique ROI pairs and selected a pair-specific bandwidth for each pair using the CL method. The resulting bandwidths were then used to estimate the time-varying correlation between each pair of ROIs. Specifically, for frequency band \(f\), we denote the estimated time-varying correlation between ROIs \(i\) and \(j\) at time \(t\) by  \( \hat{\rho}_{ij}^{f}(t), i=1,\ldots,200, j=i+1,\ldots,200\), where \(f\) indexes the frequency band. To save the memory, we only store the upper triangular part of the correlation matrix. The final output is a correlation coefficient matrix of size \(19,900\times 5740\times 6\), where 19,900 is the number of ROI pairs, 5,740 is the number of time points, and 6 is the number of frequency bands. For each frequency band and each time point, the correlation coefficients can be arranged into a \(200\times 200\) symmetric matrix. However, this symmetric correlation matrix is not guaranteed to be positive definite. We applied the method proposed by~\citet{highamComputingNearestCorrelation2002} to project it onto the nearest positive definite matrix while ensuring that the diagonal elements are 1 and the off-diagonal elements are between -1 and 1. The variance of eigenvalues of the correlation matrix is proportional to the mean linkage index~\citep{durandLinkageIndexVariables2017}. Therefore, we computed the variance of eigenvalues for each frequency band and each time point's correlation matrix and used it as the mean linkage index for that frequency band and time point. The final output is an eigenvalue variance matrix of size \(5740\times 6\). To better visualize the connectivity between the 200 ROIs, we grouped them into 7 networks based on Yeo-7 network parcellation and computed the average linkage index between each network pair. Since correlation coefficients can be both positive and negative, we applied Fisher z-transformation before averaging and then applied the inverse Fisher z-transformation to obtain the final correlation matrix. To clearly show the Yeo-7 networks' connectivity, we truncate the \(7\times 7 \) matrix by its median.  The final output is a mean linkage index matrix of size \(7\times 7\times 5740\times 6\), where 7 is the number of Yeo-7 networks, 5,740 is the number of time points, and 6 is the number of frequency bands. For each frequency band, we selected the instantaneous correlation matrices at time points 100, 1100, 2100, 3100, 4100, and 5100 and marked them on the eigenvalue variance line plot, as shown in~\Cref{fig:figs/eig_var_circle_connect_Delta.pdf,fig:figs/eig_var_circle_connect_Fast_gamma.pdf}.

\begin{figure}[!ht]
    \centering
    \includegraphics[width=\textwidth]{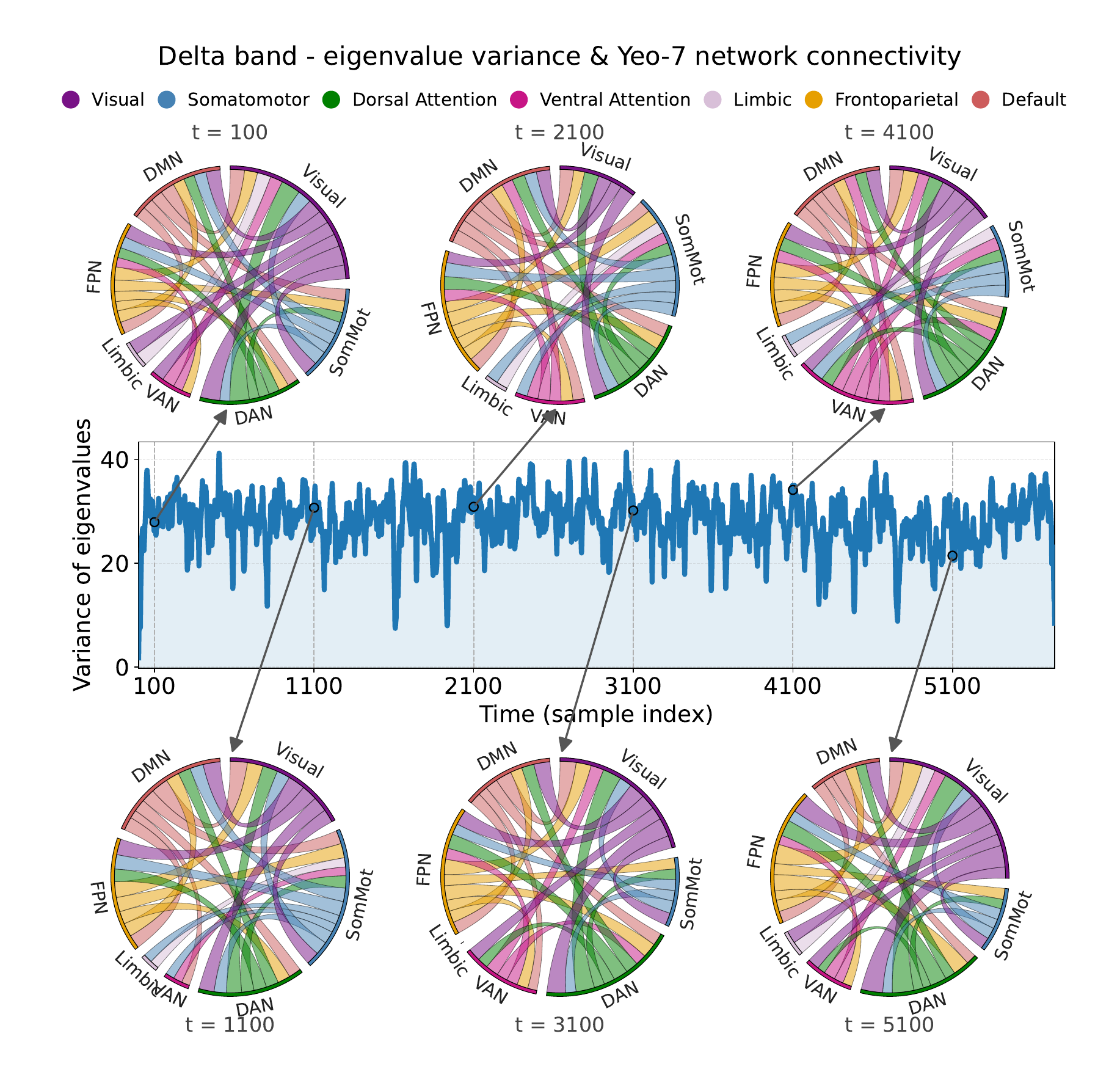}
    \caption{The eigenvalue variance of the dynamic correlation matrix for the Delta band. At the time (sample index) 100,1100,2100,3100,4100,5100, we draw the chord connectivity of Yeo-7 networks. As the syncitical \( 7\times7 \) network is still dense, we truncated the Yeo-7 networks by its median value. The names of Yeo-7 networks are ``Visual'',``Somatomotor'',``Dorsal Attention'',``Ventral Attention'',``Limbic'',``Frontoparietal'' and ``Default'', while the short names on the chord connectivity are ``Visual'',``SomMot'',``DAN'',``VAN'',``Limbic'',``FPN'' and ``DMN'' respectively.}\label{fig:figs/eig_var_circle_connect_Delta.pdf}
\end{figure}
\begin{figure}[!ht]
    \centering
    \includegraphics[width=\textwidth]{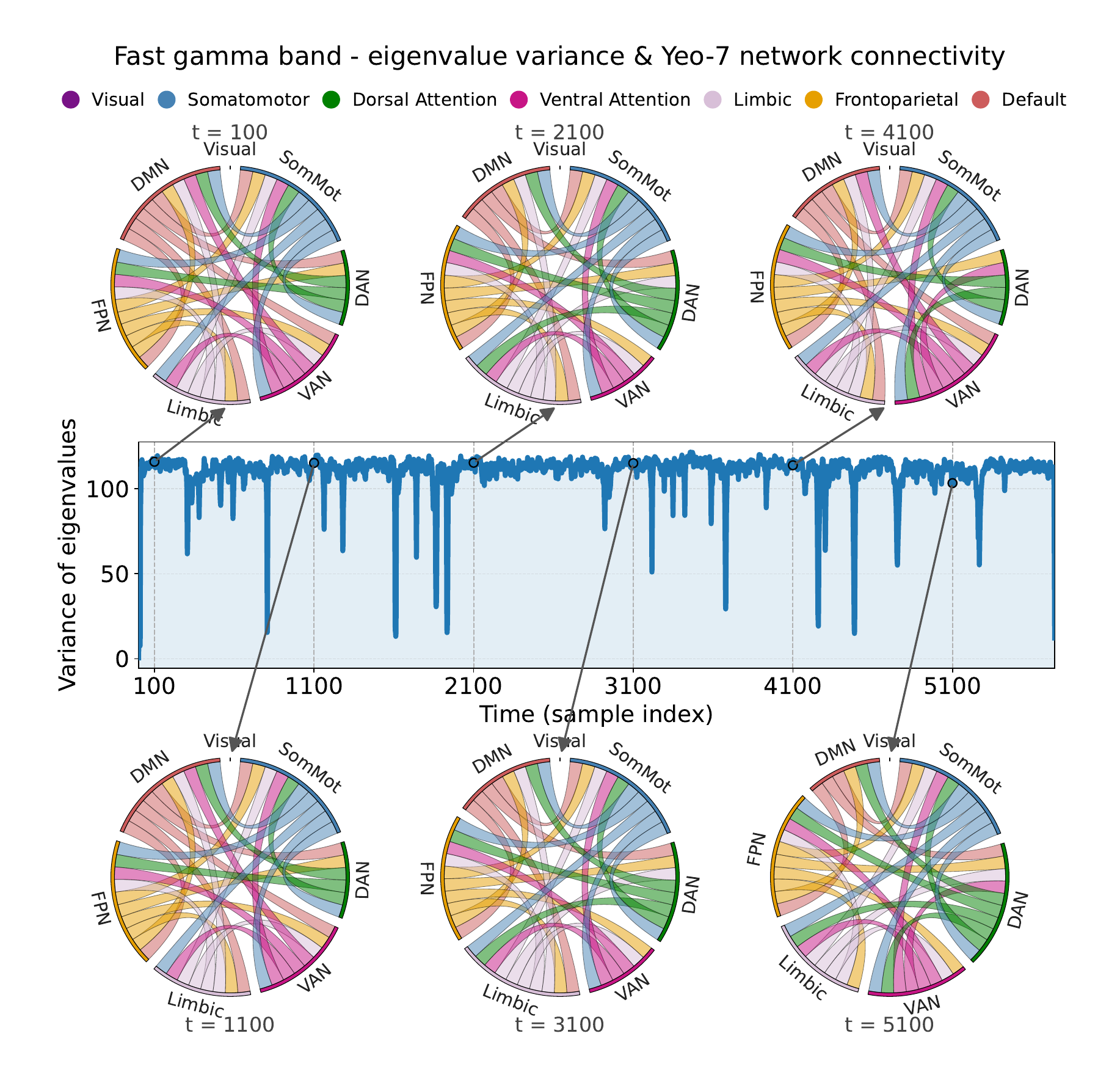}
    \caption{The eigenvalue variance of the dynamic correlation matrix for the fast Gamma band. At the time (sample index) 100,1100,2100,3100,4100,5100, we draw the chord connectivity of Yeo-7 networks. As the syncitical \( 7\times7 \) network is still dense, we truncated the Yeo-7 networks by its median value. The names of Yeo-7 networks are ``Visual'',``Somatomotor'',``Dorsal Attention'',``Ventral Attention'',``Limbic'',``Frontoparietal'' and ``Default'', while the short names on the chord connectivity are ``Visual'',``SomMot'',``DAN'',``VAN'',``Limbic'',``FPN'' and ``DMN'' respectively.}\label{fig:figs/eig_var_circle_connect_Fast_gamma.pdf}
\end{figure}

From~\Cref{fig:figs/eig_var_circle_connect_Delta.pdf,fig:figs/eig_var_circle_connect_Fast_gamma.pdf}, we can see that the variance of the eigenvalue of delta band fluctuates substantially over time, with values generally falling in the range of approximately 26-32 and frequent short-lived excursions. By contrast, the fast gamma band shows a considerably larger absolute eigenvalue variance, typically around 110-115, and remains comparatively stable for most of the recording, although multiple transient downward fluctuations are observable. Furthermore, the eigenvalue variance of theta, alpha, beta, and slow gamma bands fluctuates approximately between 67-75, 62-70, 47-59, and 106-111, respectively, see~\Cref{fig:figs/eig_var_circle_connect_Beta.pdf,fig:figs/eig_var_circle_connect_Alpha.pdf,fig:figs/eig_var_circle_connect_Theta.pdf,fig:figs/eig_var_circle_connect_Slow_gamma.pdf} in Appendix. The six frequency bands reveal a clear frequency-dependent pattern in the eigenvalue variance, with the gamma (fast and slow) band exhibiting the highest values, followed by the theta, alpha, beta and delta bands. But the connectivity does not become more stable as the frequency increases. The beta band shows markedly greater temporal variability, with a 97.31 variance of ``eigenvalue variance'' across time, than the theta and alpha bands (55.00 and 50.08, respectively), despite its intermediate position in the frequency spectrum.

For the circle connectivity plots, we can see that the delta band Yeo-7 circle connectivity shows a similar pattern across the six time indices. The relative strengths of connections involving the Visual, Somatomotor, Dorsal Attention (DAN), Ventral Attention (VAN), Frontoparietal (FPN), and Default Mode (DMN) networks vary across the six sampled time points. The Limbic and Ventral Attention (VAN) networks show relatively weaker connections compared to the other networks. However, compared with the delta band, the Fast Gamma band Yeo-7 circle connectivity shows a different pattern across the six time indices. The visual network shows relatively weaker connections compared to the other networks (in~\Cref{fig:figs/eig_var_circle_connect_Fast_gamma.pdf}, the circle connectivity is truncated by its median value). These phenomena appear in the slow gamma band as well; see~\Cref{fig:figs/eig_var_circle_connect_Slow_gamma.pdf} in the Appendix.  For the alpha band in~\Cref{fig:figs/eig_var_circle_connect_Alpha.pdf}, the visual network has strong connections with other networks, which is consistent with the conclusion in~\citet{colcloughHowReliableAre2016}. Overall, the Yeo-7 network connectivity patterns exhibit frequency-dependent variations, with different networks showing varying strengths of connections across the six sampled time points for each frequency band.

\textbf{Periodicity Analysis} To assess whether the observed temporal variations in functional connectivity exhibited significant periodicity, we applied an autoregressive (AR)-based global periodicity test to each network-pair connectivity time series. This procedure was designed to account for the substantial temporal dependence expected in dynamic functional connectivity~\citep{xuSeasonalVariationsFunctional2023,eklundFastRandomPermutation2011,choeReproducibilityTemporalStructure2015,liegeoisInterpretingTemporalFluctuations2017,vanesLargescaleCorticalFunctional2025}. For each frequency band and each pair of Yeo-7 functional networks, the resulting connectivity time series was first centred and standardised. An AR($p$) model was then fitted to the standardised series \( y_{t} \),
\begin{equation*}
    y_t=c+\sum_{\ell=1}^{p}\phi_\ell y_{t-\ell}+\varepsilon_t,
\end{equation*}
where the autoregressive order $p$ was selected by minimising the Akaike information criterion (AIC) over $p=1,\ldots,20$. The fitted autoregressive coefficients were subsequently used to whiten both the observed connectivity series and the design matrices.

To test for periodic components over a prespecified range of connectivity fluctuation frequencies, we considered a grid of 500 equally spaced frequencies within the interval corresponding to periods of 2--200 s, i.e. $f\in[1/200,1/2]$ Hz. For each candidate frequency $f$, we compared a null model containing only an intercept with a harmonic model,
\begin{equation*}
    X_1(f)=
    \begin{bmatrix}
        \mathbf{1} & \cos(2\pi f t) & \sin(2\pi f t)
    \end{bmatrix},
    \qquad
    X_0=\mathbf{1}.
\end{equation*}
Let \( L \in \mathbb{R}^{n_{eff}\times T} \) denote the whitening matrix, where \( n_{eff} \) is the effective number of observations after accounting for the AR model and \( T \) is the total number of time points. Hence, the whitened response and design matrix are:
\begin{equation*}
    \tilde{y}=Ly\in \mathbb{R}^{n_{eff}},
    \qquad
    \tilde{X}_1(f)=LX_1(f)\in \mathbb{R}^{n_{eff}\times 3},
    \qquad
    \tilde{X}_0=LX_0\in \mathbb{R}^{n_{eff}\times 1}.
\end{equation*}
After AR whitening, the two nested models \( \tilde{y}\sim \tilde{X}_1(f) \) and \( \tilde{y}\sim \tilde{X}_0 \) were fitted by least squares and compared using the statistic
\begin{equation*}
    F(f)=
    \frac{[RSS_0(f)-RSS_1(f)]/2}
    {RSS_1(f)/(n_{\mathrm{eff}}-3)},
\end{equation*}
which, under the corresponding model assumptions, was evaluated against an $F_{2,n_{\mathrm{eff}}-3}$ distribution. The frequency providing the strongest evidence for a periodic component was defined as \( \widehat f_{\mathrm{AR}}=\argmin_f p(f)\), \(p_{\mathrm{AR}}=\min_f p(f), \) where \( p(f) \) is the $p$-value corresponding to the F-statistic \( F(f) \).
The corresponding dominant period was calculated as $\widehat T_{\mathrm{AR}}=1/\widehat f_{\mathrm{AR}}$. Finally, the resulting $p_{\mathrm{AR}}$ values were adjusted for multiple comparisons across all frequency-band and network-pair combinations using the Benjamini--Hochberg false discovery rate (FDR) procedure, yielding $q_{\mathrm{AR}}$ values. Connectivity fluctuations were considered statistically significant when $q_{\mathrm{AR}}<0.005$.
\begin{figure}[ht]
    \centering
    \includegraphics[width=\textwidth]{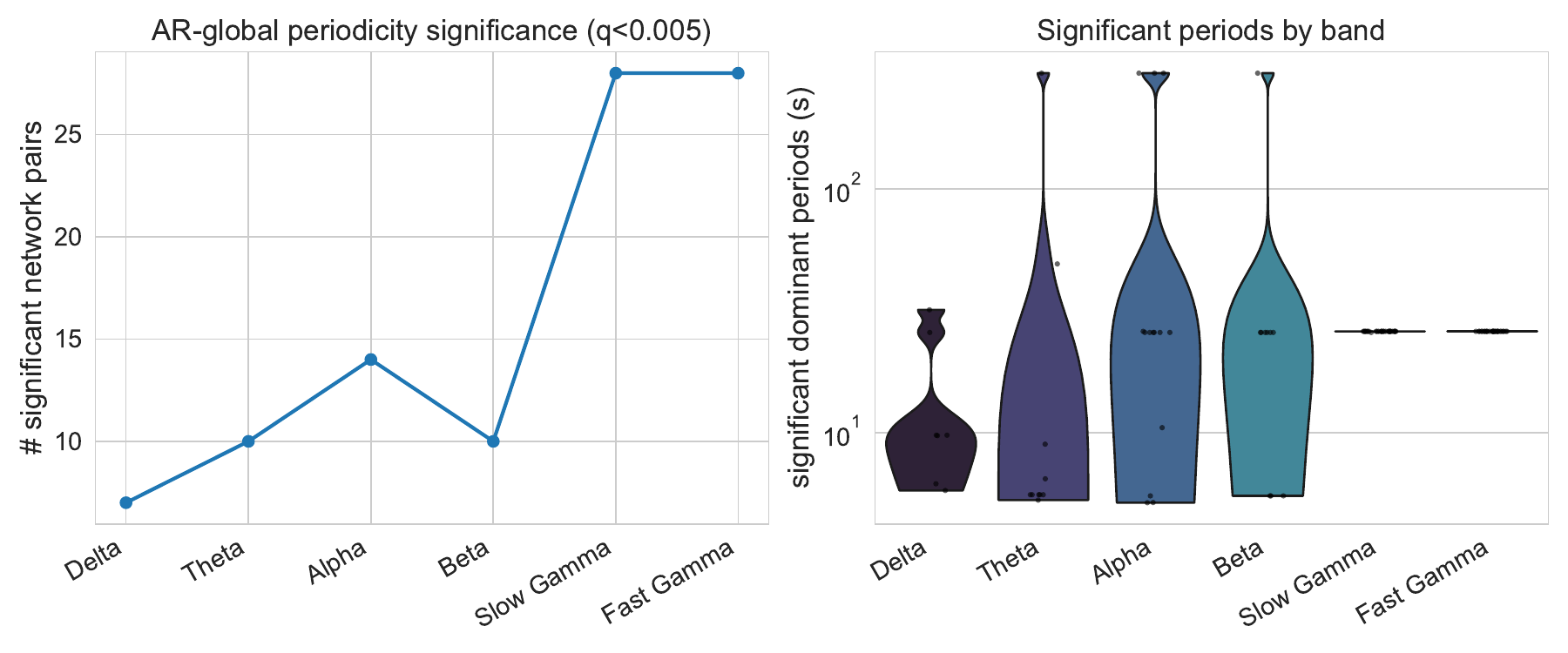}
    \caption{Multiband periodicity analysis of dynamic functional connectivity. Left: number of Yeo-7 network pairs (out of 28) whose connectivity-fluctuation spectrum contains a significant periodic component at the AR-global level \( q_{\mathrm{AR}}<0.005 \), Benjamini–Hochberg FDR across all 168 band–pair combinations. Right: violin plots of the estimated dominant period \( \widehat T_{\mathrm{AR}} \), restricted to significant pairs, overlaid with individual estimates as jittered points; the ordinate is logarithmic. The slow- and fast-gamma distributions are degenerate, because in those bands every pair is significant (28/28) and the estimated periods coincide almost exactly \( \approx 26 s \).}\label{fig:figs/multiband_significant_periods.pdf}
\end{figure}

The results are illustrated in~\Cref{fig:figs/multiband_significant_periods.pdf}. A total of 97 of the 168 band–pair combinations (57.7\%) exhibited a statistically significant periodic component at \( q_{\mathrm{AR}}<0.005 \). The number of significant pairs differed markedly across bands (left-hand panel): 7/28 (25.0\%) in delta, 10/28 (35.7\%) in theta, 14/28 (50.0\%) in alpha, 10/28 (35.7\%) in beta, and 28/28 (100\%) in both slow gamma and fast gamma. Significance therefore did not increase monotonically with carrier frequency: alpha produced the largest count outside the gamma range (14 pairs), outnumbering theta and beta (10 each), whilst delta yielded the fewest (7).

Theta, alpha and beta cluster at 5.55 s, 25.77 s and 25.77 s with the range 5.27-300, 5.14-300,  5.49-300 respectively. Delta is the most dispersed of the low-frequency bands, with three pairs near 9.74 s and the remainder scattered between 5.77 s and 31.93 s. The most striking feature is the behaviour of the gamma bands. Every one of the 28 Yeo-7 pairs was significant in both slow gamma (40-80 Hz) and fast gamma (80-120 Hz). The slow and fast gamma violins collapse into single horizontal lines at approximately 26 s, in marked contrast to the broad, multi-modal distributions of the other four bands.
The gamma band’s role in shaping large-scale functional connectivity appears to be mediated by slow, periodic fluctuations in its power envelope rather than by any intrinsic fast rhythm. Spontaneous infra-slow (<0.1 Hz) co-modulations of gamma band power have been shown to reflect the functional organisation of the cortex, with such co-modulations identifying sensory, motor, and default-mode networks and supporting the hypothesis that they represent the neurophysiological basis underlying resting-state networks~\citep{koIdentifyingFunctionalNetworks2013}. Converging EEG-fMRI and computational evidence indicates that infra-slow gamma band co-modulation, driven by balanced excitatory-inhibitory gamma oscillations, tracks fMRI connectivity dynamics across the connectome, suggesting that a common gamma band slow period may serve as a shared temporal scaffold for large-scale functional connectivity~\citep{koQuasiperiodicFluctuationsDefault2011,wirsichConcurrentEEGFMRIderived2020}.
\section{Conclusion and Discussion}\label{sec:conclusion}
In this paper, we proposed a nonparametric correlation estimator via solving a cubic equation (CE) to carry out the time-varying correlation analysis in MEG power envelope. Furthermore, we also corrected the boundary effects by implementation of the local linear regression (CL) method. The proposed estimator is a nonparametric estimator that does not require any distributional assumptions on the data, making it more flexible and robust in practice. We not only provided the theoretical properties of the proposed estimators (CE and CL), but also conducted extensive simulation studies to demonstrate their superior performance over the existing methods. Considering the trade-off between MSE and computational cost, we recommend using the CL method for large-scale time-varying correlation analysis, as the CL method reduces computational cost relative to the LL method, while simultaneously preserving a similar level of MSE.

We applied the CL method to one subject's real MEG scan from the Cam-CAN dataset. Following a series of preprocessing steps, including channel selection, band-pass filtering, power-line noise removal, anatomical processing, boundary element modelling, and computation of the forward solution, we obtained the source-space time series for the entire recording. These source estimates were subsequently parcellated according to the Schaefer-200 cortical atlas to obtain regional time series for the 200 cortical ROIs. Before extracting the power envelopes, we applied source leakage correction following established procedures in the literature.

Once the power envelopes had been obtained, we applied the proposed CL method to the time series of the 200 ROIs to estimate the time-varying correlation coefficients and construct dynamic correlation matrices for each frequency band. Analysis of these dynamic correlation matrices revealed a clear frequency-dependent pattern in the overall strength of functional connectivity, as reflected by the variance of the eigenvalues of the correlation matrices. Overall, the strength of the correlation structure tended to increase with increasing frequency band except the beta band.

The Yeo-7 networks also exhibited distinct connectivity patterns across frequency bands, with both the strength and distribution of connectivity varying across the six sampled time points within each band. These results demonstrate that the proposed nonparametric approach can capture frequency-specific and temporally varying patterns of functional connectivity in resting-state MEG data. More broadly, the proposed framework provides a flexible nonparametric approach for characterising temporally evolving dependence structures without imposing a restrictive parametric form on the underlying correlation dynamics.

\paragraph{Limitations}\label{para:limitations}
The proposed methodology also has several limitations that should be acknowledged. First, the proposed framework assumes that each pair of variables follows a bivariate normal distribution. This is a relatively strong assumption. Although a logarithmic transformation is applied before correlation estimation, we do not perform an additional formal assessment of normality. In practical applications, we recommend that the distributional assumptions be carefully assessed, and the normality of the transformed data be examined where appropriate, before applying the proposed method. If the normality test fails on the transformed data, a further discussion on transformation for heavy tailed data should be considered.  We also acknowledge that the current theoretical analysis is developed under an i.i.d. framework with unit marginal variances. These assumptions are primarily adopted to make the asymptotic analysis of the cubic estimating equation tractable and to isolate the statistical properties of the proposed nonparametric correlation estimator. Some non-i.i.d cases were fully discussed in~\citet{zhangFactorizedEstimationHighdimensional2022,liStatisticalInferenceHighdimensional2021}. The theoretical results are based on an i.i.d. unit variance setting and do not fully account for temporal dependence, bandwidth estimation, and multiple pairwise estimation effects, which we leave as important extensions for future work.

Second, because the proposed correlation estimator is constructed independently for each pair of variables, the resulting correlation matrix is not guaranteed to be positive definite. Enforcing positive definiteness at the estimation stage would require the same bandwidth to be used across all matrix entries. Although such a strategy could preserve the positive-definite structure of the estimated correlation matrix, it would introduce other important drawbacks. In particular, the computational cost would increase substantially because the bandwidth would need to be jointly selected across all variable pairs, and the resulting bandwidth could be dominated by the sparsest or most difficult-to-estimate pairs~\citep{liStatisticalInferenceHighdimensional2021}. We therefore adopt a compromise in which the correlations are estimated independently for each pair, allowing the pairwise estimation problems to be readily parallelised, and the resulting matrix is subsequently projected onto the nearest correlation matrix. While this strategy provides a practical balance between computational efficiency and matrix validity, the projection step inevitably introduces some distortion to the estimated pairwise correlations~\citep{highamComputingNearestCorrelation2002}. The extent to which this projection affects downstream statistical inference and the preservation of the underlying dependence structure has not been fully investigated in the present study and warrants further investigation.

Third, the proposed nonparametric correlation estimator is based on a local smoothing framework and may therefore become less stable when the sample size is small. This limitation is less consequential for high-frequency data, where numerous observations are typically available, such as financial transaction data, MEG, and EEG recordings.

Fourth, the bandwidth selection procedure employed in this study is based on cross-validation. Although we have introduced several computational optimisations to improve its efficiency, bandwidth selection may still impose a considerable computational burden, particularly for large-scale datasets involving numerous variable pairs or observations.

Fifth, in terms of estimation accuracy, the CL method remains slightly less accurate than the LL method in some settings. This represents a trade-off between computational efficiency and statistical accuracy that should be considered when selecting an estimator for a particular application.

Finally, the proposed nonparametric correlation estimator was developed for continuously distributed data. Its direct application to discrete-valued data is therefore not currently supported and would require appropriate modifications to the estimation procedure. This limitation is inherent to the current formulation of the estimator and represents an important direction for future methodological development.

\paragraph{Scope, Generalisability, and Future Directions}\label{para:scope_generalisability_and_future_directions}
In this study, we focused on time-varying correlation analysis based on power envelopes, as this framework has been widely used in neuroscience and provides a useful measure of functional interactions between brain regions. Nevertheless, the proposed nonparametric correlation estimation method is not restricted to power envelope connectivity and can, in principle, be extended to other forms of time-varying dependence analysis. For example, it could be applied to time-varying measures of phase synchronisation, coherence, phase-locking value (PLV), weighted phase-lag index (wPLI), imaginary coherence, and other measures of frequency-specific functional connectivity~\citep{liuBenchmarkingMethodsMapping2025}. It could also be adapted to time-varying amplitude-envelope connectivity and cross-frequency coupling measures. Future work could investigate these extensions to assess the general applicability and robustness of the proposed methodology across different representations of dynamic functional connectivity.

We used the variance of the eigenvalues as a summary measure because it provides a compact characterisation of the overall connectivity strength of a time-varying correlation matrix. Other graph-theoretic and matrix-based measures, such as mean connectivity strength, mean clustering coefficient, average path length, network density, and global efficiency, could also be used to provide a more comprehensive characterisation of the temporal evolution of functional connectivity. Owing to space limitations, we do not report these additional analyses in detail here. They nevertheless provide natural extensions for future analyses and could offer complementary perspectives on the temporal organisation of the estimated correlation matrices.

The analysis of the time-varying correlation matrices across the six frequency bands revealed clear frequency-dependent differences in the overall connectivity structure. At the same time, some similarities in connectivity patterns were also observed across different frequency bands, suggesting that frequency-specific network organisation may coexist with common features of large-scale functional architecture. These observations are consistent with the view that functional brain networks exhibit both frequency-specific and frequency-general characteristics. However, the present analysis was based on MEG data from a single participant, and the observed patterns should therefore not be interpreted as representative of the general population. Analyses based on a substantially larger cohort will be required to assess the robustness, reproducibility, and generalisability of these findings, as well as to determine whether the observed frequency-specific patterns are consistent across individuals.

Finally, the MEG data analysed in this study were acquired during the resting state. An important direction for future research would therefore be to apply the proposed methodology to task-based MEG or EEG recordings. Such analyses could provide an opportunity to investigate how functional connectivity between brain regions changes dynamically in response to different cognitive demands and whether the proposed nonparametric framework can capture task-related transitions in functional network organisation. Extending the methodology to task-evoked and other experimentally controlled settings would also provide a further test of its ability to characterise dynamic dependence structures under different neural and behavioural conditions.

\section*{Data and Code Availability}
The single subject data comes from the \href{https://cam-can.mrc-cbu.cam.ac.uk/dataset/}{Cambridge Centre for Ageing and Neuroscience (Cam-CAN) dataset}~\citep{cam-canCambridgeCentreAgeing2014}. The code and intermediate data are available from \href{https://github.com/Jieli12/nonparametric_correlation_estimator_cubic}{GitHub repository}.

\section*{Competing interests}
The authors have no competing interests.

\bibliographystyle{chicago}
\bibliography{NonparametricCorrelationEstimator.bib}

\begin{thebibliography}{}

\bibitem[\protect\citeauthoryear{Ahlfors}{Ahlfors}{2007}]{ahlforsComplexAnalysisIntroduction2007}
Ahlfors, L.~V. (2007).
\newblock {\em Complex Analysis: An Introduction to the Theory on Analytic Functions of One Complex Variable\/} (3 ed.).
\newblock International Series in Pure and Applied Mathematics. New York: McGraw-Hill.

\bibitem[\protect\citeauthoryear{Ahlfors, Han, Belliveau, and H{\"a}m{\"a}l{\"a}inen}{Ahlfors et~al.}{2010}]{ahlforsSensitivityMEGEEG2010}
Ahlfors, S.~P., J.~Han, J.~W. Belliveau, and M.~S. H{\"a}m{\"a}l{\"a}inen (2010).
\newblock Sensitivity of {{MEG}} and {{EEG}} to source orientation.
\newblock {\em Brain Topography\/}~{\em 23\/}(3), 227--232.

\bibitem[\protect\citeauthoryear{Bastos and Schoffelen}{Bastos and Schoffelen}{2016}]{bastosTutorialReviewFunctional2016}
Bastos, A.~M. and J.-M. Schoffelen (2016).
\newblock A {{Tutorial Review}} of {{Functional Connectivity Analysis Methods}} and {{Their Interpretational Pitfalls}}.
\newblock {\em Frontiers in Systems Neuroscience\/}~{\em 9}.

\bibitem[\protect\citeauthoryear{Betti, Corbetta, {de Pasquale}, Wens, and Della~Penna}{Betti et~al.}{2018}]{bettiTopologyFunctionalConnectivity2018}
Betti, V., M.~Corbetta, F.~{de Pasquale}, V.~Wens, and S.~Della~Penna (2018).
\newblock Topology of {{Functional Connectivity}} and {{Hub Dynamics}} in the {{Beta Band As Temporal Prior}} for {{Natural Vision}} in the {{Human Brain}}.
\newblock {\em The Journal of Neuroscience\/}~{\em 38\/}(15), 3858--3871.

\bibitem[\protect\citeauthoryear{Brookes, Woolrich, Luckhoo, Price, Hale, Stephenson, Barnes, Smith, and Morris}{Brookes et~al.}{2011}]{brookesInvestigatingElectrophysiologicalBasis2011}
Brookes, M.~J., M.~Woolrich, H.~Luckhoo, D.~Price, J.~R. Hale, M.~C. Stephenson, G.~R. Barnes, S.~M. Smith, and P.~G. Morris (2011).
\newblock Investigating the electrophysiological basis of resting state networks using magnetoencephalography.
\newblock {\em Proceedings of the National Academy of Sciences of the United States of America\/}~{\em 108\/}(40), 16783--16788.

\bibitem[\protect\citeauthoryear{Bruns}{Bruns}{2004}]{brunsFourierHilbertWaveletbased2004}
Bruns, A. (2004).
\newblock Fourier-, {{Hilbert-}} and wavelet-based signal analysis: Are they really different approaches?
\newblock {\em Journal of Neuroscience Methods\/}~{\em 137\/}(2), 321--332.

\bibitem[\protect\citeauthoryear{Buzs{\'a}ki}{Buzs{\'a}ki}{2011}]{buzsakiRhythmsBrain2011}
Buzs{\'a}ki, G. (2011).
\newblock {\em Rhythms of the Brain\/} (1 ed.).
\newblock New York, New York: Oxford University Press.

\bibitem[\protect\citeauthoryear{{Cam-CAN}, Shafto, Tyler, Dixon, Taylor, Rowe, Cusack, Calder, {Marslen-Wilson}, Duncan, Dalgleish, Henson, Brayne, and Matthews}{{Cam-CAN} et~al.}{2014}]{cam-canCambridgeCentreAgeing2014}
{Cam-CAN}, M.~A. Shafto, L.~K. Tyler, M.~Dixon, J.~R. Taylor, J.~B. Rowe, R.~Cusack, A.~J. Calder, W.~D. {Marslen-Wilson}, J.~Duncan, T.~Dalgleish, R.~N. Henson, C.~Brayne, and F.~E. Matthews (2014).
\newblock The {{Cambridge Centre}} for {{Ageing}} and {{Neuroscience}} ({{Cam-CAN}}) study protocol: A cross-sectional, lifespan, multidisciplinary examination of healthy cognitive ageing.
\newblock {\em BMC Neurology\/}~{\em 14\/}(1), 204.

\bibitem[\protect\citeauthoryear{{Carrasco-G{\'o}mez}, {Garc{\'i}a-Colomo}, {Cabrera-{\'A}lvarez}, Bru{\~n}a, Santos, and Maest{\'u}}{{Carrasco-G{\'o}mez} et~al.}{2026}]{carrasco-gomezDynamicsBrainConnectivity2026}
{Carrasco-G{\'o}mez}, M., A.~{Garc{\'i}a-Colomo}, J.~{Cabrera-{\'A}lvarez}, R.~Bru{\~n}a, A.~Santos, and F.~Maest{\'u} (2026).
\newblock Dynamics of brain connectivity across the {{Alzheimer}}'s disease spectrum through magnetoencephalography.
\newblock {\em Scientific Reports\/}~{\em 16\/}(1), 1--13.

\bibitem[\protect\citeauthoryear{Chen, Li, Zheng, and Fan}{Chen et~al.}{2025}]{chenDFCExpertLearningDynamic2025}
Chen, T., H.~Li, H.~Zheng, and Y.~Fan (2025).
\newblock {{dFCExpert}}: {{Learning Dynamic Functional Connectivity Patterns}} with {{Modularity}} and {{State Experts}}.
\newblock {\em bioRxiv\/}, 1--12.

\bibitem[\protect\citeauthoryear{Choe, Jones, Joel, Muschelli, Belegu, Caffo, Lindquist, van Zijl, and Pekar}{Choe et~al.}{2015}]{choeReproducibilityTemporalStructure2015}
Choe, A.~S., C.~K. Jones, S.~E. Joel, J.~Muschelli, V.~Belegu, B.~S. Caffo, M.~A. Lindquist, P.~C.~M. van Zijl, and J.~J. Pekar (2015).
\newblock Reproducibility and {{Temporal Structure}} in {{Weekly Resting-State fMRI}} over a {{Period}} of 3.5 {{Years}}.
\newblock {\em PLOS ONE\/}~{\em 10\/}(10), 1--29.

\bibitem[\protect\citeauthoryear{Coelli, Corda, and Bianchi}{Coelli et~al.}{2025}]{coelliTimevaryingBrainComprehensive2025}
Coelli, S., M.~Corda, and A.~M. Bianchi (2025).
\newblock The time-varying brain: A comprehensive review of dynamic functional connectivity analysis in {{EEG}} and {{MEG}}.
\newblock {\em Journal of Neural Engineering\/}~{\em 22\/}(5), 1--20.

\bibitem[\protect\citeauthoryear{Colclough, Brookes, Smith, and Woolrich}{Colclough et~al.}{2015}]{colcloughSymmetricMultivariateLeakage2015}
Colclough, G., M.~Brookes, S.~Smith, and M.~Woolrich (2015).
\newblock A symmetric multivariate leakage correction for {{MEG}} connectomes.
\newblock {\em Neuroimage\/}~{\em 117}, 439--448.

\bibitem[\protect\citeauthoryear{Colclough, Woolrich, Tewarie, Brookes, Quinn, and Smith}{Colclough et~al.}{2016}]{colcloughHowReliableAre2016}
Colclough, G.~L., M.~W. Woolrich, P.~K. Tewarie, M.~J. Brookes, A.~J. Quinn, and S.~M. Smith (2016).
\newblock How reliable are {{MEG}} resting-state connectivity metrics?
\newblock {\em NeuroImage\/}~{\em 138}, 284--293.

\bibitem[\protect\citeauthoryear{Colgin, Denninger, Fyhn, Hafting, Bonnevie, Jensen, Moser, and Moser}{Colgin et~al.}{2009}]{colginFrequencyGammaOscillations2009}
Colgin, L.~L., T.~Denninger, M.~Fyhn, T.~Hafting, T.~Bonnevie, O.~Jensen, M.-B. Moser, and E.~I. Moser (2009).
\newblock Frequency of gamma oscillations routes flow of information in the hippocampus.
\newblock {\em Nature\/}~{\em 462\/}(7271), 353--357.

\bibitem[\protect\citeauthoryear{Conti, D'Onofrio, Lorenzon, Grassi, Mascioli, Ferrari, Simonetta, Pavan, Di~Giuliano, Pierantozzi, Schirinzi, Centonze, Corbetta, Antonini, Stefani, and Guerra}{Conti et~al.}{2026}]{contiFrequencyNetworkSpecificChanges2026}
Conti, M., V.~D'Onofrio, L.~Lorenzon, L.~L. Grassi, D.~Mascioli, V.~Ferrari, C.~Simonetta, S.~Pavan, F.~Di~Giuliano, M.~Pierantozzi, T.~Schirinzi, D.~Centonze, M.~Corbetta, A.~Antonini, A.~Stefani, and A.~Guerra (2026).
\newblock Frequency- and {{Network-Specific Changes}} in {{Functional Connectivity Reflect Pathophysiological Mechanisms}} across {{Parkinson}}'s {{Disease Stages}}.
\newblock {\em Annals of Neurology\/}~{\em 100\/}(2), 319--333.

\bibitem[\protect\citeauthoryear{Dale, Liu, Fischl, Buckner, Belliveau, Lewine, and Halgren}{Dale et~al.}{2000}]{daleDynamicStatisticalParametric2000}
Dale, A.~M., A.~K. Liu, B.~R. Fischl, R.~L. Buckner, J.~W. Belliveau, J.~D. Lewine, and E.~Halgren (2000).
\newblock Dynamic statistical parametric mapping: Combining {{fMRI}} and {{MEG}} for high-resolution imaging of cortical activity.
\newblock {\em Neuron\/}~{\em 26\/}(1), 55--67.

\bibitem[\protect\citeauthoryear{Dale and Sereno}{Dale and Sereno}{1993}]{daleImprovedLocalizadonCortical1993}
Dale, A.~M. and M.~I. Sereno (1993).
\newblock Improved {{Localizadon}} of {{Cortical Activity}} by {{Combining EEG}} and {{MEG}} with {{MRI Cortical Surface Reconstruction}}: {{A Linear Approach}}.
\newblock {\em Journal of Cognitive Neuroscience\/}~{\em 5\/}(2), 162--176.

\bibitem[\protect\citeauthoryear{{de Pasquale}, Della~Penna, Sporns, Romani, and Corbetta}{{de Pasquale} et~al.}{2016}]{depasqualeDynamicCoreNetwork2016}
{de Pasquale}, F., S.~Della~Penna, O.~Sporns, G.~L. Romani, and M.~Corbetta (2016).
\newblock A {{Dynamic Core Network}} and {{Global Efficiency}} in the {{Resting Human Brain}}.
\newblock {\em Cerebral Cortex (New York, NY)\/}~{\em 26\/}(10), 4015--4033.

\bibitem[\protect\citeauthoryear{Desikan, S{\'e}gonne, Fischl, Quinn, Dickerson, Blacker, Buckner, Dale, Maguire, Hyman, Albert, and Killiany}{Desikan et~al.}{2006}]{desikanAutomatedLabelingSystem2006}
Desikan, R.~S., F.~S{\'e}gonne, B.~Fischl, B.~T. Quinn, B.~C. Dickerson, D.~Blacker, R.~L. Buckner, A.~M. Dale, R.~P. Maguire, B.~T. Hyman, M.~S. Albert, and R.~J. Killiany (2006).
\newblock An automated labeling system for subdividing the human cerebral cortex on {{MRI}} scans into gyral based regions of interest.
\newblock {\em NeuroImage\/}~{\em 31\/}(3), 968--980.

\bibitem[\protect\citeauthoryear{Ding, Dan, Wei, Laurienti, and Wu}{Ding et~al.}{2026}]{dingMachineLearningDynamic2026}
Ding, J., T.~Dan, Z.~Wei, P.~J. Laurienti, and G.~Wu (2026).
\newblock Machine {{Learning}} on {{Dynamic Functional Connectivity}}: {{Promise}}, {{Pitfalls}}, and {{Interpretations}}.
\newblock {\em Information Sciences\/}~{\em 740}, 1--20.

\bibitem[\protect\citeauthoryear{Durand and Le~Roux}{Durand and Le~Roux}{2017}]{durandLinkageIndexVariables2017}
Durand, J.-L. and B.~Le~Roux (2017).
\newblock Linkage index of variables and its relationship with variance of eigenvalue in {{PCA}} and {{MCA}}.
\newblock {\em Italian Journal of Applied Statistics\/}~{\em 29}, 123--136.

\bibitem[\protect\citeauthoryear{Eklund, Andersson, and Knutsson}{Eklund et~al.}{2011}]{eklundFastRandomPermutation2011}
Eklund, A., M.~Andersson, and H.~Knutsson (2011).
\newblock Fast {{Random Permutation Tests Enable Objective Evaluation}} of {{Methods}} for {{Single-Subject fMRI Analysis}}.
\newblock {\em International Journal of Biomedical Imaging\/}~{\em 2011}, 1--15.

\bibitem[\protect\citeauthoryear{Fan}{Fan}{1992}]{fanDesignadaptiveNonparametricRegression1992}
Fan, J. (1992).
\newblock Design-adaptive {{Nonparametric Regression}}.
\newblock {\em Journal of the American Statistical Association\/}~{\em 87\/}(420), 998--1004.

\bibitem[\protect\citeauthoryear{Fan and Gijbels}{Fan and Gijbels}{1995}]{fanDataDrivenBandwidthSelection1995}
Fan, J. and I.~Gijbels (1995).
\newblock Data-{{Driven Bandwidth Selection}} in {{Local Polynomial Fitting}}: {{Variable Bandwidth}} and {{Spatial Adaptation}}.
\newblock {\em Journal of the Royal Statistical Society. Series B (Methodological)\/}~{\em 57\/}(2), 371--394.

\bibitem[\protect\citeauthoryear{Fan and Gijbels}{Fan and Gijbels}{1996}]{fanLocalPolynomialModelling1996}
Fan, J. and I.~Gijbels (1996).
\newblock {\em Local Polynomial Modelling and Its Applications\/} (1 ed.).
\newblock Number~66 in Monographs on Statistics and Applied Probability. Boca Raton: Chapman \& Hall / CRC.

\bibitem[\protect\citeauthoryear{Fan, Heckman, and Wand}{Fan et~al.}{1995}]{fanLocalPolynomialKernel1995}
Fan, J., N.~E. Heckman, and M.~P. Wand (1995).
\newblock Local {{Polynomial Kernel Regression}} for {{Generalized Linear Models}} and {{Quasi-Likelihood Functions}}.
\newblock {\em Journal of the American Statistical Association\/}~{\em 90\/}(429), 141--150.

\bibitem[\protect\citeauthoryear{Fan and Yao}{Fan and Yao}{1998}]{fanEfficientEstimationConditional1998}
Fan, J. and Q.~Yao (1998).
\newblock Efficient estimation of conditional variance functions in stochastic regression.
\newblock {\em Biometrika\/}~{\em 85\/}(3), 645--660.

\bibitem[\protect\citeauthoryear{Fan, Yao, and Tong}{Fan et~al.}{1996}]{fanEstimationConditionalDensities1996}
Fan, J., Q.~Yao, and H.~Tong (1996).
\newblock Estimation of {{Conditional Densities}} and {{Sensitivity Measures}} in {{Nonlinear Dynamical Systems}}.
\newblock {\em Biometrika\/}~{\em 83\/}(1), 189--206.

\bibitem[\protect\citeauthoryear{Favaretto, Allegra, Deco, Metcalf, Griffis, Shulman, Brovelli, and Corbetta}{Favaretto et~al.}{2022}]{favarettoSubcorticalcorticalDynamicalStates2022}
Favaretto, C., M.~Allegra, G.~Deco, N.~V. Metcalf, J.~C. Griffis, G.~L. Shulman, A.~Brovelli, and M.~Corbetta (2022).
\newblock Subcortical-cortical dynamical states of the human brain and their breakdown in stroke.
\newblock {\em Nature Communications\/}~{\em 13\/}(1), 1--17.

\bibitem[\protect\citeauthoryear{Fischl}{Fischl}{2012}]{fischlFreeSurfer2012}
Fischl, B. (2012).
\newblock {{FreeSurfer}}.
\newblock {\em NeuroImage\/}~{\em 62\/}(2), 774--781.

\bibitem[\protect\citeauthoryear{Gasser, Sroka, and {Jennen-steinmetz}}{Gasser et~al.}{1986}]{gasserResidualVarianceResidual1986}
Gasser, T., L.~Sroka, and C.~{Jennen-steinmetz} (1986).
\newblock Residual variance and residual pattern in nonlinear regression.
\newblock {\em Biometrika\/}~{\em 73\/}(3), 625--633.

\bibitem[\protect\citeauthoryear{Geddes and Baker}{Geddes and Baker}{1967}]{geddesSpecificResistanceBiological1967}
Geddes, L.~A. and L.~E. Baker (1967).
\newblock The specific resistance of biological material---{{A}} compendium of data for the biomedical engineer and physiologist.
\newblock {\em Medical and biological engineering\/}~{\em 5\/}(3), 271--293.

\bibitem[\protect\citeauthoryear{Gramfort}{Gramfort}{2013}]{gramfortMEGEEGData2013}
Gramfort, A. (2013).
\newblock {{MEG}} and {{EEG}} data analysis with {{MNE-Python}}.
\newblock {\em Frontiers in Neuroscience\/}~{\em 7}, 1--13.

\bibitem[\protect\citeauthoryear{Guglielmi, Cisotto, Erseghe, and Badia}{Guglielmi et~al.}{2022}]{guglielmiFrequencyDependentFunctionalConnectivity2022}
Guglielmi, A.~V., G.~Cisotto, T.~Erseghe, and L.~Badia (2022).
\newblock Frequency-{{Dependent Functional Connectivity}} of {{Brain Networks}} at {{Resting-State}}.
\newblock In {\em 2022 14th {{Biomedical Engineering International Conference}} ({{BMEiCON}})}, pp.\  1--5.

\bibitem[\protect\citeauthoryear{H{\"a}m{\"a}l{\"a}inen and Ilmoniemi}{H{\"a}m{\"a}l{\"a}inen and Ilmoniemi}{1994}]{hamalainenInterpretingMagneticFields1994}
H{\"a}m{\"a}l{\"a}inen, M.~S. and R.~J. Ilmoniemi (1994).
\newblock Interpreting magnetic fields of the brain: Minimum norm estimates.
\newblock {\em Medical \& Biological Engineering \& Computing\/}~{\em 32\/}(1), 35--42.

\bibitem[\protect\citeauthoryear{H{\"a}m{\"a}l{\"a}inen and Sarvas}{H{\"a}m{\"a}l{\"a}inen and Sarvas}{1989}]{hamalainenRealisticConductivityGeometry1989}
H{\"a}m{\"a}l{\"a}inen, M.~S. and J.~Sarvas (1989).
\newblock Realistic conductivity geometry model of the human head for interpretation of neuromagnetic data.
\newblock {\em IEEE transactions on bio-medical engineering\/}~{\em 36\/}(2), 165--171.

\bibitem[\protect\citeauthoryear{Han, Su, He, Zhan, Plis, Calhoun, and Yang}{Han et~al.}{2026}]{hanRethinkingFunctionalBrain2026}
Han, K., Y.~Su, L.~He, L.~Zhan, S.~Plis, V.~Calhoun, and C.~Yang (2026).
\newblock Rethinking functional brain connectome analysis: Do graph deep learning models {{Help}}.
\newblock {\em npj Artificial Intelligence\/}~{\em 2\/}(1), 1--14.

\bibitem[\protect\citeauthoryear{Higham}{Higham}{2002}]{highamComputingNearestCorrelation2002}
Higham, N.~J. (2002).
\newblock Computing the nearest correlation matrix---a problem from finance.
\newblock {\em IMA Journal of Numerical Analysis\/}~{\em 22\/}(3), 329--343.

\bibitem[\protect\citeauthoryear{Hipp, Hawellek, Corbetta, Siegel, and Engel}{Hipp et~al.}{2012}]{hippLargescaleCorticalCorrelation2012}
Hipp, J.~F., D.~J. Hawellek, M.~Corbetta, M.~Siegel, and A.~K. Engel (2012).
\newblock Large-scale cortical correlation structure of spontaneous oscillatory activity.
\newblock {\em Nature Neuroscience\/}~{\em 15\/}(6), 884--890.

\bibitem[\protect\citeauthoryear{Huang, Yurgil, Robb, Angeles, Diwakar, Risbrough, Nichols, McLay, Theilmann, Song, Huang, Lee, and Baker}{Huang et~al.}{2014}]{huangVoxelwiseRestingstateMEG2014}
Huang, M.-X., K.~A. Yurgil, A.~Robb, A.~Angeles, M.~Diwakar, V.~B. Risbrough, S.~L. Nichols, R.~McLay, R.~J. Theilmann, T.~Song, C.~W. Huang, R.~R. Lee, and D.~G. Baker (2014).
\newblock Voxel-wise resting-state {{MEG}} source magnitude imaging study reveals neurocircuitry abnormality in active-duty service members and veterans with {{PTSD}}.
\newblock {\em NeuroImage : Clinical\/}~{\em 5}, 408--419.

\bibitem[\protect\citeauthoryear{Huang, Nouranizadeh, Ahrends, and Xu}{Huang et~al.}{2025}]{huangBrainATCLAdaptiveTemporal2025}
Huang, Y., A.~Nouranizadeh, C.~Ahrends, and M.~Xu (2025).
\newblock {{BrainATCL}}: {{Adaptive Temporal Brain Connectivity Learning}} for {{Functional Link Prediction}} and {{Age Estimation}}.

\bibitem[\protect\citeauthoryear{Jin, Ranasinghe, Prabhu, Dale, Gao, Kudo, Vossel, Raj, Nagarajan, and Jiang}{Jin et~al.}{2023}]{jinDynamicFunctionalConnectivity2023}
Jin, H., K.~G. Ranasinghe, P.~Prabhu, C.~Dale, Y.~Gao, K.~Kudo, K.~Vossel, A.~Raj, S.~S. Nagarajan, and F.~Jiang (2023).
\newblock Dynamic functional connectivity {{MEG}} features of {{Alzheimer}}'s disease.
\newblock {\em NeuroImage\/}~{\em 281}, 1--30.

\bibitem[\protect\citeauthoryear{Kajimura, Margulies, and Smallwood}{Kajimura et~al.}{2023}]{kajimuraFrequencyspecificBrainNetwork2023}
Kajimura, S., D.~Margulies, and J.~Smallwood (2023).
\newblock Frequency-specific brain network architecture in resting-state {{fMRI}}.
\newblock {\em Scientific Reports\/}~{\em 13\/}(1), 1--9.

\bibitem[\protect\citeauthoryear{Kendall, Stuart, Ord, Stuart, and Kendall}{Kendall et~al.}{1973}]{kendallInferenceRelationship1973}
Kendall, M.~G., A.~Stuart, J.~K. Ord, A.~Stuart, and M.~G. Kendall (1973).
\newblock {\em Inference and Relationship\/} (3 ed.).
\newblock Number~2 in The Advanced Theory of Statistics. London: Griffin.

\bibitem[\protect\citeauthoryear{Ko, Darvas, Poliakov, Ojemann, and Sorensen}{Ko et~al.}{2011}]{koQuasiperiodicFluctuationsDefault2011}
Ko, A.~L., F.~Darvas, A.~Poliakov, J.~Ojemann, and L.~B. Sorensen (2011).
\newblock Quasi-periodic {{Fluctuations}} in {{Default Mode Network Electrophysiology}}.
\newblock {\em The Journal of Neuroscience\/}~{\em 31\/}(32), 11728--11732.

\bibitem[\protect\citeauthoryear{Ko, Weaver, Hakimian, and Ojemann}{Ko et~al.}{2013}]{koIdentifyingFunctionalNetworks2013}
Ko, A.~L., K.~E. Weaver, S.~Hakimian, and J.~G. Ojemann (2013).
\newblock Identifying functional networks using endogenous connectivity in gamma band electrocorticography.
\newblock {\em Brain Connectivity\/}~{\em 3\/}(5), 491--502.

\bibitem[\protect\citeauthoryear{Li}{Li}{2021}]{liStatisticalInferenceHighdimensional2021}
Li, J. (2021).
\newblock Statistical {{Inference}} for {{High-dimensional Nonparametric Models}}.

\bibitem[\protect\citeauthoryear{Li{\'e}geois, Laumann, Snyder, Zhou, and Yeo}{Li{\'e}geois et~al.}{2017}]{liegeoisInterpretingTemporalFluctuations2017}
Li{\'e}geois, R., T.~O. Laumann, A.~Z. Snyder, J.~Zhou, and B.~T.~T. Yeo (2017).
\newblock Interpreting temporal fluctuations in resting-state functional connectivity {{MRI}}.
\newblock {\em NeuroImage\/}~{\em 163}, 437--455.

\bibitem[\protect\citeauthoryear{Lin, Witzel, Ahlfors, Stufflebeam, Belliveau, and H{\"a}m{\"a}l{\"a}inen}{Lin et~al.}{2006}]{linAssessingImprovingSpatial2006}
Lin, F.-H., T.~Witzel, S.~P. Ahlfors, S.~M. Stufflebeam, J.~W. Belliveau, and M.~S. H{\"a}m{\"a}l{\"a}inen (2006).
\newblock Assessing and improving the spatial accuracy in {{MEG}} source localization by depth-weighted minimum-norm estimates.
\newblock {\em NeuroImage\/}~{\em 31\/}(1), 160--171.

\bibitem[\protect\citeauthoryear{Lin, Witzel, H{\"a}m{\"a}l{\"a}inen, Dale, Belliveau, and Stufflebeam}{Lin et~al.}{2004}]{linSpectralSpatiotemporalImaging2004}
Lin, F.-H., T.~Witzel, M.~S. H{\"a}m{\"a}l{\"a}inen, A.~M. Dale, J.~W. Belliveau, and S.~M. Stufflebeam (2004).
\newblock Spectral spatiotemporal imaging of cortical oscillations and interactions in the human brain.
\newblock {\em NeuroImage\/}~{\em 23\/}(2), 582--595.

\bibitem[\protect\citeauthoryear{Liu, Fukunaga, {de Zwart}, and Duyn}{Liu et~al.}{2010}]{liuLargescaleSpontaneousFluctuations2010}
Liu, Z., M.~Fukunaga, J.~A. {de Zwart}, and J.~H. Duyn (2010).
\newblock Large-scale spontaneous fluctuations and correlations in brain electrical activity observed with magnetoencephalography.
\newblock {\em NeuroImage\/}~{\em 51\/}(1), 102--111.

\bibitem[\protect\citeauthoryear{Liu, Luppi, Hansen, Tian, Zalesky, Yeo, Fulcher, and Misic}{Liu et~al.}{2025}]{liuBenchmarkingMethodsMapping2025}
Liu, Z.-Q., A.~I. Luppi, J.~Y. Hansen, Y.~E. Tian, A.~Zalesky, B.~T.~T. Yeo, B.~D. Fulcher, and B.~Misic (2025).
\newblock Benchmarking methods for mapping functional connectivity in the brain.
\newblock {\em Nature Methods\/}~{\em 22\/}(7), 1593--1602.

\bibitem[\protect\citeauthoryear{Mack and Silverman}{Mack and Silverman}{1982}]{mackWeakStrongUniform1982}
Mack, Y.~P. and B.~W. Silverman (1982).
\newblock Weak and strong uniform consistency of kernel regression estimates.
\newblock {\em Zeitschrift f\"ur Wahrscheinlichkeitstheorie und Verwandte Gebiete\/}~{\em 61\/}(3), 405--415.

\bibitem[\protect\citeauthoryear{Marinkovic, Dhond, Dale, Glessner, Carr, and Halgren}{Marinkovic et~al.}{2003}]{marinkovicSpatiotemporalDynamicsModalityspecific2003}
Marinkovic, K., R.~P. Dhond, A.~M. Dale, M.~Glessner, V.~Carr, and E.~Halgren (2003).
\newblock Spatiotemporal dynamics of modality-specific and supramodal word processing.
\newblock {\em Neuron\/}~{\em 38\/}(3), 487--497.

\bibitem[\protect\citeauthoryear{M{\"u}ller}{M{\"u}ller}{1991}]{mullerSmoothOptimumKernel1991}
M{\"u}ller, H.-G. (1991).
\newblock Smooth optimum kernel estimators near endpoints.
\newblock {\em Biometrika\/}~{\em 78\/}(3), 521--530.

\bibitem[\protect\citeauthoryear{O'Neill, Barratt, Hunt, Tewarie, and Brookes}{O'Neill et~al.}{2015}]{oneillMeasuringElectrophysiologicalConnectivity2015}
O'Neill, G.~C., E.~L. Barratt, B.~A.~E. Hunt, P.~K. Tewarie, and M.~J. Brookes (2015).
\newblock Measuring electrophysiological connectivity by power envelope correlation: A technical review on {{MEG}} methods.
\newblock {\em Physics in Medicine and Biology\/}~{\em 60\/}(21), R271--295.

\bibitem[\protect\citeauthoryear{Rahimi, Jackson, Farahibozorg, and Hauk}{Rahimi et~al.}{2023}]{rahimiTimeLaggedMultidimensionalPattern2023}
Rahimi, S., R.~Jackson, S.-R. Farahibozorg, and O.~Hauk (2023).
\newblock Time-{{Lagged Multidimensional Pattern Connectivity}} ({{TL-MDPC}}): {{An EEG}}/{{MEG}} pattern transformation based functional connectivity metric.
\newblock {\em NeuroImage\/}~{\em 270}, 1--16.

\bibitem[\protect\citeauthoryear{Ruppert and Wand}{Ruppert and Wand}{1994}]{ruppertMultivariateLocallyWeighted1994}
Ruppert, D. and M.~P. Wand (1994).
\newblock Multivariate {{Locally Weighted Least Squares Regression}}.
\newblock {\em Annals of Statistics\/}~{\em 22\/}(3), 1346--1370.

\bibitem[\protect\citeauthoryear{Schaefer, Kong, Gordon, Laumann, Zuo, Holmes, Eickhoff, and Yeo}{Schaefer et~al.}{2018}]{schaeferLocalGlobalParcellationHuman2018}
Schaefer, A., R.~Kong, E.~M. Gordon, T.~O. Laumann, X.-N. Zuo, A.~J. Holmes, S.~B. Eickhoff, and B.~T.~T. Yeo (2018).
\newblock Local-{{Global Parcellation}} of the {{Human Cerebral Cortex}} from {{Intrinsic Functional Connectivity MRI}}.
\newblock {\em Cerebral Cortex\/}~{\em 28\/}(9), 3095--3114.

\bibitem[\protect\citeauthoryear{Schoffelen and Gross}{Schoffelen and Gross}{2009}]{schoffelenSourceConnectivityAnalysis2009}
Schoffelen, J.-M. and J.~Gross (2009).
\newblock Source connectivity analysis with {{MEG}} and {{EEG}}.
\newblock {\em Human Brain Mapping\/}~{\em 30\/}(6), 1857--1865.

\bibitem[\protect\citeauthoryear{S{\'e}gonne, Dale, Busa, Glessner, Salat, Hahn, and Fischl}{S{\'e}gonne et~al.}{2004}]{segonneHybridApproachSkull2004}
S{\'e}gonne, F., A.~M. Dale, E.~Busa, M.~Glessner, D.~Salat, H.~K. Hahn, and B.~Fischl (2004).
\newblock A hybrid approach to the skull stripping problem in {{MRI}}.
\newblock {\em NeuroImage\/}~{\em 22\/}(3), 1060--1075.

\bibitem[\protect\citeauthoryear{Silverman}{Silverman}{1978}]{silvermanWeakStrongUniform1978}
Silverman, B.~W. (1978).
\newblock Weak and {{Strong Uniform Consistency}} of the {{Kernel Estimate}} of a {{Density}} and its {{Derivatives}}.
\newblock {\em The Annals of Statistics\/}~{\em 6\/}(1), 177--184.

\bibitem[\protect\citeauthoryear{Silverman}{Silverman}{1986}]{silvermanDensityEstimationStatistics1986}
Silverman, B.~W. (1986).
\newblock {\em Density {{Estimation}} for {{Statistics}} and {{Data Analysis}}}.
\newblock CRC Press.

\bibitem[\protect\citeauthoryear{Staniswalis and Lee}{Staniswalis and Lee}{1998}]{staniswalisNonparametricRegressionAnalysis1998}
Staniswalis, J.~G. and J.~J. Lee (1998).
\newblock Nonparametric {{Regression Analysis}} of {{Longitudinal Data}}.
\newblock {\em Journal of the American Statistical Association\/}~{\em 93\/}(444), 1403--1418.

\bibitem[\protect\citeauthoryear{Staudenmayer and Ruppert}{Staudenmayer and Ruppert}{2004}]{staudenmayerLocalPolynomialRegression2004}
Staudenmayer, J. and D.~Ruppert (2004).
\newblock Local {{Polynomial Regression}} and {{Simulation}}--{{Extrapolation}}.
\newblock {\em Journal of the Royal Statistical Society Series B: Statistical Methodology\/}~{\em 66\/}(1), 17--30.

\bibitem[\protect\citeauthoryear{Taulu and Simola}{Taulu and Simola}{2006}]{tauluSpatiotemporalSignalSpace2006}
Taulu, S. and J.~Simola (2006).
\newblock Spatiotemporal signal space separation method for rejecting nearby interference in {{MEG}} measurements.
\newblock {\em Physics in Medicine and Biology\/}~{\em 51\/}(7), 1759--1768.

\bibitem[\protect\citeauthoryear{Taulu, Simola, and Kajola}{Taulu et~al.}{2005}]{tauluApplicationsSignalSpace2005}
Taulu, S., J.~Simola, and M.~Kajola (2005).
\newblock Applications of the signal space separation method.
\newblock {\em IEEE Transactions on Signal Processing\/}~{\em 53\/}(9), 3359--3372.

\bibitem[\protect\citeauthoryear{Tholkage, Zheng, and Kulasekera}{Tholkage et~al.}{2022}]{tholkageConditionalKaplanMeier2022}
Tholkage, S., Q.~Zheng, and K.~B. Kulasekera (2022).
\newblock Conditional {{Kaplan}}--{{Meier Estimator}} with {{Functional Covariates}} for {{Time-to-Event Data}}.
\newblock {\em Stats\/}~{\em 5\/}(4), 1113--1129.

\bibitem[\protect\citeauthoryear{{van Es}, Higgins, Gohil, Quinn, Vidaurre, and Woolrich}{{van Es} et~al.}{2025}]{vanesLargescaleCorticalFunctional2025}
{van Es}, M. W.~J., C.~Higgins, C.~Gohil, A.~J. Quinn, D.~Vidaurre, and M.~W. Woolrich (2025).
\newblock Large-scale cortical functional networks are organized in structured cycles.
\newblock {\em Nature Neuroscience\/}~{\em 28\/}(10), 2118--2128.

\bibitem[\protect\citeauthoryear{Vorwerk, Cho, Rampp, Hamer, Kn{\"o}sche, and Wolters}{Vorwerk et~al.}{2014}]{vorwerkGuidelineHeadVolume2014}
Vorwerk, J., J.-H. Cho, S.~Rampp, H.~Hamer, T.~R. Kn{\"o}sche, and C.~H. Wolters (2014).
\newblock A guideline for head volume conductor modeling in {{EEG}} and {{MEG}}.
\newblock {\em NeuroImage\/}~{\em 100}, 590--607.

\bibitem[\protect\citeauthoryear{Wen, Wang, Liu, Meng, and Jiao}{Wen et~al.}{2025}]{wenSpatiotemporalDynamicFunctional2025}
Wen, S., J.~Wang, W.~Liu, X.~Meng, and Z.~Jiao (2025).
\newblock Spatio-temporal dynamic functional brain network for mild cognitive impairment analysis.
\newblock {\em Frontiers in Neuroscience\/}~{\em 19}, 1--15.

\bibitem[\protect\citeauthoryear{Widmann, Schr{\"o}ger, and Maess}{Widmann et~al.}{2015}]{widmannDigitalFilterDesign2015}
Widmann, A., E.~Schr{\"o}ger, and B.~Maess (2015).
\newblock Digital filter design for electrophysiological data -- a practical approach.
\newblock {\em Journal of Neuroscience Methods\/}~{\em 250}, 34--46.

\bibitem[\protect\citeauthoryear{Wirsich, Giraud, and Sadaghiani}{Wirsich et~al.}{2020}]{wirsichConcurrentEEGFMRIderived2020}
Wirsich, J., A.-L. Giraud, and S.~Sadaghiani (2020).
\newblock Concurrent {{EEG-}} and {{fMRI-derived}} functional connectomes exhibit linked dynamics.
\newblock {\em NeuroImage\/}~{\em 219}, 1--13.

\bibitem[\protect\citeauthoryear{Woolrich, Baker, Luckhoo, Mohseni, Barnes, Brookes, and Rezek}{Woolrich et~al.}{2013}]{woolrichDynamicStateAllocation2013}
Woolrich, M.~W., A.~Baker, H.~Luckhoo, H.~Mohseni, G.~Barnes, M.~Brookes, and I.~Rezek (2013).
\newblock Dynamic state allocation for {{MEG}} source reconstruction.
\newblock {\em NeuroImage\/}~{\em 77}, 77--92.

\bibitem[\protect\citeauthoryear{Xu, Choi, Zhao, Li, Rogers, Anderson, Gore, Gao, and Ding}{Xu et~al.}{2023}]{xuSeasonalVariationsFunctional2023}
Xu, L., S.~Choi, Y.~Zhao, M.~Li, B.~P. Rogers, A.~Anderson, J.~C. Gore, Y.~Gao, and Z.~Ding (2023).
\newblock Seasonal variations of functional connectivity of human brains.
\newblock {\em Scientific Reports\/}~{\em 13\/}(1), 1--12.

\bibitem[\protect\citeauthoryear{Zhang and Li}{Zhang and Li}{2022}]{zhangFactorizedEstimationHighdimensional2022}
Zhang, J. and J.~Li (2022).
\newblock Factorized estimation of high-dimensional nonparametric covariance models.
\newblock {\em Scandinavian Journal of Statistics\/}~{\em 49\/}(2), 542--567.

\bibitem[\protect\citeauthoryear{Zhang, Lyu, Yu, Zhang, Cao, Chen, Chen, Zhuang, Liu, and Zhu}{Zhang et~al.}{2025}]{zhangCLASSIFFICATIONMILDCOGNITIVE2025}
Zhang, J., Y.~Lyu, X.~Yu, L.~Zhang, C.~Cao, T.~Chen, M.~Chen, Y.~Zhuang, T.~Liu, and D.~Zhu (2025).
\newblock Classification of {{Mild Cognitive Impairment Based}} on {{Dynamic Functional Connectivity Using Spatio-Temporal Transformer}}.
\newblock {\em Proceedings. IEEE International Symposium on Biomedical Imaging\/}~{\em 2025}, 1--12.

\bibitem[\protect\citeauthoryear{Zheng, Bieri, Hsiao, and Colgin}{Zheng et~al.}{2016}]{zhengSpatialSequenceCoding2016}
Zheng, C., K.~W. Bieri, Y.-T. Hsiao, and L.~L. Colgin (2016).
\newblock Spatial sequence coding differs during slow and fast gamma rhythms in the hippocampus.
\newblock {\em Neuron\/}~{\em 89\/}(2), 398--408.

\bibitem[\protect\citeauthoryear{Zink, M{\"u}ckschel, and Beste}{Zink et~al.}{2021}]{zinkRestingstateEEGDynamics2021}
Zink, N., M.~M{\"u}ckschel, and C.~Beste (2021).
\newblock Resting-state {{EEG Dynamics Reveals Differences}} in {{Network Organization}} and its {{Fluctuation}} between {{Frequency Bands}}.
\newblock {\em Neuroscience\/}~{\em 453}, 43--56.

\end{thebibliography}
\newpage
\section*{Appendix: Proofs of Theoretical Results}\label{sec:appendix_proofs}
\subsection{Preliminary Lemmas}\label{subsec:preliminary_lemmas}
\begin{lemma}\label{lem:bias-variance expansion}
    Let \(Y_{i}=\varphi(\mathbf X_{i})\) with \(\mathbb{E}\left[ Y\mid U=u \right]=m(u)\in C^{2}\), \(\mathrm{Var}(Y\mid U=u)=\sigma^{2}(u)\) are continuous. Define \(\hat m_{n}(u_{0})=\frac{\sum_{i} Y_{i} K_{h}(U_{i}-u_{0})}{\sum_{i} K_{h}(U_{i}-u_{0})}\). Under Assumptions (A1) and (A2), for interior \(u_{0}\), i.e., \([u_{0}-h,u_{0}+h]\subset(-1,1)\),  we have the following bias and variance expansions for the Nadaraya-Watson estimator \(\hat m_{n}(u_{0})\):
    \begin{align*}
        \mathbb{E}\left[ \hat m_{n}(u_{0}) \right]-m(u_{0})& =   \frac{h^{2}\mu_{2}(K)}{2}\Big[m^{\prime\prime}(u_{0})+2\frac{f_{U}^{\prime}(u_{0})}{f_{U}(u_{0})}m^{\prime}(u_{0})\Big]+o(h^{2}) \\
        \mathrm{Var}\big(\hat m_{n}(u_{0})\big)            & =  \frac{R(K)\,\sigma^{2}(u_{0})}{nh\,f_{U}(u_{0})}+o\left(\frac{1}{nh}\right).
    \end{align*}
\end{lemma}
\begin{proof}
    Write \(\hat g_{n}(u_{0})=\frac{1}{n}\sum_{i} Y_{i} K_{h}(U_{i}-u_{0})\), \(\hat f_{n}(u_{0})=\frac{1}{n}\sum_{i} K_{h}(U_{i}-u_{0})\). Substituting \(t=(u-u_{0})/h\), we have
    \begin{equation*}
        \mathbb{E}\left[ \hat g_{n}(u_{0}) \right]=\int m(u_{0}+ht)f_{U}(u_{0}+ht)K(t)\,\mathrm{d}t.
    \end{equation*}
    Take the Taylor-expanding of \(m(u_{0}+ht)f_{U}(u_{0}+ht)\) to the second order and use the facts: \(\int K=1,\int tK=0,\int t^{2}K=\mu_{2}(K)\), we have
    \begin{equation*}
        \mathbb{E}\left[ \hat g_{n}(u_{0}) \right]=m(u_{0})f_{U}(u_{0})+\frac{h^{2}\mu_{2}(K)}{2}\big[m f_{U}\big]^{\prime\prime}(u_{0})+o(h^{2}),
    \end{equation*}
    and similarly \(\mathbb{E}\left[ \hat f_{n}(u_{0}) \right]=f_{U}(u_{0})+\frac{h^{2}\mu_{2}(K)}{2}f_{U}^{\prime\prime}(u_{0})+o(h^{2})\). Forming the ratio and using \([mf_{U}]^{\prime\prime}=m^{\prime\prime}f_{U}+2m^{\prime}f_{U}^{\prime}+mf_{U}^{\prime\prime}\) gives the bias formula. For the variance, \(\mathrm{Var}(YK_{h}(U-u_{0}))=\frac{1}{h}\sigma^{2}(u_{0})f_{U}(u_{0})R(K)+o(1/h)\) by the same substitution, and the delta method for the ratio \(\hat g_{n}/\hat f_{n}\) (with \(\hat f_{n}\to f_{U}(u_{0})>0\)) yields the stated variance.
\end{proof}
\begin{lemma}\label{lem:uniform-consistency}
    \textbf{Uniform Consistency}. Under Assumptions (A1), (A2), (A4) and (A6), for any compact set \(\mathcal{U}_{\delta}=[-1+\delta,1-\delta]\), we have
    \begin{align*}
        \sup_{u_{0}\in\mathcal{U}_{\delta}}|A_{n}(u_{0})-A(u_{0})& | = O_{p}\!\left(\sqrt{\frac{\log(1/h)}{nh}}+h^{2}\right),  \\
        \sup_{u_{0}\in\mathcal{U}_{\delta}}|B_{n}(u_{0})-B(u_{0})& | =  O_{p}\!\left(\sqrt{\frac{\log(1/h)}{nh}}+h^{2}\right).
    \end{align*}
\end{lemma}
\begin{proof}
    This proof uses the classical uniform-consistency argument for kernel regression \citep{silvermanWeakStrongUniform1978,mackWeakStrongUniform1982}. We first decompose the error and then control each component.

    \textbf{Step 1: Decomposition.}
    We write the error in \(A_{n}(u_{0})\) as a ratio:
    \begin{equation*}
        A_{n}(u_{0})-A(u_{0}) = \frac{\hat g_{n}^A(u_{0})-\mathbb{E}[\hat g_{n}^A(u_{0})]}{\hat f_{n}(u_{0})} + \frac{\mathbb{E}[\hat g_{n}^A(u_{0})]-A(u_{0})\mathbb{E}[\hat f_{n}(u_{0})]}{\hat f_{n}(u_{0})} + \frac{A(u_{0})[\mathbb{E}[\hat f_{n}(u_{0})]-\hat f_{n}(u_{0})]}{\hat f_{n}(u_{0})},
    \end{equation*}
    where \(\hat g_{n}^A(u_{0})=\frac{1}{n}\sum_{i} (x_{i1}^{2}+x_{i2}^{2})K_{h}(U_{i}-u_{0})\) and \(\hat f_{n}(u_{0})=\frac{1}{n}\sum_{i} K_{h}(U_{i}-u_{0})\). The second term (bias) is \(O(h^{2})\) by \Cref{lem:bias-variance expansion}. For the first and third terms (stochastic terms), we need uniform bounds over \(u_{0} \in \mathcal{U}_{\delta}\).

    \textbf{Step 2: Uniform Stochastic Bound via Empirical Process Theory.}
    Define the empirical process indexed by \(u_{0}\):
    \begin{equation*}
        V_{n}(u_{0}) \coloneqq \frac{1}{n}\sum_{i} \big[w_{i}(u_{0})Y_{i} - \mathbb{E}[w_{i}(u_{0})Y_{i}]\big],
    \end{equation*}
    where \(w_{i}(u_{0})=K_{h}(U_{i}-u_{0})\) and \(Y_{i} = (x_{i1}^{2}+x_{i2}^{2})\) (or 1 for the denominator). This is a centred empirical process indexed by the location parameter \(u_{0} \in \mathcal{U}_{\delta}\).

    The key observation is that the kernel function \(K_{h}(\cdot)\) is Lipschitz continuous with respect to \(u_{0}\): \(|K_{h}(u_{i}-u_{0})-K_{h}(u_{i}-u_{0}^{\prime})| \leq L_K h^{-1}|u_{0}-u_{0}^{\prime}|\) for some constant \(L_K\). This Lipschitz structure, combined with the bounded variation property of the kernel, allows us to apply empirical process theory.

    \textbf{Step 3: Application of Chaining/Bracketing.}
    We discretise \(\mathcal{U}_{\delta}\) into a finite net with mesh size proportional to \(h\). For each grid point \(u_{0}^{(j)}\), by Bernstein's inequality applied to the sub-exponential random variables (bounded after truncation by Assumption (A5)), we have
    \begin{equation*}
        \mathbb{P}\Big(|V_{n}(u_{0}^{(j)})| > t\Big) \leq 2\exp\left(-cnt^{2}/(1+t)\right)
    \end{equation*}
    for some constant \(c>0\) depending on the sub-exponential norm. Union bound over the \(O(h^{-1})\) grid points gives
    \begin{equation*}
        \sup_{u_{0} \in \mathcal{U}_{\delta}}|V_{n}(u_{0})| = O_{p}\!\left(\sqrt{\frac{\log(1/h)}{n}}\right).
    \end{equation*}
    Since we are working with sample size \(n\) and bandwidth \(h\), the effective number of observations in the local window is \(nh\), so the rate becomes \(O_{p}\!\left(\sqrt{\frac{\log(1/h)}{nh}}\right)\).

    \textbf{Step 4: Control of the Denominator.}
    By the Law of Large Numbers and the Lipschitz property, \(\hat f_{n}(u_{0}) \to f_{U}(u_{0})\) uniformly over \(\mathcal{U}_{\delta}\). Since \(\inf_{u_{0}\in\mathcal{U}_{\delta}}f_{U}(u_{0})\geq f_{\min}>0\) (the density is bounded away from zero on the interior compact set by Assumption (A1)), we have \(\hat f_{n}(u_{0}) \geq f_{\min}/2\) with high probability for large \(n\). Thus, the ratio \(1/\hat f_{n}(u_{0})\) is uniformly bounded.

    \textbf{Step 5: Combining Terms.}
    Combining the bias term (\(O(h^{2})\)) with the stochastic terms \( \left(O_{p}\!\left(\sqrt{\frac{\log(1/h)}{nh}}\right)\right) \) via the uniform bound on \(1/\hat f_{n}(u_{0})\) yields
    \begin{equation*}
        \sup_{u_{0}\in\mathcal{U}_{\delta}}|A_{n}(u_{0})-A(u_{0})| = O_{p}\!\left(\sqrt{\frac{\log(1/h)}{nh}}+h^{2}\right).
    \end{equation*}
    The same argument applies to \(B_{n}(u_{0})\) which completes the proof.
\end{proof}
\begin{lemma}\label{lem:continuity_polynomial_roots}
    \textbf{Continuity of polynomial roots --- Rouché's  Theorem}. Let \(p(z)=z^{3}-B_{0}z^{2}+(A_{0}-1)z-B_{0}\) have a simple root \(z_{0}\). If \(p_{n}(z)=z^{3}-B_{n} z^{2}+(A_{n}-1)z-B_{n}\) with \(A_{n}\to A_{0}\), \(B_{n}\to B_{0}\), then for every \(\varepsilon>0\) small enough that the disk \(\overline D(z_{0},\varepsilon)\) contains no other root of \(p\), there is \(N(\varepsilon)\) such that for \(n\geq N(\varepsilon)\), \(p_{n}\) has \textbf{exactly one} root \(z_{n}\in D(z_{0},\varepsilon)\), and \(z_{n}\to z_{0}\).

\end{lemma}
\begin{proof}
    The proof relies on the persistence of roots under small perturbations, established via Rouché's theorem.

    Given that \(z_{0}\) is a simple root of \(p\), we have \(p(z_{0})=0\) and \(p^{\prime}(z_{0}) \neq 0\). Choose \(\varepsilon>0\) small enough that the closed disk \(\overline{D}(z_{0},\varepsilon)\) contains \(z_{0}\) as the unique root of \(p\) in its interior. (This is possible by the argument principle: the number of zeros in \(D(z_{0},\varepsilon)\) counted with multiplicity is \(\frac{1}{2\pi i}\oint_{|z-z_{0}|=\varepsilon}\frac{p^{\prime}(z)}{p(z)}\,dz\), which equals 1 near a simple root, so we can always find a small neighbourhood containing only that root.)

    On the boundary circle \(\Gamma\coloneqq\{z:|z-z_{0}|=\varepsilon\}\), the polynomial \(p(z)\) is continuous and non-zero (since \(z_{0}\) is the only root inside). Therefore, the minimum
    \begin{equation*}
        m \coloneqq \min_{|z-z_{0}|=\varepsilon}|p(z)| > 0.
    \end{equation*}
    This ensures a positive lower bound on \(|p|\) when restricted to the circle.

    Since the coefficients satisfy \(A_{n}\to A_{0}\) and \(B_{n}\to B_{0}\), the polynomials \(p_{n}(z)=z^{3}-B_{n} z^{2}+(A_{n}-1)z-B_{n}\) converge pointwise to \(p(z)\). On the compact set \(\Gamma\), this convergence is uniform:
    \begin{equation*}
        \sup_{z \in \Gamma}|p_{n}(z)-p(z)| \to 0 \quad \text{as } n\to\infty.
    \end{equation*}
    Therefore, for \(n\) sufficiently large (say \(n \geq N(\varepsilon)\)),
    \begin{equation*}
        \sup_{z \in \Gamma}|p_{n}(z)-p(z)| < m.
    \end{equation*}

    Rouché's theorem \citep{ahlforsComplexAnalysisIntroduction2007} states: If \(f\) and \(g\) are holomorphic inside and on a simple closed contour \(\Gamma\), and if \(|f(z)|>|g(z)|\) on \(\Gamma\), then \(f\) and \(f+g\) have the same number of zeros (counted with multiplicity) inside \(\Gamma\).

    In our setting, set \(f(z)=p(z)\) and \(g(z)=p_{n}(z)-p(z)\). The condition \(|p_{n}(z)-p(z)|<|p(z)|\) on \(\Gamma\) satisfies the hypothesis of Rouché's theorem. Thus, \(p(z)=f(z)\) and \(p_{n}(z)=f(z)+g(z)\) have the same number of zeros inside \(D(z_{0},\varepsilon)\). Since \(p\) has exactly one zero (the simple root \(z_{0}\)) in this disk, \(p_{n}\) also has exactly one zero, say \(z_{n}\), in \(D(z_{0},\varepsilon)\).

    For each \(n\geq N(\varepsilon)\), the root \(z_{n}\) lies in the compact disk \(\overline{D}(z_{0},\varepsilon)\). Therefore, any subsequence of \(\{z_{n}\}\) has a convergent sub-subsequence. If \(z_{n_k}\to z^{*}\) along a sub-sequence, then by continuity of \(p_{n}\), we have \(0=\lim_{k\to\infty}p_{n_k}(z_{n_k})=p_{0}(z^{*})=p(z^{*})\), so \(z^{*}\) is a root of \(p\) in \(\overline{D}(z_{0},\varepsilon)\). By uniqueness of the root \(z_{0}\) in this region, \(z^{*}=z_{0}\). Since any convergent sub-subsequence of \(\{z_{n}\}\) converges to \(z_{0}\), we conclude \(z_{n}\to z_{0}\) as \(n\to\infty\).
\end{proof}

\begin{lemma}\label{lem:local_linear_bias_variance}
    Under Assumptions~(A1)--(A6), for \(u_{0}=-1+ch\) with fixed \(c\geq 0\) (including \(c\to\infty\), the interior case),
    \begin{equation}\label{eq:local_linear_bias}
        \mathbb{E}\big[B_{n}^{\mathrm{CL}}(u_{0})\big]-\rho(u_{0})=\frac{h^{2}}{2}\rho^{\prime\prime}(u_{0})\,\mu_{2}(K^{*},c)+o(h^{2}),
    \end{equation}
    \begin{equation}\label{eq:local_linear_variance}
        \mathrm{Var}\big(B_{n}^{\mathrm{CL}}(u_{0})\big)=\frac{1+\rho(u_{0})^{2}}{nh\,f_{U}(u_{0})}\,R(K^{*},c)+o\!\left(\frac{1}{nh}\right),
    \end{equation}
    where
    \begin{equation}\label{eq:boundary_kernel_constants}
        \mu_{2}(K^{*},c)=\frac{S_{2}^{2}(c)-S_{1}(c)S_{3}(c)}{S_{0}(c)S_{2}(c)-S_{1}(c)^{2}},\qquad
        R(K^{*},c)=\int_{-c}^{1} K^{*}(t;c)^{2}\,\mathrm{d}t,
    \end{equation}
    with
    \begin{equation}\label{eq:boundary_kernel_roughness}
        R(K^{*},c)=\frac{S_{2}(c)^{2}R_{0}(c)-2S_{1}(c)S_{2}(c)R_{1}(c)+S_{1}(c)^{2}R_{2}(c)}{\big(S_{0}(c)S_{2}(c)-S_{1}(c)^{2}\big)^{2}},\qquad
        R_{j}(K,c)=\int_{-c}^{1} t^{j} K(t)^{2}\,\mathrm{d}t.
    \end{equation}
\end{lemma}

\begin{proof}
    This follows from the standard theory of local polynomial regression~\citep{fanDesignadaptiveNonparametricRegression1992,fanLocalPolynomialModelling1996}, adapted to the truncated support. We sketch the argument. Expand \(\rho(u_{0}+ht)=\rho(u_{0})+ht\rho^{\prime}(u_{0})+\frac{h^{2}t^{2}}{2}\rho^{\prime\prime}(u_{0})+o(h^{2})\) and substitute into the weighted least squares normal equations derived from~\eqref{eq:local_linear_wls}. The key observation is that the linear term \(ht\rho^{\prime}(u_{0})\) is eliminated by the zero-first-moment condition in~\eqref{eq:boundary_kernel_moments}, leaving only the \(O(h^{2})\) curvature term. This produces~\eqref{eq:local_linear_bias}. The variance expansion~\eqref{eq:local_linear_variance} follows from
    \begin{equation*}
        \mathrm{Var}(\hat{\rho}_{0})=\frac{1}{n^{2}}\sum_{i=1}^{n} \ell_{i}(u_{0})^{2}\,\mathrm{Var}(x_{i1}x_{i2}\mid U_{i}),
    \end{equation*}
    where \(\mathrm{Var}(X_{1}X_{2}\mid U=u_{0})=1+\rho(u_{0})^{2}\) by the fourth-moment formula for bivariate normals. Standard kernel-regression arguments then yield~\eqref{eq:local_linear_variance}.
\end{proof}
\subsection{Proof}\label{subsec:proof}
\paragraph{The proof of~\Cref{thm:uniform_consistency_an_bn}}
\begin{proof}
    We establish uniform consistency in three steps: first, we analyse the population cubic equation, then we transfer the result to the empirical cubic via continuity of roots, and finally we identify the estimator with the unique real root.

    \textbf{Step 1: The population cubic has a unique simple real root.}
    Substituting the population values \(A(u_{0})=2\) and \(B(u_{0})=\rho(u_{0})=:\rho_{0}\) from~\eqref{eq:an_bn} into the cubic equation, we obtain
    \begin{equation}\label{eq:population_cubic_factorization}
        \rho^{3}-\rho_{0}\rho^{2}+\rho-\rho_{0}=\rho(\rho^{2}+1)-\rho_{0}(\rho^{2}+1)=(\rho-\rho_{0})(\rho^{2}+1).
    \end{equation}
    Hence, the population cubic factors exactly with roots \(\{\rho_{0}, i,-i\}\): a unique real root \(\rho_{0}=\rho(u_{0})\in(-1,1)\) by Assumption (A3), and a complex-conjugate pair at distance exactly \(1\) from the real axis, independent of \(u_{0}\). The root \(\rho_{0}\) is simple since
    \begin{equation*}
        \frac{\partial}{\partial\rho}\big[(\rho-\rho_{0})(\rho^{2}+1)\big]\bigg|_{\rho=\rho_{0}}=\rho_{0}^{2}+1\geq 1>0.
    \end{equation*}

    \textbf{Step 2: Transfer to the empirical cubic via root continuity.}
    By~\Cref{thm:uniform_consistency_an_bn}, \((A_{n}(u_{0}),B_{n}(u_{0}))\xrightarrow{p}(2,\rho_{0})\) uniformly over \(u_{0}\in\mathcal{U}_{\delta}\). Define the empirical cubic polynomial
    \begin{equation*}
        p_{n}(\rho;u_{0})=\rho^{3}-B_{n}(u_{0})\rho^{2}+\big(A_{n}(u_{0})-1\big)\rho-B_{n}(u_{0}),
    \end{equation*}
    and the population cubic \(p(\rho;u_{0})=\rho^{3}-\rho_{0}\rho^{2}+\rho-\rho_{0}\). We apply~\Cref{lem:continuity_polynomial_roots} at each of the three population roots:
    \begin{itemize}
        \item At \(z_{0}=\rho_{0}\) (real, simple): with probability tending to \(1\), \(p_{n}(\cdot;u_{0})\) has exactly one root \(z_{n}^{(1)}\) within any \(\varepsilon\)-ball of \(\rho_{0}\), and \(z_{n}^{(1)}\xrightarrow{p}\rho_{0}\). Moreover, \(z_{n}^{(1)}\) must be real: if it were not, its conjugate \(\overline{z_{n}^{(1)}}\) would also lie in the same \(\varepsilon\)-ball (since \(|\overline{z_{n}^{(1)}}-\rho_{0}|=|z_{n}^{(1)}-\rho_{0}|<\varepsilon\) as \(\rho_{0}\) is real), contradicting the uniqueness from Rouché's theorem unless \(z_{n}^{(1)}=\overline{z_{n}^{(1)}}\), i.e., \(z_{n}^{(1)}\) is real.
        \item At \(z_{0}=\pm i\) (simple, \(\mathrm{Im}(z_{0})=\pm1\)): analogously, \(p_{n}(\cdot;u_{0})\) has roots \(z_{n}^{(2)},z_{n}^{(3)}\xrightarrow{p}\pm i\) in probability. For \(\varepsilon<1/2\), these roots satisfy \(|\mathrm{Im}(z_{n}^{(2,3)})|>1/2\) for \(n\) sufficiently large, hence are non-real.
    \end{itemize}
    Consequently, with probability tending to \(1\), \(p_{n}(\cdot;u_{0})=0\) has exactly one real root, and it converges in probability to \(\rho_{0}=\rho(u_{0})\).

    \textbf{Step 3: Identify \(\hat{\rho}(u_{0})\) with the unique real root.}
    By~\Cref{pro:finite-sample}, for every \(n\) there exists at least one real root in \([-1,1]\). Combined with Step 2, this root is (eventually, with high probability) the unique real root and lies in \((-1,1)\) since it is close to \(\rho_{0}\in(-1,1)\). The estimator \(\hat{\rho}(u_{0})\), defined as the root in \((-1,1)\) minimising \(Q_{n}(\rho;u_{0})\), therefore coincides with this root with high probability (the tie-breaking rule is inoperative asymptotically since there is only one candidate).

    Since the convergence \((A_{n},B_{n})\to(2,\rho_{0})\) is uniform over \(u_{0}\in\mathcal{U}_{\delta}\) by~\Cref{thm:uniform_consistency_an_bn}, and the root-separation radius in~\Cref{lem:continuity_polynomial_roots} can be chosen uniformly in \(u_{0}\) on \(\mathcal{U}_{\delta}\) (the gap between \(\rho_{0}\) and \(\pm i\) is \(\geq\sqrt{1+\eta^{2}}\) uniformly by Assumption (A3)), the consistency is uniform: \(\sup_{u_{0}\in\mathcal{U}_{\delta}}|\hat{\rho}(u_{0})-\rho(u_{0})|\xrightarrow{p}0\).
\end{proof}
\paragraph{The proof of~\Cref{thm:interior_asymptotic_normality}}\label{para:the_proof_of_interior_asymptotic_normality}
\begin{proof}
    The proof proceeds through the delta method applied to the implicit function defined by the cubic equation. Define \(g(\rho;A,B)=\rho^{3}-B\rho^{2}+(A-1)\rho-B\). At the true parameter point \((\rho_{0},2,\rho_{0})\), we have \(g(\rho_{0};2,\rho_{0})=0\) by~\eqref{eq:population_cubic_factorization}. Computing the partial derivatives at this point yields
    \begin{align*}
        g_\rho& \coloneqq\frac{\partial g}{\partial\rho}\bigg|_{(\rho_{0},2,\rho_{0})}=3\rho_{0}^{2}-2\rho_{0}\cdot\rho_{0}+(2-1)=\rho_{0}^{2}+1, \\
        g_{A} & \coloneqq\frac{\partial g}{\partial A}\bigg|_{(\rho_{0},2,\rho_{0})}=\rho_{0},                                                    \\
        g_{B} & \coloneqq\frac{\partial g}{\partial B}\bigg|_{(\rho_{0},2,\rho_{0})}=-(\rho_{0}^{2}+1).
    \end{align*}
    Since \(\hat{\rho}(u_{0})\) solves \(g(\hat{\rho}; A_{n}, B_{n})=0\) and \(\hat{\rho}(u_{0})\xrightarrow{p}\rho_{0}\) by~\Cref{thm:uniform_consistency}, with \(g_\rho=\rho_{0}^{2}+1\geq1>0\), the implicit function theorem yields a first-order Taylor expansion:
    \begin{equation}\label{eq:delta_method_expansion}
        0=g(\hat{\rho};A_{n},B_{n})\approx g_\rho(\hat{\rho}-\rho_{0})+g_{A}(A_{n}-2)+g_{B}(B_{n}-\rho_{0})+o_{p}\big(|A_{n}-2|+|B_{n}-\rho_{0}|\big).
    \end{equation}
    Rearranging and substituting the partial derivatives, we obtain
    \begin{equation}\label{eq:rho_linearization}
        \hat{\rho}(u_{0})-\rho_{0} = \frac{1}{\rho_{0}^{2}+1}\Big[-\rho_{0}(A_{n}(u_{0})-2)+(\rho_{0}^{2}+1)(B_{n}(u_{0})-\rho_{0})\Big]+o_{p}\big(|A_{n}-2|+|B_{n}-\rho_{0}|\big).
    \end{equation}

    To establish the joint asymptotic distribution of \((A_{n}, B_{n})\), we apply the Isserlis (Wick) moment formulas for bivariate normal random variables. Let \(Y_{A}=X_{1}^{2}+X_{2}^{2}\) and \(Y_{B}=X_{1}X_{2}\). Under the bivariate normal distribution with \(\mathbb{E}[X_{1}]=\mathbb{E}[X_{2}]=0\), \(\mathrm{Var}(X_{1})=\mathrm{Var}(X_{2})=1\), and \(\mathrm{Cov}(X_{1},X_{2})=\rho_{0}\), we have
    \begin{align*}
        \mathrm{Var}(Y_{A})      & =\mathbb{E}\left[ (X_{1}^{2}+X_{2}^{2})^{2} \right]-4=\mathbb{E}[X_{1}^{4}]+\mathbb{E}[X_{2}^{4}]+2\mathbb{E}[X_{1}^{2}X_{2}^{2}]-4 \\
        & =3+3+2(1+2\rho_{0}^{2})-4=4+4\rho_{0}^{2},                                                                                          \\
        \mathrm{Var}(Y_{B})      & =\mathbb{E}[X_{1}^{2}X_{2}^{2}]-\rho_{0}^{2}=1+2\rho_{0}^{2}-\rho_{0}^{2}=1+\rho_{0}^{2},                                           \\
        \mathrm{Cov}(Y_{A},Y_{B})& =\mathbb{E}[(X_{1}^{2}+X_{2}^{2})X_{1}X_{2}]-2\rho_{0}=\mathbb{E}[X_{1}^{3}X_{2}]+\mathbb{E}[X_{1}X_{2}^{3}]-2\rho_{0}=4\rho_{0}.
    \end{align*}
    Thus, the covariance matrix is
    \begin{equation}\label{eq:sigma_matrix}
        \Sigma(u_{0})=
        \begin{pmatrix}4+4\rho_{0}^{2} & 4\rho_{0}\\ 4\rho_{0} & 1+\rho_{0}^{2}
        \end{pmatrix}.
    \end{equation}

    By the multivariate extension of~\Cref{lem:bias-variance expansion} and the central limit theorem for kernel-weighted sums, we have
    \begin{equation}\label{eq:joint_clt}
        \sqrt{nh}\left[
            \begin{pmatrix}A_{n}(u_{0})-2\\ B_{n}(u_{0})-\rho_{0}
        \end{pmatrix}-\mathrm{Bias}\right]\xrightarrow{d} N\!\left(0,\ \frac{R(K)}{f_{U}(u_{0})}\Sigma(u_{0})\right),
    \end{equation}
    where the bias vector is
    \begin{equation}\label{eq:bias_vector}
        \mathrm{Bias}=
        \begin{pmatrix}o(h^{2})\\[4pt] \frac{h^{2}\mu_{2}(K)}{2}\Big[\rho^{\prime\prime}(u_{0})+2\frac{f_{U}^{\prime}(u_{0})}{f_{U}(u_{0})}\rho^{\prime}(u_{0})\Big]
        \end{pmatrix}.
    \end{equation}
    The first component is \(o(h^{2})\) because \(A(u)\equiv 2\) identically, so \(A^{\prime}(u)=A^{\prime\prime}(u)=0\) for all \(u\).

    Applying the delta method to~\eqref{eq:rho_linearization} with the linear transformation vector \(c^{\top}=\big(-\rho_{0}/(1+\rho_{0}^{2}),\,1\big)\), we compute the asymptotic variance:
    \begin{align*}
        c^{\top}\Sigma(u_{0})c& =\frac{\rho_{0}^{2}}{(1+\rho_{0}^{2})^{2}}(4+4\rho_{0}^{2})-2\cdot\frac{\rho_{0}}{1+\rho_{0}^{2}}\cdot 4\rho_{0}+(1+\rho_{0}^{2}) \\
        & =\frac{4\rho_{0}^{2}(1+\rho_{0}^{2})}{(1+\rho_{0}^{2})^{2}}-\frac{8\rho_{0}^{2}}{1+\rho_{0}^{2}}+(1+\rho_{0}^{2})                 \\
        & =\frac{4\rho_{0}^{2}-8\rho_{0}^{2}+(1+\rho_{0}^{2})^{2}}{1+\rho_{0}^{2}}=\frac{(1-\rho_{0}^{2})^{2}}{1+\rho_{0}^{2}}.
    \end{align*}
    Combining this with~\eqref{eq:joint_clt} and~\eqref{eq:rho_linearization} establishes~\eqref{eq:interior_limit} and~\eqref{eq:interior_bias}.
\end{proof}
\paragraph{The proof of~\Cref{thm:boundary_asymptotic_normality}}

\begin{proof}
    At a left-boundary point \(u_{0}=-1+ch\), the kernel window \(K_{h}(U_{i}-u_{0})\) intersects with the design support \([-1,1]\) in a truncated manner. Specifically, the standardized variable \(t=(U_{i}-u_{0})/h\) satisfies \(t\in[-c,1]\) rather than the full support \([-1,1]\) of the kernel. This truncation breaks the symmetry assumption \(\mu_{1}(K)=0\) used in the interior case.

    The proof of~\Cref{lem:bias-variance expansion} extends to this setting with the truncated moments defined in~\eqref{eq:truncated_moments}. The key observation is that \(A_{n}(u_{0})\) still has negligible bias because the population function \(A(u)\equiv 2\) is constant, so \(A^{\prime}(u)=A^{\prime\prime}(u)=0\) regardless of kernel truncation. Therefore, all bias terms originate from \(B_{n}(u_{0})\) estimating the varying function \(\rho(\cdot)\).

    For the bias of \(B_{n}(u_{0})\), the first-order Taylor expansion gives
    \begin{equation*}
        \mathbb{E}[B_{n}(u_{0})]-\rho(u_{0})=h\rho^{\prime}(u_{0})\frac{\mu_{1}(K,c)}{\mu_{0}(K,c)}+o(h),
    \end{equation*}
    where the leading term is \(O(h)\) because \(\mu_{1}(K,c)\neq0\) in general. This contrasts with the interior case where symmetry yields \(\mu_{1}(K)=0\) and thus \(O(h^{2})\) bias.

    For the variance, the same Isserlis moment calculations yield the covariance matrix~\eqref{eq:sigma_matrix}, and the central limit theorem holds with the truncated kernel integrals:
    \begin{equation*}
        \sqrt{nh}\left[
            \begin{pmatrix}A_{n}(u_{0})-2\\ B_{n}(u_{0})-\rho_{0}
        \end{pmatrix}-\mathrm{Bias}\right]\xrightarrow{d} N\!\left(0,\ \frac{R(K,c)}{\mu_{0}(K,c)^{2}f_{U}(-1^+)}\Sigma(u_{0})\right).
    \end{equation*}
    Applying the same delta-method linearization as in the proof of~\Cref{thm:interior_asymptotic_normality} yields~\eqref{eq:boundary_limit}.
\end{proof}
\paragraph{The proof of~\Cref{cor:optimal_bandwidth}}\label{para:the_proof_of}
\begin{proof}
    For interior points, the asymptotic mean squared error is the sum of squared bias and variance from~\Cref{thm:interior_asymptotic_normality}:
    \begin{equation*}
        \mathrm{AMSE}(h)=h^{4}C_{1}(u_{0})+\frac{C_{2}(u_{0})}{nh}.
    \end{equation*}
    Differentiating with respect to \(h\) and setting to zero yields
    \begin{equation*}
        \frac{d}{dh}\mathrm{AMSE}(h)=4C_{1}(u_{0})h^{3}-\frac{C_{2}(u_{0})}{nh^{2}}=0,
    \end{equation*}
    which gives \(h^{5}=C_{2}(u_{0})/(4C_{1}(u_{0})n)\). At this bandwidth,
    \begin{align*}
        \mathrm{AMSE}(h_{\mathrm{opt}})& =C_{1}(u_{0})\left(\frac{C_{2}(u_{0})}{4C_{1}(u_{0})n}\right)^{4/5}+C_{2}(u_{0})\left(\frac{4C_{1}(u_{0})n}{C_{2}(u_{0})}\right)^{1/5}\frac{1}{n} \\
        & =\left(\frac{5}{4}\right)C_{1}(u_{0})^{1/5}C_{2}(u_{0})^{4/5}n^{-4/5}\asymp n^{-4/5}.
    \end{align*}

    For boundary points, the asymptotic MSE is
    \begin{equation*}
        \mathrm{AMSE}(h)=h^{2}C_{3}+\frac{C_{4}}{nh}.
    \end{equation*}
    Setting \(\frac{d}{dh}\mathrm{AMSE}(h)=2C_{3}h-C_{4}/(nh^{2})=0\) gives \(h^{3}=C_{4}/(2C_{3}n)\), and similarly \(\mathrm{AMSE}(h_{\mathrm{opt}})\asymp n^{-2/3}\).
\end{proof}
\subsection{Local Linear Estimation of Correlation Coefficient}\label{subsec:local_linear_estimation_of_correlation_coefficient}
Recall the observations \((U_{i}, x_{i1}, x_{i2})\), \(i=1,\ldots,n\) in the main text; the kernel weighted likelihood  function at  \( u_{0} \) is defined as:
\begin{equation*}
    Q_{n}(\rho;u_{0})=\frac{1}{n}\sum_{i=1}^{n} K_{h}(U_{i}-u_{0})\,\left[\frac{x_{i1}^{2}+x_{i2}^{2}-2\rho x_{i1}x_{i2}}{1-\rho^{2}}+\log(1-\rho^{2})\right].
\end{equation*}
We assume that the correlation coefficient is a smooth function of \( u \) and can be locally approximated by a linear function around \( u_{0} \), i.e., \( \rho_{i}(\boldsymbol{\beta})\sim \beta_{0}+\beta_{1}(U_{i}-u_{0}) \), where \( \boldsymbol{\beta}=(\beta_{0},\beta_{1})^{\top} \). Here, \( \beta_{0} \) represents the actual correlation coefficient at \( u_{0} \) and \( \beta_{1} \) represents the slope of the correlation coefficient at \( u_{0} \). Substituting \( \rho_{i}(\boldsymbol{\beta}) \) into the kernel weighted likelihood function, we have
\begin{equation*}
    Q_{n}(\boldsymbol{\beta};u_{0})=\frac{1}{n}\sum_{i=1}^{n} K_{h}(U_{i}-u_{0})\,\left[\frac{x_{i1}^{2}+x_{i2}^{2}-2(\beta_{0}+\beta_{1}(U_{i}-u_{0})) x_{i1}x_{i2}}{1-(\beta_{0}+\beta_{1}(U_{i}-u_{0}))^{2}}+\log(1-(\beta_{0}+\beta_{1}(U_{i}-u_{0}))^{2})\right].
\end{equation*}
It is difficult to optimise \( Q_{n}(\boldsymbol{\beta};u_{0}) \) directly as the correlation coefficient is constrained in the interval \( (-1,1) \). To prevent numerical collapse, we apply Fisher's \( z \)-transformation to the correlation coefficient, i.e., let \( \gamma_{i}=\alpha_{0}+\alpha_{1}(U_{i}-u_{0}) \)  and \( \rho_{i}(\boldsymbol{\beta})=\tanh(\gamma_{i}) \). By the identities \( 1-\tanh^{2}(x)=\text{sech}^{2}(x) \), \( \text{sech}(x)=1/\cosh(x)\), the safely bounded, globally stable objective function becomes:
\begin{equation*}
    Q_{n}(\alpha_{0},\alpha_{1};u_{0})=\frac{1}{n}\sum_{i=1}^{n} K_{h}(U_{i}-u_{0})\,\left[\cosh^{2}(\gamma_{i})(x_{i1}^{2}+x_{i2}^{2})-\sinh(2\gamma_{i}) x_{i1}x_{i2}-2\log(\cosh(\gamma_{i}))\right].
\end{equation*}
By minimising \( Q_{n}(\alpha_{0},\alpha_{1};u_{0}) \) with respect to \( (\alpha_{0},\alpha_{1}) \), we obtain the local linear estimator of the correlation coefficient at \( u_{0} \): \( \hat{\rho}^{\mathrm{LL}}(u_{0})=\tanh(\hat{\alpha}_{0}) \).

\begin{figure}[ht]
    \centering
    \includegraphics[width=0.95\textwidth]{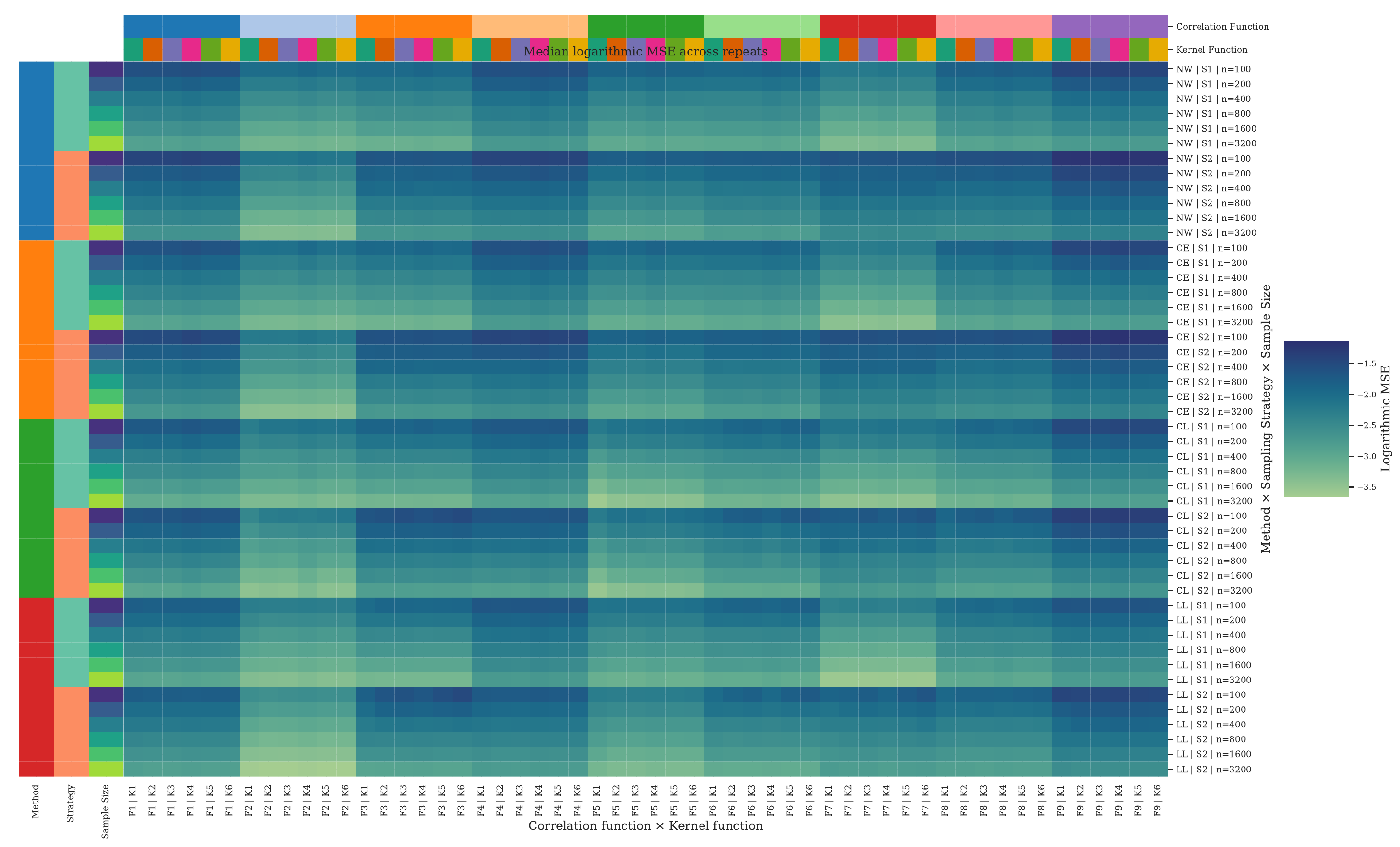}
    \caption{Heatmap of the median logarithmic MSEs across different correlation functions and kernel functions for each combination of estimation method, sampling strategy, and sample size.}\label{fig:figs/panel_mse_heatmap.pdf}
\end{figure}
\begin{figure}[ht]
    \centering
    \includegraphics[width=\textwidth]{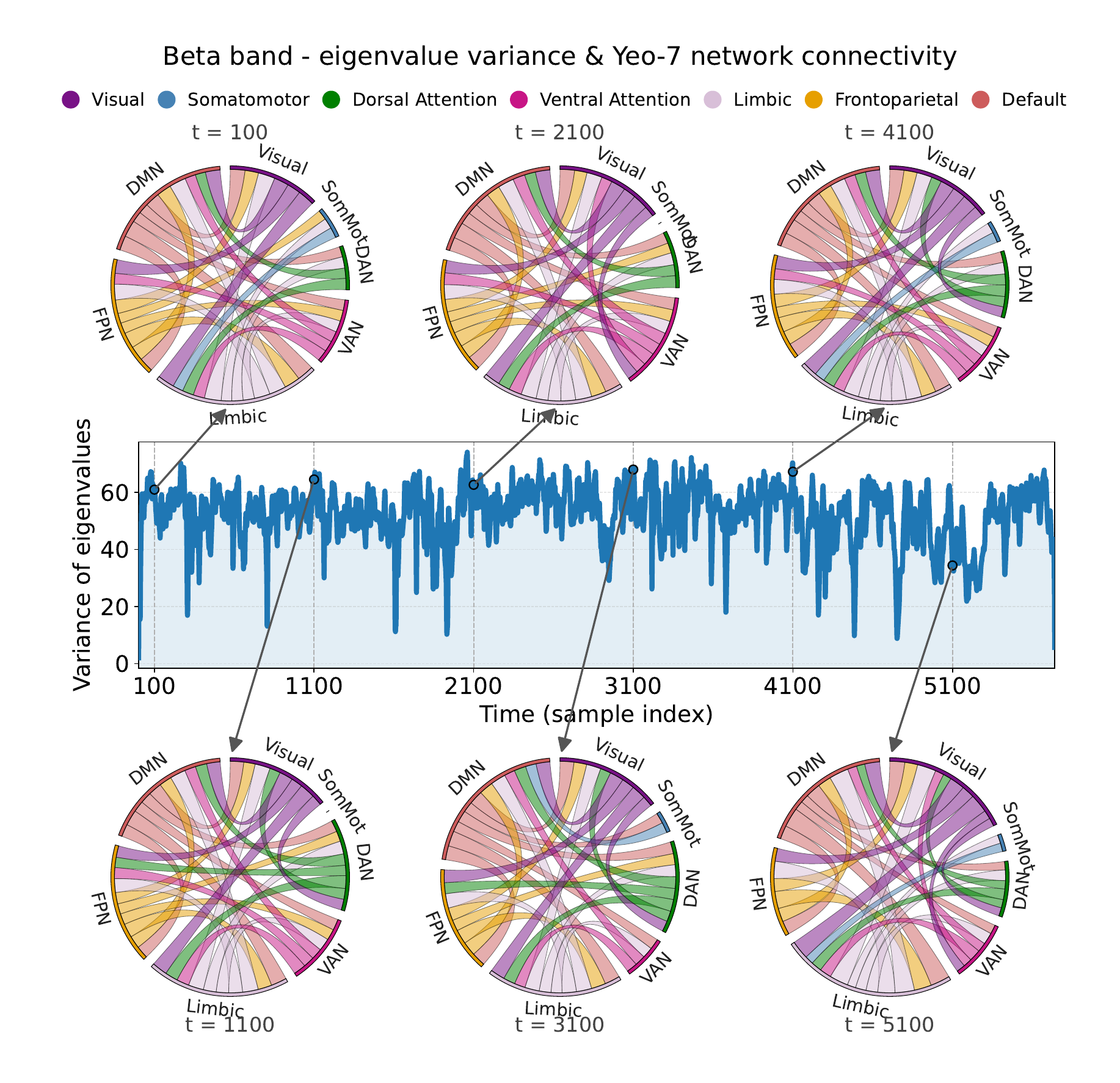}
    \caption{The eigenvalue variance of the dynamic correlation matrix for the Beta band. At the time (sample index) 100,1100,2100,3100,4100,5100, we draw the chord connectivity of Yeo-7 networks. As the syncitical \( 7\times7 \) network is still dense, we truncated the Yeo-7 networks by its median value. The names of Yeo-7 networks are ``Visual'',``Somatomotor'',``Dorsal Attention'',``Ventral Attention'',``Limbic'',``Frontoparietal'' and ``Default'', while the short names on the chord connectivity are ``Visual'',``SomMot'',``DAN'',``VAN'',``Limbic'',``FPN'' and ``DMN'' respectively.}\label{fig:figs/eig_var_circle_connect_Beta.pdf}
\end{figure}
\begin{figure}[ht]
    \centering
    \includegraphics[width=\textwidth]{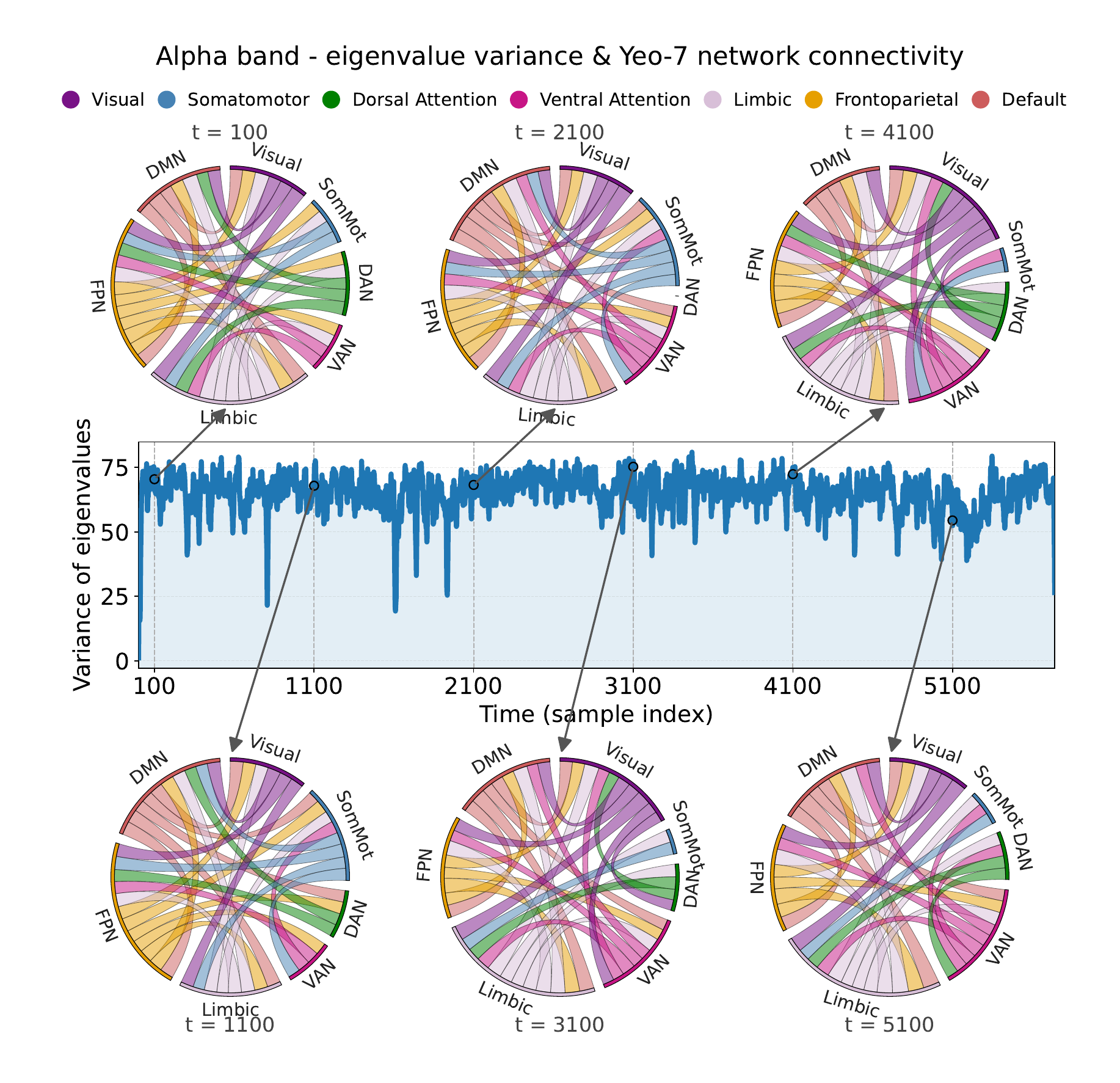}
    \caption{The eigenvalue variance of the dynamic correlation matrix for the Alpha band. At the time (sample index) 100,1100,2100,3100,4100,5100, we draw the chord connectivity of Yeo-7 networks. As the syncitical \( 7\times7 \) network is still dense, we truncated the Yeo-7 networks by its median value. The names of Yeo-7 networks are ``Visual'',``Somatomotor'',``Dorsal Attention'',``Ventral Attention'',``Limbic'',``Frontoparietal'' and ``Default'', while the short names on the chord connectivity are ``Visual'',``SomMot'',``DAN'',``VAN'',``Limbic'',``FPN'' and ``DMN'' respectively.}\label{fig:figs/eig_var_circle_connect_Alpha.pdf}
\end{figure}
\begin{figure}[ht]
    \centering
    \includegraphics[width=\textwidth]{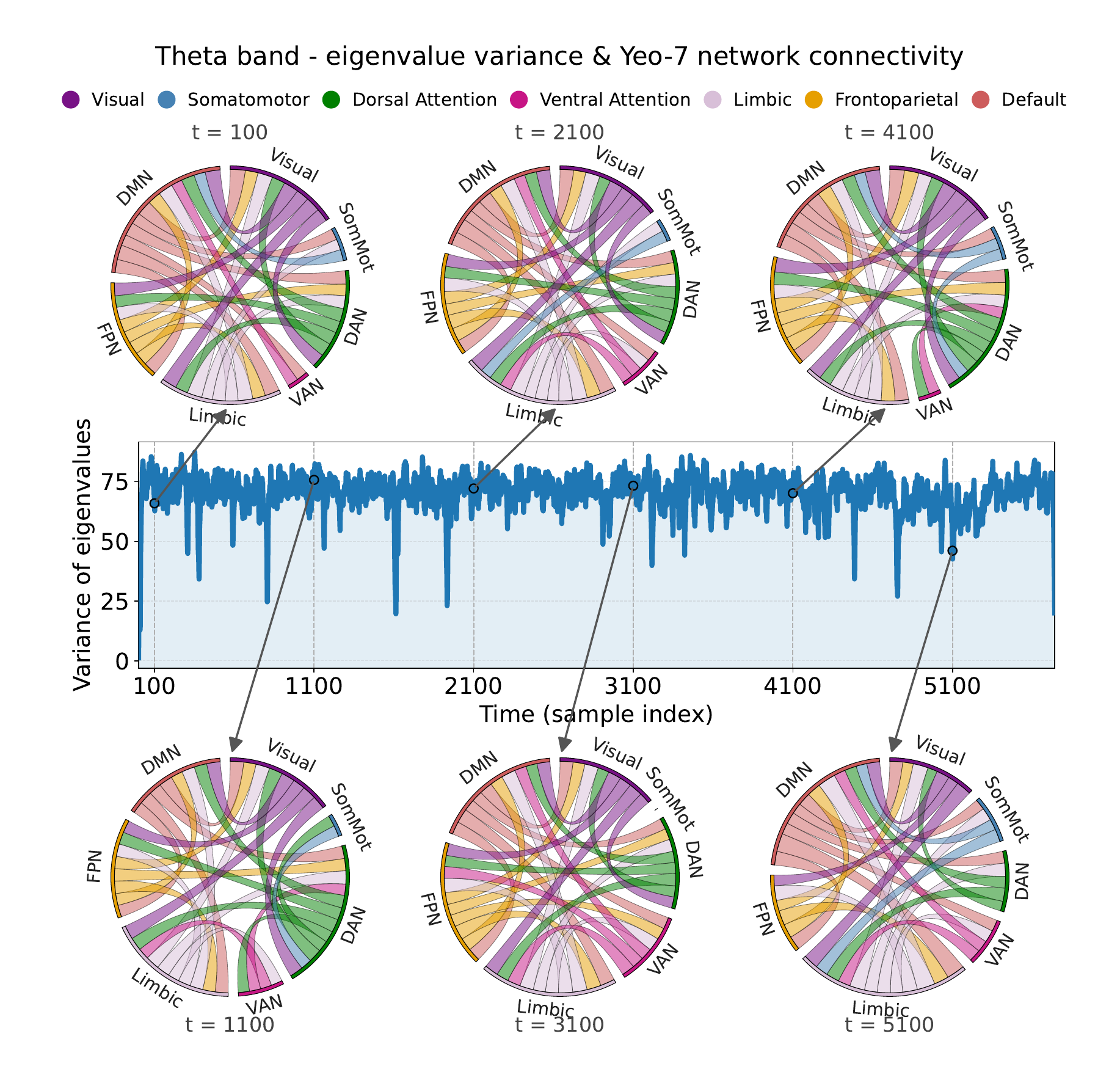}
    \caption{The eigenvalue variance of the dynamic correlation matrix for the Theta band. At the time (sample index) 100,1100,2100,3100,4100,5100, we draw the chord connectivity of Yeo-7 networks. As the syncitical \( 7\times7 \) network is still dense, we truncated the Yeo-7 networks by its median value. The names of Yeo-7 networks are ``Visual'',``Somatomotor'',``Dorsal Attention'',``Ventral Attention'',``Limbic'',``Frontoparietal'' and ``Default'', while the short names on the chord connectivity are ``Visual'',``SomMot'',``DAN'',``VAN'',``Limbic'',``FPN'' and ``DMN'' respectively.}\label{fig:figs/eig_var_circle_connect_Theta.pdf}
\end{figure}
\begin{figure}[ht]
    \centering
    \includegraphics[width=\textwidth]{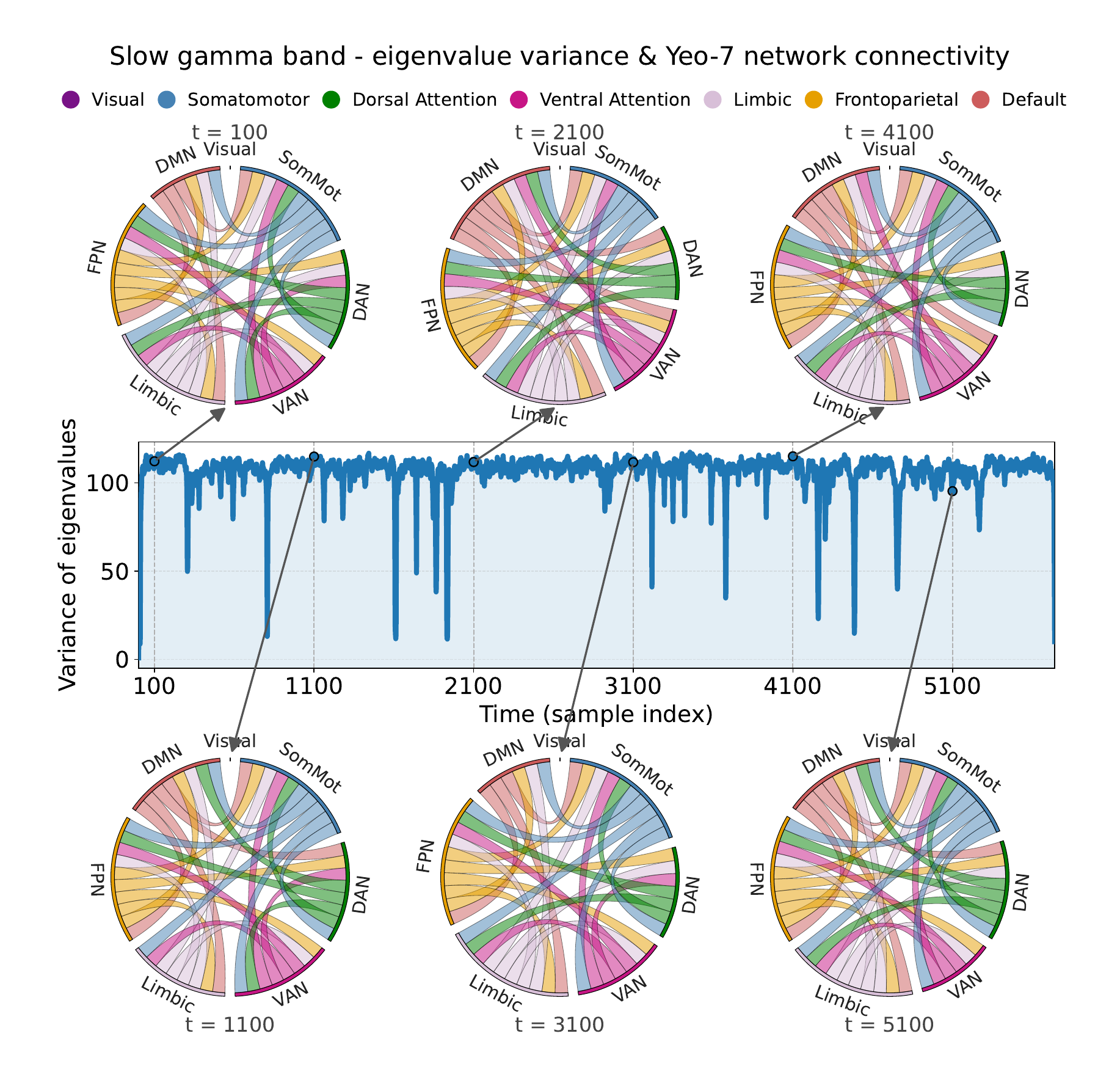}
    \caption{The eigenvalue variance of the dynamic correlation matrix for the slow Gamma band. At the time (sample index) 100,1100,2100,3100,4100,5100, we draw the chord connectivity of Yeo-7 networks. As the syncitical \( 7\times7 \) network is still dense, we truncated the Yeo-7 networks by its median value. The names of Yeo-7 networks are ``Visual'',``Somatomotor'',``Dorsal Attention'',``Ventral Attention'',``Limbic'',``Frontoparietal'' and ``Default'', while the short names on the chord connectivity are ``Visual'',``SomMot'',``DAN'',``VAN'',``Limbic'',``FPN'' and ``DMN'' respectively.}\label{fig:figs/eig_var_circle_connect_Slow_gamma.pdf}
\end{figure}

\end{document}